%% file: main.tex
\documentclass[12pt]{article}

\usepackage{agca}
\usepackage{bibunits}

\defaultbibliographystyle{chicago}
\defaultbibliography{articles}

\title{Anchored Geodesic Analysis for Multivariate Extremes}
\author{Alberto Quaini\thanks{Corresponding author: quaini@ese.eur.nl. We thank Nicola Gnecco for the valuable comments.} \qquad Chen Zhou\\
Erasmus University Rotterdam}
\date{}

\begin{document}

\begin{bibunit}
\maketitle
\enlargethispage{.75\baselineskip}

\begin{abstract}
Extremal dependence is naturally described by the angular law of large multivariate observations.
We introduce anchored geodesic component analysis (AGCA), a dimension-reduction method for
extremal angular laws on the positive unit sphere. AGCA approximates angular variation by great
subspheres constrained to pass through a chosen reference direction, with balanced complete
dependence as the default anchor. Under a bounded sine-squared geodesic loss, the population and
empirical problems reduce exactly to eigenanalysis of a second-moment matrix of anchored tangent
departures. The resulting scores, loadings, residual risks and explained-variation summaries
describe departures from the benchmark and remain well defined for face and near-axis extremes.
Low-rank AGCA reconstructions also support tail simulation: bounded Lipschitz functionals and
homogeneous tail scores, including portfolio capped excesses and value-at-risk, inherit explicit
error bounds from the AGCA residual risk. We establish top-\(k\) consistency for oracle and
rank-Pareto AGCA summaries and an oracle central limit theorem whose covariance is that of an
independent sample from the limiting angular law. In daily equity-portfolio losses, AGCA finds
concentrated benchmark-relative tail directions: ten components explain about \(91\%\) of anchored
variation and approximate capped-excess and normalized value-at-risk summaries with about
\(1.25\%\) average relative error.
\end{abstract}



\subfile{01_introduction}
\subfile{02_population_theory}
\subfile{03_estimation}
\subfile{04_empirics_portfolios}
\subfile{05_discussion}

\putbib

\clearpage

\end{bibunit}

\begin{bibunit}

\begin{center}
{\Large Supplementary Material for\\[0.5em]
Anchored Geodesic Analysis for Multivariate Extremes}
\end{center}
\bigskip

\setcounter{section}{0}
\setcounter{subsection}{0}
\setcounter{figure}{0}
\setcounter{table}{0}
\setcounter{equation}{0}
\setcounter{assumption}{0}
\setcounter{theorem}{0}
\setcounter{proposition}{0}
\setcounter{lemma}{0}
\setcounter{corollary}{0}
\setcounter{definition}{0}
\setcounter{remark}{0}
\renewcommand{\thesection}{S\arabic{section}}
\renewcommand{\thefigure}{S\arabic{figure}}
\renewcommand{\thetable}{S\arabic{table}}
\renewcommand{\theequation}{S\arabic{equation}}
\renewcommand{\theassumption}{S\arabic{assumption}}
\renewcommand{\thetheorem}{S\arabic{theorem}}
\renewcommand{\theproposition}{S\arabic{proposition}}
\renewcommand{\thelemma}{S\arabic{lemma}}
\renewcommand{\thecorollary}{S\arabic{corollary}}
\renewcommand{\thedefinition}{S\arabic{definition}}
\renewcommand{\theremark}{S\arabic{remark}}
\renewcommand{\theHsection}{supp.section.\arabic{section}}
\renewcommand{\theHsubsection}{supp.subsection.\arabic{section}.\arabic{subsection}}
\renewcommand{\theHfigure}{supp.figure.\arabic{figure}}
\renewcommand{\theHtable}{supp.table.\arabic{table}}
\renewcommand{\theHequation}{supp.equation.\arabic{equation}}
\renewcommand{\theHassumption}{supp.assumption.\arabic{assumption}}
\renewcommand{\theHtheorem}{supp.theorem.\arabic{theorem}}
\renewcommand{\theHproposition}{supp.proposition.\arabic{proposition}}
\renewcommand{\theHlemma}{supp.lemma.\arabic{lemma}}
\renewcommand{\theHcorollary}{supp.corollary.\arabic{corollary}}
\renewcommand{\theHdefinition}{supp.definition.\arabic{definition}}
\renewcommand{\theHremark}{supp.remark.\arabic{remark}}

\subfile{appendix}
\subfile{appendix_simulations}
\subfile{appendix_empirics_portfolios}

\putbib

\end{bibunit}

\end{document}

%% file: 01_introduction.tex
\section{Introduction}
\label{sec:introduction}

Extremes increasingly arise as multivariate events. A storm is a joint profile over a region, a
market crisis is a configuration of simultaneous tail losses, and an infrastructure failure may
be visible only through several stressed measurements. For such data, the central question is
extremal dependence: how do coordinates participate in large events, and can the relative pattern
of participation be summarized in a small number of interpretable coordinates?

Classical multivariate extreme value theory formalizes this question through angular measures.
After standardizing margins to a common tail scale, a large observation decomposes into radial
size and angular direction; the limiting angular law records the relative contribution of the
coordinates in the tail \citep{de2006extreme,resnick2007heavy}. This representation underlies
peaks-over-threshold models and multivariate generalized Pareto limits
\citep{rootzen2006multivariate,rootzen2018multivariate,mourahib2025multivariate}. It is also a
natural starting point for dimension reduction, because the angular law is the object that
remains after radial severity has been separated from extremal direction.

This paper develops anchored geodesic component analysis (AGCA) for that purpose. We represent
standardized extreme directions on the Euclidean unit sphere, which preserves extremal rays in
finite dimension while placing the angular component on a familiar Riemannian manifold. The
default reference direction is the balanced complete-dependence anchor \(\mu_0\) in
\eqref{eq:canonical-anchor}, corresponding to equal participation of all variables after
marginal standardization. AGCA fits great subspheres through this anchor and measures
reconstruction by a bounded sine-squared geodesic loss. The resulting scores, loadings, residual
risks, and explained variation describe angular departures from balanced extremal
co-movement.

The construction is related to principal geodesic analysis and principal nested spheres, which
adapt principal component ideas to nonlinear spaces
\citep[e.g.,][]{fletcher2004principal,jung2012analysis,sommer2014optimization,
tabaghi2024principal}.
Those methods typically build geodesic summaries
around data-adaptive base points and use squared geodesic reconstruction loss. AGCA instead uses
a reference direction tailored to extremal dependence and a bounded sine-squared geodesic loss.
In contrast to principal nested spheres, which allow nested small subspheres rather than only
great subspheres \citep{jung2012analysis}, AGCA deliberately restricts to great subspheres
through the anchor. The restriction makes the bounded sine-squared loss reduce spherical
projection to ordinary orthogonal projection of bounded departures in the tangent hyperplane.
Consequently, the population and empirical AGCA problems reduce exactly to eigendecomposition of
an anchored second-moment matrix. In addition, since all standardized extremal rays are
represented on the sphere, near-axis directions associated with asymptotic independence remain
finite angular observations.

This exact spectral form is useful both computationally and statistically. It gives one
eigensolution for the fitted component space and for diagnostics such as anchored variation
explained, residual geodesic risk, scores, loadings, and component-space stability. It also makes
perturbation analysis transparent: error in the anchored second-moment matrix translates into
eigenvalue error, residual-risk error, and, under an eigengap, projector error for the fitted
component space. The loadings retain a direct extremal interpretation. Near \(\mu_0\), loading
differences determine how pairwise angular participation changes relative to balanced complete
dependence.

Several neighboring literatures reduce the angular complexity, but they target different
objects. Support and sparsity methods ask which variables, faces, subcones, or representative
directions can become extreme together
\citep{chautru2015dimension,goix2016sparsity,meyer2021sparse}. Extremal graphical models encode
conditional independence and sparse Markov structure in tail limits
\citep{engelke2020graphical}. Tail dependence matrices and their decompositions summarize
pairwise extremal association, while necessarily losing higher-order angular information
\citep{cooley2019decompositions}. Clustering methods for angular extremes group directions into
representative regimes \citep{janssen2020kmeans,fomichov2023spherical}.
Principal component methods for extremes seek low-dimensional linear subspaces close to extreme
directions and study reconstruction risk, excess risk, and component selection
\citep{Drees2021principal,drees2025asymptotic,butsch2025estimation,cooley2026principal}.
Kernel, hyperplane, and compositional approaches provide further projections or transformed
linear coordinates for angular samples \citep{medina2025insights,wan2026characterizing,
lhaut2026simulation}. These methods are valuable when the target is an active face set, a sparse
tail graph or direction set, a collection of representative regimes, a pairwise matrix, or a
linear representation. They do not directly provide a benchmark-relative component analysis of
the angular law itself. Several of these approaches impose sparsity, discrete-regime,
or linear-subspace structure, whereas AGCA is designed to summarize the angular law itself
without assuming sparsity and remains well defined for face and near-axis extremes.

The paper contributes at population, estimation, and inferential levels. At the population level,
we define the anchored geodesic reconstruction problem, prove the exact projection identity,
derive the eigensolution and loading interpretation, and characterize face- and axis-supported
angular laws. Then we show how low-rank AGCA reconstructions can simulate bounded Lipschitz and
homogeneous tail functionals, including portfolio capped excesses and Value-at-Risk. At the
sample level, we give the rank-Pareto top-\(k\) construction used in applications, building on
empirical angular measure estimation as in
\citet{einmahl1997estimating,einmahl2001nonparametric,einmahl2009maximum}, and establish
consistency and stability of scalar and component-space AGCA summaries. We also prove an oracle
central limit theorem for the anchored top-\(k\) second-moment matrix. Under a second-order
angular-bias condition, random threshold selection has the same first-order limit as an
independent sample from the limiting angular law, yielding plug-in inference for the
AGCA-specific eigenvalues, anchored explained variation, residual geodesic risk, and component
projectors. Unlike the multivariate regular variation framework common to this line
\citep{Drees2021principal,drees2025asymptotic,butsch2025estimation}, the oracle AGCA CLT depends
on angular convergence and a second-order angular-bias condition rather than full radial regular
variation or tail-moment assumptions.

The empirical application studies daily equity-portfolio loss extremes. In a
\(d=24\) Fama--French panel \citep{kenFrenchDataLibrary}, the AGCA spectrum is concentrated: the
first five components explain \(76.3\%\) of anchored variation and the first ten explain
\(91.1\%\). The loadings identify repeated cross-sectional contrasts in extreme portfolio-loss
directions. The low-rank reconstruction also preserves the portfolio tail summaries motivated by
the theory: by rank ten, the average relative error across capped-excess and normalized
value-at-risk summaries is about \(1.25\%\) over a broad portfolio class. A second daily anomaly
portfolio panel from Open Source Asset Pricing \citep{openSourceAssetPricing} gives the same
qualitative message, with a less compressed but still stable anchored spectrum.

The Supplementary Material collects proofs, complementary theory, simulation
analyses, and complementary empirical results. The simulations illustrate AGCA
in designs that vary in complexity of the
angular law and the number of near-axis directions; see
Figure~\ref{fig:intro-sim-3d-sphere} for a simple visual preview.
The AGCA methodology is implemented in the R package \texttt{AGCA4extremes}
(\url{https://github.com/a91quaini/AGCA4extremes}). Data and replication code
for the simulation and empirical analyses are available at
\url{https://github.com/a91quaini/replicateAGCApaper}.

\begin{figure}[tbp]
\centering
\begin{minipage}[t]{0.45\linewidth}
\centering
\appinclude[width=\linewidth]{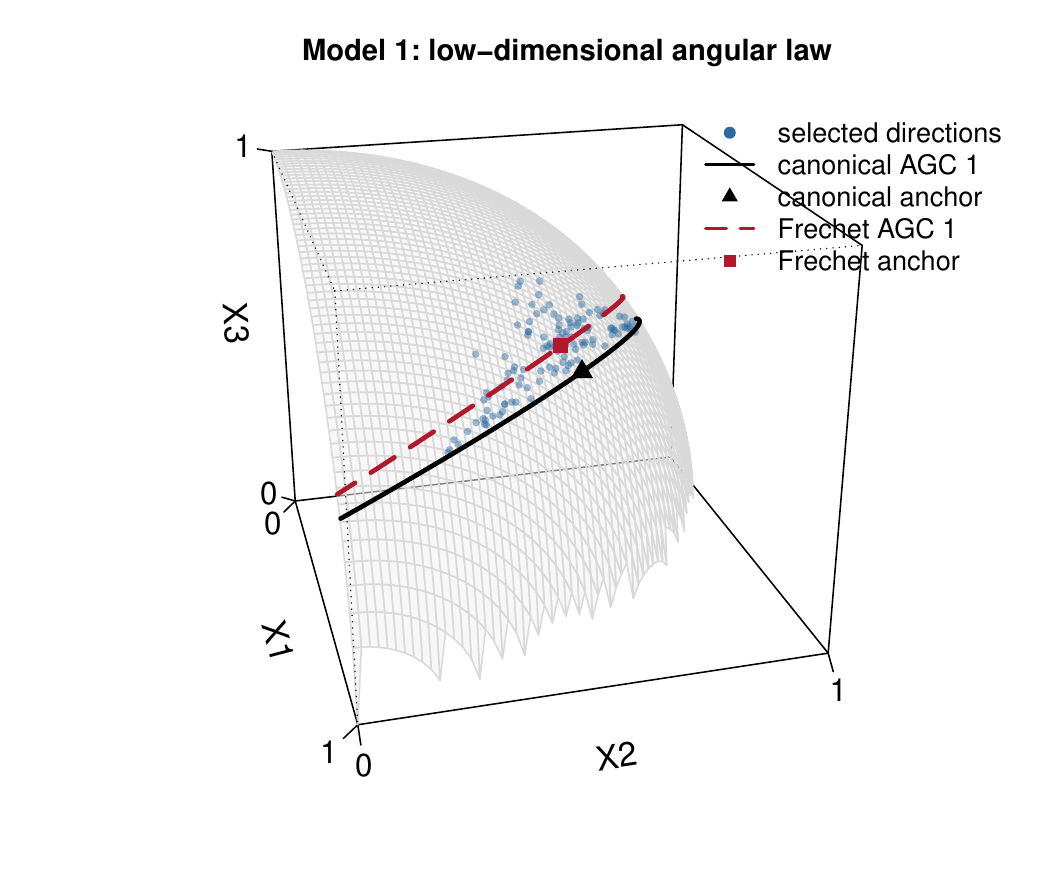}
\end{minipage}\hfill
\begin{minipage}[t]{0.45\linewidth}
\centering
\appinclude[width=\linewidth]{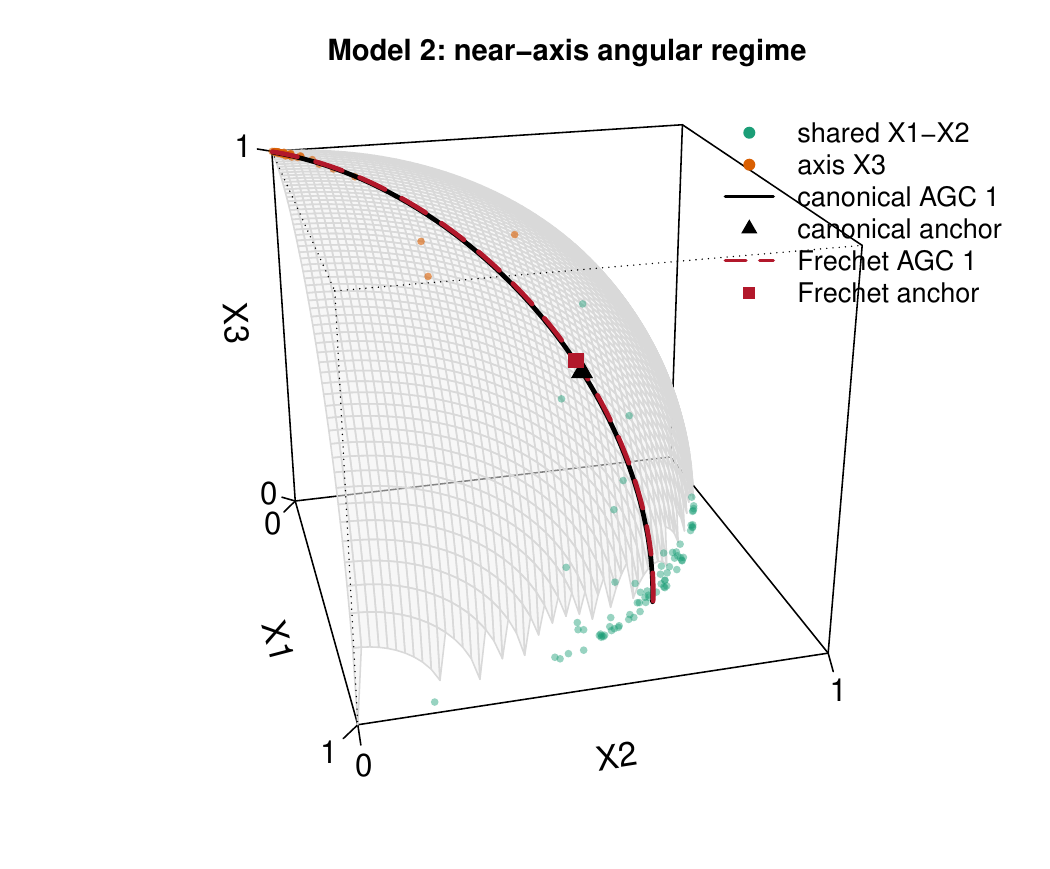}
\end{minipage}
\caption{Sphere-level geometry in two controlled simulations. The left panel shows a
low-dimensional angular law; the right panel shows a variable-specific near-axis regime. Each
panel displays selected extreme directions, the canonical anchor and first AGC, and the
Fr\'echet anchor and first AGC.}
\label{fig:intro-sim-3d-sphere}
\end{figure}

%% file: 02_population_theory.tex
\section{AGCA: population theory}
\label{sec:population-theory}
\label{sec:spherical-model}

We use the following notation throughout. For \(x,y\in\R^d\),
\(\ip{x}{y}=x^{\T}y\) and \(\norm{x}_2=\ip{x}{x}^{1/2}\).
Let \(\Sphere^{d-1}:=\{x\in\R^d:\norm{x}_2=1\}\), and let
\(\Sphere^{d-1}_{+}\) and \(\Sphere^{d-1}_{++}\) denote its intersections
with the nonnegative orthant and the strictly positive orthant, respectively.
For \(x,y\in\Sphere^{d-1}\), \[d_g(x,y):=\arccos(\ip{x}{y})\] denotes the
spherical geodesic distance, and for \(B\subset\Sphere^{d-1}\),
\(d_g(x,B):=\inf_{b\in B}d_g(x,b)\). For \(\mu\in\Sphere^{d-1}\),
\(\mu^\perp:=\{x\in\R^d:\ip{x}{\mu}=0\}\) denotes the tangent hyperplane to the
sphere at \(\mu\), and for a linear subspace \(W\), \(\Proj_W\) is the
Euclidean orthogonal projector onto \(W\). The matrix \(I\) is the identity,
\(\one\) is the vector of ones, and
\(\Span(\cdot)\), \(\tr(\cdot)\), \(\rank(\cdot)\),
\(\operatorname{range}(\cdot)\), and \(\diag(\cdot)\) denote span, trace, rank,
range, and diagonal matrix, respectively. The operator norm is
\(\norm{\cdot}_{\mathrm{op}}\), \(\abs{B}\) is the cardinality of a finite set
\(B\), and \(a_+:=\max(a,0)\) is the positive part of \(a\in\R\). The indicator
of an event \(A\) is \(\mathbb I\{A\}\), \(\mathcal L(Z)\) is the law of a
random element \(Z\), \(\Prob\), \(\E\), \(\Var\), and \(\Cov\) denote
probability, expectation, variance, and covariance, respectively,
\(\rightsquigarrow\) denotes weak convergence, and stochastic order notation
such as \(o_p(\cdot)\) has its usual meaning. For a bounded real-valued
function \(f:\R^d\to\R\), let
\[
\norm{f}_\infty:=\sup_{z\in\R^d}\abs{f(z)},
\qquad
\Lip(f):=\sup_{z\ne z'}\frac{\abs{f(z)-f(z')}}{\norm{z-z'}_2},
\]
with the convention \(\Lip(f)=\infty\) if \(f\) is not Lipschitz.

\subsection{Angular limits and the anchor}
\label{sec:population-angular-anchor}

Let \(Y=(Y_1,\ldots,Y_d)\) have continuous marginal distribution functions
\(F_1,\ldots,F_d\), and put the data on a common tail scale by the standard Pareto transform
\[
X
:=
\left(
(1-F_1(Y_1))^{-1},\ldots,(1-F_d(Y_d))^{-1}
\right).
\]
Write \(R:=\norm{X}_2\). For standardized \(X\), the norm used for polar normalization only
chooses one representative on each positive ray; it preserves component ratios and does not
change the finite-dimensional tail measure \citep[Chapter~6]{de2006extreme,resnick2007heavy}.
We use the Euclidean representative because it places angular extremes on the ordinary unit
sphere, where geodesics, tangent spaces, great subspheres, and orthogonal projections have their
standard forms.

The population target is the limiting conditional law of the direction \(X/\norm{X}_2\) given
large \(R\). It records relative participation across variables in an extreme episode after
marginal tail scales and radial size have been removed.

\begin{assumption}[Euclidean-spherical angular limit]
\label{ass:angular-limit}
There exists a random vector \(G\) taking values in \(\Sphere^{d-1}_{+}\) such that
\[
\mathcal L\!\left(
X/\norm{X}_2
\,\middle|\, \norm{X}_2>r
\right)
\rightsquigarrow
\mathcal L(G),
\qquad r\to\infty.
\]
\end{assumption}

Assumption~\ref{ass:angular-limit} is the Euclidean-spherical form of standard angular
convergence in multivariate extremes. It is deliberately only directional, so it allows mass near
the diagonal, coordinate faces, or axes. It is implied by Euclidean-polar regular variation as in
Assumption~\ref{ass:erv} below, but is weaker because it does not require radial regular
variation or a homogeneous tail measure on the cone.

AGCA is built relative to an interior reference direction \(\mu\in\Sphere^{d-1}_{++}\).
The anchor fixes the origin of the score system. It may be specified from
subject-matter considerations, chosen as a canonical benchmark,
or estimated from the angular law. Our default choice is
\begin{equation}
\label{eq:canonical-anchor}
\mu_0:=d^{-1/2}\one.
\end{equation}
The direction \(\mu_0\) in \eqref{eq:canonical-anchor} is the spherical image of the diagonal
ray \(t\mapsto t\one\). This is the balanced
complete-dependence direction: after Pareto standardization, all components participate equally in
the angular profile. With \(\mu_0\), AGCA components are coordinates for asymmetric departures
from balanced extremal co-movement. Because the canonical anchor does not depend on the sample,
threshold, or fitted angular law, it can be used
across data sets and thresholds, and the resulting loadings retain the interpretation
of departures from balanced co-movement.
The canonical anchor also has a distribution-free minimax justification: among all interior
anchors, it uniquely maximizes the smallest cosine similarity with the positive sphere,
equivalently minimizing the worst-case geodesic distance and worst-case gnomonic coordinate size;
see Supplementary Proposition~\ref{prop:app-canonical-anchor}.

Data-adaptive anchors can improve empirical alignment, but they change the reference profile:
the resulting components describe variation around an empirical angular benchmark rather than a
fixed scientific one. This is analogous to the role of estimated base points in geodesic
dimension-reduction methods
\citep{fletcher2004principal,huckemann2006principal,jung2012analysis,dai2018principal}.
Two commonly used references are the spherical Fr\'echet anchor and the principal anchor,
\begin{equation}
\label{eq:data-adaptive-anchors}
\mu_{\mathrm F}
\in
\argmin_{\mu\in\Sphere^{d-1}_{+}}\E[d_g(G,\mu)^2],
\qquad
\mu_{\mathrm P}:=v_1(\E[GG^{\T}]),
\end{equation}
where \(v_1(M)\) denotes the unit leading eigenvector of a matrix \(M\), oriented to lie in
\(\Sphere^{d-1}_{++}\) when possible. We use the Fr\'echet reference as an AGCA anchor when a
chosen minimizer lies in \(\Sphere^{d-1}_{++}\); Supplementary
Section~\ref{app:alternative-routes} discusses these anchors and anchor-free alternatives.

\subsection{Anchored coordinates and model spaces}
\label{sec:population-coordinates-models}

For an interior anchor \(\mu\), all positive angular directions lie in the open hemisphere
\[
\mathcal H_\mu:=\{g\in\Sphere^{d-1}:\ip{\mu}{g}>0\},
\]
because \(\ip{\mu}{g}>0\) for every \(g\in\Sphere^{d-1}_{+}\). Thus coordinate-face and
near-axis extremes remain in a single anchored chart. For \(g\in\mathcal H_\mu\), write
\(a_\mu(g):=\ip{\mu}{g}\) for the anchor coordinate and
\(u_\mu(g):=(I-\mu\mu^{\T})g\in\mu^\perp\) for the bounded tangent departure. Since
\(a_\mu(g)=\cos d_g(g,\mu)\), the distance \(d_g(g,\mu)=\arccos a_\mu(g)\) measures how far
\(g\) lies from the anchor along the sphere, while \(u_\mu(g)\) records the departure from the
anchor axis inside the tangent hyperplane. The decomposition
\[
g=a_\mu(g)\mu+u_\mu(g),
\qquad
a_\mu(g)=(1-\norm{u_\mu(g)}_2^2)^{1/2},
\]
separates the anchored coefficient from the orthogonal departure.
We also use the gnomonic coordinate centered at \(\mu\),\footnote{Gnomonic coordinates project
points from the sphere's center onto the tangent affine hyperplane at \(\mu\). Great subspheres
through \(\mu\) become linear subspaces in this chart.}
\[
v_\mu(g):=\frac{g}{\ip{\mu}{g}}-\mu\in\mu^\perp,\qquad
\phi_\mu(v):=\frac{\mu+v}{(1+\norm{v}_2^2)^{1/2}},
\qquad v\in\mu^\perp.
\]
The coordinate \(v_\mu(g)\) is the displacement obtained by projecting the ray through \(g\) onto
the tangent affine hyperplane \(\mu+\mu^\perp\), and \(\phi_\mu\) maps such tangent displacements
back to the sphere. The bounded and gnomonic coordinates satisfy
\[
u_\mu(g)=\frac{v_\mu(g)}{(1+\norm{v_\mu(g)}_2^2)^{1/2}},
\qquad
v_\mu(g)=\frac{u_\mu(g)}{a_\mu(g)}.
\]
Thus \(v_\mu\) is the linear chart for anchored great subspheres, while \(u_\mu\) is the bounded
rescaling used for fitting. Near the edge of the hemisphere, \(v_\mu(g)\) can be large, whereas
\(u_\mu(g)\) always has norm at most one and keeps the same tangent direction.

For a linear subspace \(W\subset\mu^\perp\), define the anchored great-subsphere model
\[
M_\mu(W)
:=
\Span(\{\mu\}\cup W)\cap\mathcal H_\mu.
\]
Thus \(M_\mu(W)\) is the great subsphere through the anchor whose allowed directions away
from the anchor are the tangent directions in \(W\), restricted to the open hemisphere
\(\mathcal H_\mu\).
The rank of the approximation is \(\dim(W)\): rank zero is the anchor alone, and rank \(p\)
allows \(p\) tangent directions through the anchor. The target is therefore a low-dimensional
geodesic summary of angular variation around a fixed benchmark. Figure~\ref{fig:gca-geometry}
shows the model set, the gnomonic chart, and the bounded projection geometry.

\begin{figure}[t]
\centering
\begin{minipage}[t]{0.31\textwidth}
\centering
\appinclude[width=0.9\linewidth]{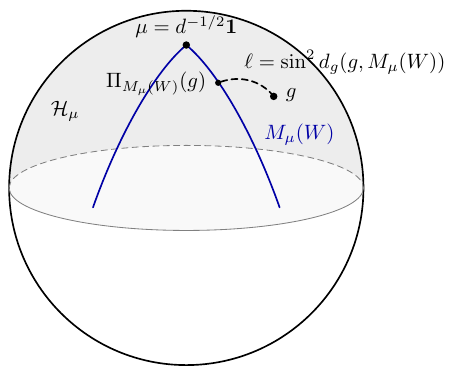}
\end{minipage}
\hfill
\begin{minipage}[t]{0.31\textwidth}
\centering
\appinclude[width=0.9\linewidth]{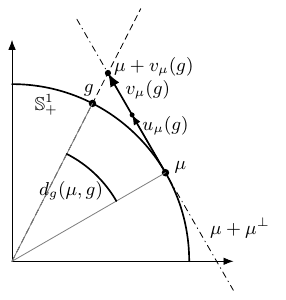}
\end{minipage}
\hfill
\begin{minipage}[t]{0.31\textwidth}
\centering
\appinclude[width=0.9\linewidth]{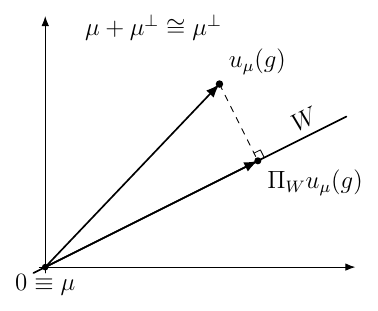}
\end{minipage}
\caption{Geometry of anchored geodesic component analysis. In the left panel, the open
hemisphere \(\mathcal H_\mu\) contains the positive Euclidean-spherical directions, and the
anchored great subsphere \(M_\mu(W)\) is the approximating set through the anchor. In the
central panel, the gnomonic chart centered at \(\mu\) sends the hemisphere \(\mathcal H_\mu\) to
the tangent affine hyperplane \(\mu+\mu^\perp\), identified with \(\mu^\perp\), and maps the
anchored great subsphere \(M_\mu(W)\) to the linear subspace \(W\). The bounded departure
\(u_\mu(g)\) is a radial rescaling of the gnomonic coordinate \(v_\mu(g)\). In the right panel,
under the loss \(\sin^2 d_g\), geodesic approximation of \(g\) by \(M_\mu(W)\) on the sphere is
the same as Euclidean orthogonal projection of \(u_\mu(g)\) onto \(W\) in \(\mu^\perp\).}
\label{fig:gca-geometry}
\end{figure}

\subsection{Population criterion and eigensolution}
\label{sec:population-projection-eigensolution}

With the anchor \(\mu\) fixed, the population fitting problem is to choose the tangent subspace
\(W\subset\mu^\perp\) that best reconstructs the limiting angular direction. Let \(G\) denote the
limiting angular direction in Assumption~\ref{ass:angular-limit}. For a \(p\)-dimensional subspace
\(W\subset\mu^\perp\), AGCA uses the anchored geodesic risk
\begin{equation}
\label{eq:population-anchored-risk}
\mathcal R_\mu(W)
:=
\E[\sin^2 d_g(G,M_\mu(W))],\qquad d_g(g,M_\mu(W))
:=
\inf_{m\in M_\mu(W)}\arccos(\ip{g}{m}).
\end{equation}
Here,
\(d_g(g,M_\mu(W))\)
denotes the shortest spherical distance from \(g\in\Sphere^{d-1}_{+}\) to the anchored great-subsphere model.
The rank \(p\) is the dimension of the tangent subspace used to reconstruct the angular law.
Equivalently, it is the number of anchored geodesic component directions retained.

The sine-squared geodesic residual is a central modeling choice. In principal geodesic analysis
and principal nested spheres, reconstruction is usually measured by squared geodesic distance,
often around a Fr\'echet mean, fitted base point, or recursively fitted small sphere
\citep[e.g.,][]{fletcher2004principal,jung2012analysis,sommer2014optimization,
tabaghi2024principal}. That loss is the direct Riemannian analogue of
Euclidean least squares, but it does not generally yield a
quadratic projection problem; implementations typically require tangent-space approximations or
nonlinear optimization over geodesics, base points, or spherical submanifolds. For the anchored
positive-sphere problem we instead use \(\sin^2 d_g\): on the admissible hemisphere it is bounded,
monotone in geodesic residual distance, comparable to squared geodesic loss, and has the same
zero-risk sets; see Supplementary Proposition~\ref{prop:app-loss-comparability}. Its key advantage is exact
linearization: the spherical residual becomes an ordinary Euclidean residual of the bounded
departure \(u_\mu(g)\), which leads to the eigendecomposition below.

\begin{proposition}[Exact anchored projection identity]
\label{prop:projection-identity}
For every linear subspace \(W\subset\mu^\perp\) and every \(g\in\Sphere^{d-1}_{+}\),
\begin{equation}
\label{eq:projection-identity}
\sin^2 d_g(g,M_\mu(W))
=\norm{u_\mu(g)-\Proj_Wu_\mu(g)}_2^2.
\end{equation}
Moreover, the unique geodesic projection of \(g\) onto \(M_\mu(W)\) is
\begin{equation}
\label{eq:projection-reconstruction}
\Proj_{M_\mu(W)}(g)
=
\frac{a_\mu(g)\mu+\Proj_Wu_\mu(g)}
{(a_\mu(g)^2+\norm{\Proj_Wu_\mu(g)}_2^2)^{1/2}}.
\end{equation}
\end{proposition}

Proposition~\ref{prop:projection-identity} uses standard spherical geometry to provide the exact
linearization step for our fixed-anchor extremal-angular reconstruction problem, as illustrated
in the right panel of Figure~\ref{fig:gca-geometry}. For a fixed anchor, the bounded geodesic
residual is exactly the Euclidean residual obtained by projecting \(u_\mu(g)\) onto \(W\). The
formula for \(\Proj_{M_\mu(W)}(g)\) then normalizes the retained anchored departure back to the
sphere. The identity is conditional on the anchor being held fixed: it applies after choosing
either the canonical anchor or a data-adaptive anchor, but it is not a joint optimization over the
anchor and the component space.

Applying the identity to the angular limit \(G\) from Assumption~\ref{ass:angular-limit} turns
AGCA into a second-moment problem for bounded anchored departures. Define
\begin{equation}
\label{eq:population-anchored-second-moment}
U_\mu:=u_\mu(G),
\qquad
\Sigma_\mu:=\E[U_\mu U_\mu^{\T}],
\end{equation}
where \(\Sigma_\mu\) acts on the tangent hyperplane \(\mu^\perp\). The vector \(U_\mu\) records
how the angular direction departs from the fixed anchor after radial magnitude has been removed.
Because \(U_\mu\) is bounded, \(\Sigma_\mu\) is finite even when the angular law has mass near
coordinate faces or axes, as in asymptotically independent regimes. Thus the origin of
\(\mu^\perp\) is the anchor, and AGCA approximates benchmark-relative departures by linear
subspaces through that origin. This boundary behavior is a distinguishing feature relative to
many projection, graphical, and hyperplane summaries of extremes:
Supplementary Section~\ref{app:axis-supported-agca} shows that face- and axis-supported angular
laws enter the same finite second-moment problem, with axis-dominant laws producing interpretable
loading contrasts and symmetric axis laws producing tied eigenvalues.

\begin{theorem}[Population eigensolution]
\label{thm:population-eigensolution}
For every \(p=0,\ldots,d-1\) and every \(p\)-dimensional \(W\subset\mu^\perp\),
the anchored geodesic risk satisfies
\(\mathcal R_\mu(W)=\tr((I-\Proj_W)\Sigma_\mu)\).
Write the spectral decomposition of \(\Sigma_\mu\) on \(\mu^\perp\),
with \(b_1,\ldots,b_{d-1}\) orthonormal in \(\mu^\perp\), as
\begin{equation}
\label{eq:population-spectral-decomposition}
\Sigma_\mu=\sum_{j=1}^{d-1}\lambda_j b_j b_j^{\T},
\qquad
\lambda_1\ge\cdots\ge\lambda_{d-1}\ge0.
\end{equation}
For \(p=0,\ldots,d-1\), define the population component space and spectral projector, and
for \(1\le p\le d-2\), define the associated eigengap, by
\begin{equation}
\label{eq:population-projector-eigengap}
\begin{gathered}
W_{\mu,p}:=\Span(b_1,\ldots,b_p),\qquad W_{\mu,0}:=\{0\},\\
P_{\mu,p}:=\sum_{j=1}^p b_jb_j^{\T},\qquad P_{\mu,0}:=0,\qquad
\Delta_{\mu,p}:=\lambda_p-\lambda_{p+1},\quad 1\le p\le d-2.
\end{gathered}
\end{equation}
Then
\(W_{\mu,p}\) minimizes \(\mathcal R_\mu(W)\) among all
\(p\)-dimensional anchored subspaces, and
\(\mathcal R_\mu(W_{\mu,p})=\sum_{j=p+1}^{d-1}\lambda_j\).
\end{theorem}

Theorem~\ref{thm:population-eigensolution} is the population PCA theorem for bounded anchored
departures. It gives an exact spectral solution for a geodesic reconstruction problem on
extremal angular laws: the leading eigenvectors of \(\Sigma_\mu\) are the optimal anchored
component directions, and the trailing eigensum is the residual geodesic risk.

\subsection{Rank, scores, and loading interpretation}
\label{sec:population-rank-scores-loadings}

The total anchored variation is
\begin{equation}
\label{eq:population-total-anchored-variation}
\tau_\mu
:=
\tr(\Sigma_\mu)
=
\E[\norm{U_\mu}_2^2]
=
\E[1-\ip{\mu}{G}^2].
\end{equation}
It is the rank-zero risk, namely the average sine-squared distance from \(G\) to the anchor.
The eigenvalues \(\lambda_j\) decompose this variation exactly: \(\lambda_j\) is the amount of
anchored geodesic variation captured by the \(j\)th component, and the trailing eigensum is the
residual risk after retaining \(p\) components. For \(p=0,\ldots,d-1\), the population
rank-\(p\) residual risk and, when \(\tau_\mu>0\), the proportion of anchored geodesic variation
explained are
\begin{equation}
\label{eq:population-rank-diagnostics}
\rho_{\mu,p}:=\sum_{j=p+1}^{d-1}\lambda_j,
\qquad
\mathrm{AVE}_{\mu,p}
:=(\tau_\mu)^{-1}\sum_{j=1}^p\lambda_j.
\end{equation}
Here \(P_{\mu,p}\) in \eqref{eq:population-projector-eigengap} is the rank-\(p\)
population spectral projector, and \(\Delta_{\mu,p}\) is the eigengap separating the retained
and discarded component spaces.
This gives the same operational clarity as ordinary PCA, but the interpretation is tied to
geodesic approximation of spherical extremal directions through the fixed anchor.
The next corollary identifies \(r_\mu\) as the exact anchored dimension: the
smallest number of anchored component directions needed for zero population reconstruction risk.

\begin{corollary}[Exact anchored rank]
\label{cor:anchored-rank}
Let \(r_\mu:=\rank(\Sigma_\mu)\). There exists a \(p\)-dimensional subspace
\(W\subset\mu^\perp\) with \(\mathcal R_\mu(W)=0\) if and only if \(p\ge r_\mu\). At the minimal
dimension \(p=r_\mu\), the exact component space is unique and equals
\(\operatorname{range}(\Sigma_\mu)\).
\end{corollary}

The scores and loadings have a direct interpretation as extremal contrasts. A loading
eigenvector \(b_\ell\in\mu^\perp\) is a tangent direction through the anchor, and the score of an
angular direction \(g_i\) is its coordinate along that direction,
\begin{equation}
\label{eq:population-agca-score}
\alpha_{i\ell}:=\ip{b_\ell}{u_\mu(g_i)}.
\end{equation}
Large positive or negative scores indicate strong movement away from the anchor along the
corresponding loading, while scores near zero indicate little contribution from that contrast.
The overall sign of a loading is conventional, but relative signs within a loading are meaningful
once an orientation is fixed.

Given a rank \(p\), the retained tangent departure is
\begin{equation}
\label{eq:population-rank-p-tangent-projection}
\Proj_{W_{\mu,p}}u_\mu(g_i)
=
\sum_{\ell=1}^p\alpha_{i\ell}b_\ell.
\end{equation}
Using \eqref{eq:projection-reconstruction}, normalization back to the sphere gives
\begin{equation}
\label{eq:population-rank-p-reconstruction}
g_i^{(p)}
=
\frac{a_\mu(g_i)\mu+\sum_{\ell=1}^p\alpha_{i\ell}b_\ell}
{\left(a_\mu(g_i)^2+\norm{\sum_{\ell=1}^p\alpha_{i\ell}b_\ell}_2^2\right)^{1/2}}.
\end{equation}
This is the AGCA reconstruction of the angular profile; radial severity is not reconstructed.

\begin{proposition}[Loadings as first-order log-ratio contrasts]
\label{prop:loading-log-ratio}
Let \(\eta\in\mu^\perp\), and let \(c_\eta(t):=\Expmap_\mu(t\eta)\) be the geodesic through
\(\mu\) with initial tangent direction \(\eta\), where \(\Expmap_\mu\) is the Riemannian
exponential map on the unit sphere at \(\mu\). Then, for every pair of coordinate indices \(j,m\),
\[
\left.\frac{\mathrm d}{\mathrm dt}\right|_{t=0}
\log\frac{c_{\eta,j}(t)}{c_{\eta,m}(t)}
=
\frac{\eta_j}{\mu_j}-\frac{\eta_m}{\mu_m}.
\]
\end{proposition}

Proposition~\ref{prop:loading-log-ratio} gives the local ratio interpretation. Along a component
direction \(\eta\), the sign of \(\eta_j/\mu_j-\eta_m/\mu_m\) determines whether positive scores
increase or decrease the pairwise angular ratio \(g_j/g_m\) relative to the anchor ratio
\(\mu_j/\mu_m\). For the canonical anchor \(\mu_0\), the baseline ratio is one and the derivative
reduces to \(d^{1/2}(\eta_j-\eta_m)\). Thus loading contrasts identify which variables drive
benchmark-relative changes in angular participation. This is why the canonical anchor is the
default in the simulations and empirical application: data-driven anchors may improve fit, but
their loadings are less clearly interpretable because they describe departures from an empirical
angular benchmark rather than from balanced complete dependence.

\subsection{Low-rank reconstruction and tail simulation}
\label{sec:population-reconstruction-simulation}

The population eigensolution also defines a natural angular reconstruction. This is useful when
AGCA is used not only as a descriptive summary, but also as a low-dimensional device for simulating
extreme directions. The reconstruction should be understood as an angular approximation: AGCA
controls how far the limiting direction moves on the sphere, while a separate radial tail model is
needed to simulate the size of the extreme event.

Using the component space \(W_{\mu,p}\) in \eqref{eq:population-rank-p-tangent-projection}, we
define the rank-\(p\) reconstructed angular direction as the geodesic projection of the angular
limit \(G\) from Assumption~\ref{ass:angular-limit} onto the AGCA model:
\begin{equation}
\label{eq:population-reconstructed-angular}
\widetilde G_{\mu,p}
:=
\Proj_{M_\mu(W_{\mu,p})}(G)=
\frac{a_\mu(G)\mu+\Proj_{W_{\mu,p}}u_\mu(G)}
{\left(a_\mu(G)^2+\norm{\Proj_{W_{\mu,p}}u_\mu(G)}_2^2\right)^{1/2}}.
\end{equation}
Thus the reconstruction retains the projection of \(u_\mu(G)\) onto the leading AGCA component
space and renormalizes the result to a spherical direction. The next proposition shows that this
unconstrained spherical reconstruction has an exact AGCA error and a
corresponding Euclidean error bound.

\begin{proposition}[AGCA angular reconstruction]
\label{prop:population-agca-reconstruction}
For \(p=0,\ldots,d-1\),
\[\sin^2 d_g(G,\widetilde G_{\mu,p})
=
\norm{u_\mu(G)-\Proj_{W_{\mu,p}}u_\mu(G)}_2^2.\]
Consequently,
\[
\E[\sin^2 d_g(G,\widetilde G_{\mu,p})]
=
\rho_{\mu,p}, \qquad
\E\left[\norm{G-\widetilde G_{\mu,p}}_2^2\right]
=
\E\left[
\frac{2\sin^2 d_g(G,\widetilde G_{\mu,p})}
{1+\cos d_g(G,\widetilde G_{\mu,p})}
\right]
\le
2\rho_{\mu,p}.
\]
\end{proposition}

To turn an angular reconstruction into a tail simulator, the radial size must be generated
separately and then attached to the reconstructed direction. Let \(S>0\) denote such a radial
draw. If the same radius is attached to \(G\) and to \(\widetilde G_{\mu,p}\), the relative
Euclidean error is purely angular and is controlled by the AGCA residual risk:
\[
\E\left[
\norm{SG-S\widetilde G_{\mu,p}}_2^2/S^2
\right]
=
\E\left[\norm{G-\widetilde G_{\mu,p}}_2^2\right]
\le
2\rho_{\mu,p}.
\]
Thus AGCA controls the relative vector error once the radial size has been supplied, but it does
not supply the radial tail law itself. For simulation of extremes one must combine the reconstructed
angular law with an extreme-value model for the radius. We use the following standard
Euclidean-polar regular variation condition for this purpose.

\begin{assumption}[Euclidean-polar regular variation]
\label{ass:erv}
There exists \(\alpha>0\) such that, for a Pareto variable \(P_\alpha\) independent of \(G\) and
satisfying \(\Prob[P_\alpha>x]=x^{-\alpha}\) for \(x\ge1\),
\[
\mathcal L\!\left(
R/r,\ X/R
\,\middle|\, R>r
\right)
\rightsquigarrow
\mathcal L(P_\alpha,G),
\qquad r\to\infty.
\]
\end{assumption}

Assumption~\ref{ass:erv} strengthens Assumption~\ref{ass:angular-limit} by adding regular radial
scaling. By the continuous mapping theorem,
\(
\mathcal L\!\left(X/r\,\middle|\, R>r\right)
\rightsquigarrow
\mathcal L(P_\alpha G)\),
so the normalized rank-\(p\) AGCA tail-simulation target is
\(P_\alpha\widetilde G_{\mu,p}\). Multiplying by the threshold gives
\(rP_\alpha\widetilde G_{\mu,p}\), the corresponding simulator on the original Pareto scale.

In applications, however, the goal is often not to approximate the whole conditional law of
extremes for its own sake, but to approximate tail functionals computed from that law. Examples
include portfolio exceedance probabilities, value-at-risk, expected shortfall when the relevant
tail moment exists, and stress scores such as the maximum loss over a collection of positions. The
next two results cover complementary classes of such summaries. The first gives a uniform
transport-type bound for bounded Lipschitz functionals of the simulated extreme vector. For
probability laws on \(\R^d\), write the bounded-Lipschitz distance as
\[
d_{\mathrm{BL}}(\mathcal L(A),\mathcal L(B))
:=
\sup_{\norm{f}_\infty\le1,\ \Lip(f)\le1}
\abs{\E[f(A)]-\E[f(B)]}.
\]

\begin{proposition}[Tail simulation under regular variation]
\label{prop:population-tail-simulation-bound}
Suppose Assumption~\ref{ass:erv} holds, and let \(P_\alpha\) be independent of \(G\) with
\(\Prob[P_\alpha>x]=x^{-\alpha}\), \(x\ge1\). Put
\begin{equation}
\label{eq:population-tail-threshold-error}
\varepsilon_r
:=
d_{\mathrm{BL}}\!\left(
\mathcal L\!\left(X/r\,\middle|\, R>r\right),
\mathcal L(P_\alpha G)
\right).
\end{equation}
Then \(\varepsilon_r\to0\). Moreover, for every \(0<\beta\le1\) with \(\beta<\alpha\),
\[
d_{\mathrm{BL}}\!\left(
\mathcal L\!\left(X/r\,\middle|\, R>r\right),
\mathcal L(P_\alpha\widetilde G_{\mu,p})
\right)
\le
\varepsilon_r
+
2^{1-\beta/2}\frac{\alpha}{\alpha-\beta}\,
\rho_{\mu,p}^{\beta/2}.
\]
\end{proposition}

The proposition gives a uniform bound over bounded Lipschitz test functionals. This is the sense
in which the low-rank AGCA simulator is controlled by the sum of two errors: the threshold error
from replacing finite exceedances by the regular-variation limit, and the angular approximation
error \(\rho_{\mu,p}\). When \(\alpha>1\), one may take \(\beta=1\) and obtain the familiar
\(\sqrt{\rho_{\mu,p}}\) rate for the angular term. With exact standard Pareto marginal scaling,
the radial index is \(\alpha=1\), so the statement uses any \(\beta<1\).

The reconstruction \(\widetilde G_{\mu,p}\) is an unconstrained spherical projection. Hence
\(M_\mu(W_{\mu,p})\subset\mathcal H_\mu\) need not lie in \(\Sphere^{d-1}_{+}\), and simulated
vectors \(P_\alpha\widetilde G_{\mu,p}\) need not remain in the positive cone. For positive-cone
simulation, one should verify positivity on the relevant support, impose a constrained angular
reconstruction, or model residual variation around the unconstrained AGCA fit.
Supplementary Section~\ref{app:positive-post-projection} develops a positive post-projection
that enforces \(\Sphere^{d-1}_{+}\)-valued simulated directions while adding an explicit
correction term to the simulation bounds.

Many risk summaries are not bounded. For positively homogeneous scores, the Pareto radial excess
can instead be separated from the angular profile. Let \(h:\R^d\to[0,\infty)\) satisfy
\(h(cz)=c h(z)\) for \(c\ge0\). Suppose that, on \(\Sphere^{d-1}\), for constants
\(0<H_h<\infty\) and \(L_h<\infty\),
\[
0\le h(s)\le H_h,
\qquad
\abs{h(s)-h(t)}\le L_h\norm{s-t}_2.
\]
Set \(C_h:=\E[h(G)^\alpha]\), \(\widetilde C_{h,p}:=\E[h(\widetilde G_{\mu,p})^\alpha]\),
\(\gamma:=\min(\alpha,1)\), and define
\[
\eta_{r,h}
:=
\sup_{x\ge H_h}
\abs{
\frac{\Prob[h(X)>rx]}{\Prob[R>r]}
-
x^{-\alpha}C_h
},\qquad
K_{\alpha,h}:=
\begin{cases}
L_h^\alpha, & 0<\alpha\le1,\\
\alpha H_h^{\alpha-1}L_h, & \alpha>1.
\end{cases}
\]
The next proposition is the homogeneous-score analogue of the bounded-Lipschitz bound: the
finite-threshold error is inherited from Assumption~\ref{ass:erv}, and the additional error comes
from the AGCA angular approximation.

\begin{proposition}[Homogeneous tail scores]
\label{prop:population-homogeneous-tail-scores}
Suppose Assumption~\ref{ass:erv} holds. Then \(\eta_{r,h}\to0\) as \(r\to\infty\). Moreover, for
every \(x\ge H_h\),
\[
\Prob[h(P_\alpha G)>x]=x^{-\alpha}C_h,
\qquad
\Prob[h(P_\alpha\widetilde G_{\mu,p})>x]=x^{-\alpha}\widetilde C_{h,p}.
\]
The angular constants satisfy
\(\abs{C_h-\widetilde C_{h,p}}
\le
K_{\alpha,h}(2\rho_{\mu,p})^{\gamma/2}\).
Consequently,
\[
\sup_{x\ge H_h}
\abs{
\frac{\Prob[h(X)>rx]}{\Prob[R>r]}
-
\Prob[h(P_\alpha\widetilde G_{\mu,p})>x]
}
\le
\eta_{r,h}
+
H_h^{-\alpha}K_{\alpha,h}(2\rho_{\mu,p})^{\gamma/2}.
\]
\end{proposition}

Thus the two simulation bounds have the same structure: a first-order regular-variation error,
\(\varepsilon_r\) in \eqref{eq:population-tail-threshold-error} or \(\eta_{r,h}\), plus an AGCA
angular approximation error. Assumption~\ref{ass:erv} ensures that the threshold terms vanish, but
does not quantify their rates. Rates for either term would require second-order regular variation
or another quantitative tail approximation condition.

The preceding bounds do not imply that the AGCA simulator can approximate
\(\Prob[P_\alpha G\in A]\) uniformly over all Borel sets \(A\). Such a statement would require
total-variation control of the law of \(P_\alpha G\) by the law of
\(P_\alpha\widetilde G_{\mu,p}\). Low-rank reconstruction changes the support of the angular law:
\(P_\alpha\widetilde G_{\mu,p}\) is supported on the cone generated by \(M_\mu(W_{\mu,p})\),
whereas \(P_\alpha G\) need not be.

\begin{proposition}[Total-variation obstruction]
\label{prop:population-tv-obstruction}
For \(P_\alpha\) as above, let
\[
\mathcal C_{\mu,p}
:=
\{z\in\R^d:\norm{z}_2\ge1,\ z/\norm{z}_2\in M_\mu(W_{\mu,p})\}.
\]
Then
\(\Prob[P_\alpha\widetilde G_{\mu,p}\in\mathcal C_{\mu,p}]=1\) and
\(\Prob[P_\alpha G\in\mathcal C_{\mu,p}]
=
\Prob[G\in M_\mu(W_{\mu,p})]\).
Consequently,
\[
\sup_A
\abs{\Prob[P_\alpha G\in A]-\Prob[P_\alpha\widetilde G_{\mu,p}\in A]}
\ge
1-\Prob[G\in M_\mu(W_{\mu,p})],
\]
where the supremum is over Borel \(A\subset\R^d\).
\end{proposition}

If the angular law is full-dimensional and \(p<d-1\), the probability
\(\Prob[G\in M_\mu(W_{\mu,p})]\) is typically zero, so the total-variation discrepancy is maximal
even when \(\rho_{\mu,p}\) is small. This is not a defect of the projection formula; it is the
usual support mismatch created by any exact low-dimensional representation. Total-variation
closeness is available only when the angular law is itself supported on the fitted model, or after
adding an explicit residual/noise model around \(M_\mu(W_{\mu,p})\). This low-dimensional support
mismatch is separate from the positivity issue discussed above.

\subsection{Portfolio functionals for tail simulation}
\label{sec:population-portfolio-var}

We make the preceding bounds operational through portfolio scores. Let \(z\in\R^d\) denote a vector
of standardized portfolio losses and let \(w\in\R^d\) collect portfolio positions. The portfolio
loss score is \(\pi_w(z):=w^{\T}z\), and we impose the budget constraint
\(\one^{\T}w=1\). In applications below, the leverage constraint \(\norm{w}_1\le L\),
for a positive constant \(L\), gives uniform control over a portfolio class. We highlight two
summaries that are useful for portfolio
comparison and risk control. The first is a bounded score for comparing candidate portfolios under
simulated tail scenarios. For \(t\ge0\), define the capped portfolio excess
\begin{equation}
\label{eq:population-capped-portfolio-excess-score}
f_{w,t,L}(z)
:=
\min\left\{1,L^{-1}(\pi_w(z)-t)_+\right\}.
\end{equation}
It records excess over a portfolio stress level, but caps the contribution of very large radial
draws. Since \(\norm{w}_2\le\norm{w}_1\le L\), it satisfies
\(\norm{f_{w,t,L}}_\infty\le1\) and \(\Lip(f_{w,t,L})\le1\), and is therefore covered directly by
Proposition~\ref{prop:population-tail-simulation-bound}.

\begin{corollary}[Capped portfolio excess]
\label{cor:population-capped-portfolio-excess}
Suppose Assumption~\ref{ass:erv} holds. Let \(w\in\R^d\), \(\one^{\T}w=1\),
\(\norm{w}_1\le L\) for some \(L\ge1\), \(t\ge0\), and let \(\varepsilon_r\) be as in
\eqref{eq:population-tail-threshold-error}. Then, for every \(0<\beta\le1\) with \(\beta<\alpha\),
\[
\abs{
\E\!\left[f_{w,t,L}(X/r)\mid R>r\right]
-
\E[f_{w,t,L}(P_\alpha\widetilde G_{\mu,p})]
}
\le
\varepsilon_r
+
2^{1-\beta/2}\frac{\alpha}{\alpha-\beta}\rho_{\mu,p}^{\beta/2}.
\]
\end{corollary}

The second summary is value-at-risk. For a scalar portfolio loss score \(Y\), the upper-tail
value-at-risk at tail probability \(u\),
\(\operatorname{VaR}_u(Y):=\inf\{q\in\R:\Prob[Y>q]\le u\}\), is the high quantile exceeded with
probability at most \(u\). In the Pareto limit, VaR for the score \(\pi_w\) is driven by the
unbounded homogeneous score \((w^{\T}z)_+\). Define the true and reconstructed angular tail constants
\begin{equation}
\label{eq:population-portfolio-tail-constants}
C_w:=\E[(w^{\T}G)_+^\alpha],
\qquad
\widetilde C_{w,p}:=\E[(w^{\T}\widetilde G_{\mu,p})_+^\alpha].
\end{equation}
For tail probability \(u\), define the normalized Pareto-limit VaR values for the true and AGCA-reconstructed portfolio
tails
\begin{equation}
\label{eq:population-pareto-var-values}
q_w(u):=(C_w/u)^{1/\alpha},
\qquad
\widetilde q_{w,p}(u):=(\widetilde C_{w,p}/u)^{1/\alpha}.
\end{equation}
Multiplying them by a high threshold \(r\) gives approximate VaR values on the original
Pareto scale.

\begin{corollary}[Portfolio tail constants and VaR]
\label{cor:population-portfolio-var}
Suppose Assumption~\ref{ass:erv} holds. Let \(\pi_w(z):=w^{\T}z\), where
\(w\in\R^d\), \(\one^{\T}w=1\), and \(\norm{w}_1\le L\) for some \(L\ge1\). Then, for every
\(x\ge L\),
\[
\Prob[\pi_w(P_\alpha G)>x]=x^{-\alpha}C_w,
\qquad
\Prob[\pi_w(P_\alpha\widetilde G_{\mu,p})>x]=x^{-\alpha}\widetilde C_{w,p}.
\]
For \(0<u<L^{-\alpha}\min(C_w,\widetilde C_{w,p})\),
the normalized VaR ratio satisfies
\(q_w(u)/\widetilde q_{w,p}(u)=(C_w/\widetilde C_{w,p})^{1/\alpha}\). If
\(\min(C_w,\widetilde C_{w,p})\ge c>0\), then
\[
\abs{\log q_w(u)-\log\widetilde q_{w,p}(u)}
\le
\frac{1}{\alpha c}(2\rho_{\mu,p})^{\min(\alpha,1)/2}
\begin{cases}
L^\alpha, & 0<\alpha\le1,\\
\alpha L^\alpha, & \alpha>1.
\end{cases}
\]
\end{corollary}

Thus the capped excess gives a bounded simulation diagnostic, while VaR gives a scale-normalized
tail risk summary. A related important measure is expected shortfall,
\(\operatorname{ES}_u(Y):=\E[Y\mid Y>\operatorname{VaR}_u(Y)]\), the conditional mean beyond
the corresponding high quantile. For the Pareto limit considered here, 
this mean is finite only when \(\alpha>1\). In that case,
the limiting expected shortfall equals \(\alpha/(\alpha-1)\) times the corresponding VaR, so the
same relative angular error bound applies. For exact standard Pareto marginal scaling,
\(\alpha=1\), the limiting expected shortfall is infinite, whereas VaR remains finite.

%% file: 03_estimation.tex
\section{AGCA: estimation and inference}
\label{sec:estimation}

In data analysis we observe sample vectors \(Y_i=(Y_{i1},\ldots,Y_{id})\), standardize the
margins, select observations with large radial size, and apply the population eigensolution of
Section~\ref{sec:population-theory} to the selected angular directions. The tail sample size
\(k\) plays the usual peaks-over-threshold role: it balances extremeness against estimation
variance. In asymptotic statements we use an intermediate sequence \(k_n\), with
\(k_n\to\infty\) and \(k_n/n\to0\).
We use hats with the superscript \(\mathrm{orc}\) to denote empirical quantities computed
from the true Pareto-standardized margins, while hats without the oracle superscript denote
empirical quantities computed from rank-Pareto margins.
The section separates the deterministic eigensolution, consistency under top-\(k\) selection and
rank marginal standardization, oracle distributional inference, and the remaining first-order
effect of estimating the margins.

\subsection{Empirical construction and eigensolution}
\label{sec:estimation-empirical-eigensolution}

For a fixed rank \(p\), anchor \(\mu\), and tail sample size \(k\), we use two parallel sample
constructions. For the oracle theory, let \(X_1,\ldots,X_n\) be independent copies of the
Pareto-standardized vector, set \(R_i:=\norm{X_i}_2\), \(G_i:=X_i/R_i\), and
\(U_{\mu,i}:=u_\mu(G_i)\), and let \(\mathcal I_n^{(k)}\) be the indices of the \(k\) largest
values \(R_i\). With unknown margins, a standard rank-Pareto transform is
\begin{equation}
\label{eq:pseudo-pareto-transform}
\widehat X_{ij}
:=
\frac{n+1}{n+1-\operatorname{rank}(Y_{ij})},
\qquad i=1,\ldots,n,\ j=1,\ldots,d,
\end{equation}
with ranks computed within each component. This is the standard empirical EVT device for removing
marginal tail scales before estimating extremal dependence. Set
\(\widehat R_i:=\norm{\widehat X_i}_2\), \(\widehat G_i:=\widehat X_i/\widehat R_i\), and
\(\widehat U_{\mu,i}:=u_\mu(\widehat G_i)\), and let
\(\widehat{\mathcal I}_n^{(k)}\) be the indices of the \(k\) largest values \(\widehat R_i\).
In both constructions, radii select tail observations and directions are the angular profiles
summarized by AGCA. The corresponding anchored second-moment matrices are
\begin{equation}
\label{eq:sample-anchored-second-moments}
\widehat\Sigma_{\mu,n}^{(k),\mathrm{orc}}
:=
\frac1{k}\sum_{i\in\mathcal I_n^{(k)}}U_{\mu,i}U_{\mu,i}^{\T},
\qquad
\widehat\Sigma_{\mu,n}^{(k),\mathrm{emp}}
:=
\frac1{k}\sum_{i\in\widehat{\mathcal I}_n^{(k)}}
\widehat U_{\mu,i}\widehat U_{\mu,i}^{\T}.
\end{equation}
Their eigenvectors represent departures from the fixed benchmark \(\mu\). The rank-Pareto
eigensolution diagonalizes \(\widehat\Sigma_{\mu,n}^{(k),\mathrm{emp}}\) in the tangent hyperplane
\(\mu^\perp\). Explicitly, let
\(Q_\mu\in\R^{d\times(d-1)}\) have orthonormal columns spanning \(\mu^\perp\), so that
\(Q_\mu^{\T}Q_\mu=I_{d-1}\) and \(Q_\mu Q_\mu^{\T}=I-\mu\mu^{\T}\).
Diagonalize the ordinary \((d-1)\times(d-1)\) coordinate matrix
\[
Q_\mu^{\T}
\widehat\Sigma_{\mu,n}^{(k),\mathrm{emp}}
Q_\mu
=
H\widehat\Lambda H^{\T},
\qquad
\widehat\Lambda=\diag(\widehat\lambda_1,\ldots,\widehat\lambda_{d-1}),
\]
with \(\widehat\lambda_1\ge\cdots\ge\widehat\lambda_{d-1}\ge0\). If \(h_\ell\) is the
\(\ell\)th column of \(H\), the corresponding AGCA loading in the original coordinates is
\(\widehat b_\ell:=Q_\mu h_\ell\in\mu^\perp\). For later use, write
\(\widehat P_{\mu,p}:=\sum_{\ell=1}^p\widehat b_\ell\widehat b_\ell^{\T}\), with
\(\widehat P_{\mu,0}:=0\), for the empirical rank-\(p\) spectral projector. The empirical scores are the plug-in analogues
of \eqref{eq:population-agca-score},
\(\widehat\alpha_{i\ell}:=\ip{\widehat b_\ell}{\widehat U_{\mu,i}}\),
and the plug-in version of \eqref{eq:population-rank-p-reconstruction}, equivalently
\eqref{eq:projection-reconstruction}, gives the rank-\(p\) reconstructions. The choice of
orthonormal basis \(Q_\mu\) only changes the coordinates of the calculation; the resulting
subspaces in \(\mu^\perp\) are the same, up to the usual sign and eigenspace
nonuniqueness.

The next statement treats the selected directions as fixed, so the corresponding second moment
and risk are written without hats.

\begin{proposition}[Finite-sample eigensolution for selected directions]
\label{prop:empirical-eigensolution}
Let \(\mathcal I\subset\mathbb N\) be a nonempty finite index set and let
\(\{g_i:i\in\mathcal I\}\subset\Sphere^{d-1}_{+}\) be selected angular directions. Define
\[
\Sigma_{\mu,\mathcal I}
:=
\frac1{\abs{\mathcal I}}\sum_{i\in\mathcal I}
u_\mu(g_i)u_\mu(g_i)^{\T},\qquad \mathcal R_{\mu,\mathcal I}(W)
:=
\frac1{\abs{\mathcal I}}\sum_{i\in\mathcal I}
\sin^2 d_g(g_i,M_\mu(W)),\quad W\subset\mu^\perp.
\]
Then, for every linear subspace \(W\subset\mu^\perp\),
\(
\mathcal R_{\mu,\mathcal I}(W)
=
\tr((I-\Proj_W)\Sigma_{\mu,\mathcal I}).
\)
Consequently, if \(b_{\mathcal I,1},\ldots,b_{\mathcal I,d-1}\) are orthonormal eigenvectors of
\(\Sigma_{\mu,\mathcal I}\) in \(\mu^\perp\), with eigenvalues
\(\lambda_{\mathcal I,1}\ge\cdots\ge\lambda_{\mathcal I,d-1}\), then
\(
W_{\mu,\mathcal I,p}:=\Span(b_{\mathcal I,1},\ldots,b_{\mathcal I,p})
\)
minimizes \(\mathcal R_{\mu,\mathcal I}(W)\) over all \(p\)-dimensional
\(W\subset\mu^\perp\), and the corresponding rank-\(p\) residual risk is
\(
\mathcal R_{\mu,\mathcal I}(W_{\mu,\mathcal I,p})
=
\sum_{j=p+1}^{d-1}\lambda_{\mathcal I,j}.
\)
\end{proposition}

The selected set \(\mathcal I\) in Proposition~\ref{prop:empirical-eigensolution} can be the
top-\(k\) radial set or a rank-margin version of that set. The proposition is deterministic:
conditional on those directions, empirical AGCA is the leading eigenspace of the second-moment
operator formed from the selected directions, with only the usual eigenspace nonuniqueness under
ties.

\subsection{\texorpdfstring{Top-\(k\)}{Top-k} consistency and marginal standardization}
\label{sec:estimation-topk-margins}

The following theorem gives the oracle top-\(k\) justification for
\(\widehat\Sigma_{\mu,n}^{(k_n),\mathrm{orc}}\) in
\eqref{eq:sample-anchored-second-moments}. The next theorem gives the corresponding primitive
rank-Pareto consistency result.

\begin{theorem}[Consistency under empirical intermediate thresholds]
\label{thm:top-k-consistency}
Assume that \(R_1\) has a continuous distribution and that
Assumption~\ref{ass:angular-limit} holds. Let \(k_n\to\infty\) and \(k_n/n\to0\).
Then \(\norm{\widehat\Sigma_{\mu,n}^{(k_n),\mathrm{orc}}-\Sigma_\mu}_{\mathrm{op}}
\to0\) in probability.
Consequently, the eigenvalues, anchored residual risks, and, under an eigengap, the rank-\(p\)
AGCA projector obtained from the top \(k_n\) radial observations are consistent for their
population counterparts.
\end{theorem}

The main tail condition in Theorem~\ref{thm:top-k-consistency} is
Assumption~\ref{ass:angular-limit}: conditional on being far enough in the tail, the Euclidean
spherical direction must approach a limiting angular law. This is a standard condition that
underlies empirical angular-measure estimation in multivariate EVT\@.
The continuity assumption on \(R_1\) ensures that the top-\(k_n\)
radial set is unambiguously defined: with probability one, no positive number of observations
has the same radius as the empirical threshold. It is mild for continuous measurements
and for continuous Pareto-standardized variables, and it does not exclude discrete angular
regimes such as mass on coordinate axes.\footnote{The assumption could be relaxed by specifying a
deterministic or randomized convention for observations with equal radii, or by replacing exact
top-\(k_n\) selection with a threshold exceedance set and assuming that the number of observations
at the boundary is \(o_p(k_n)\). Such extensions are useful for heavily rounded, count-valued, or
otherwise coarse data.}
The intermediate sequence \(k_n\) only requires increasingly many selected extremes with vanishing
tail fraction \(k_n/n\); finite-sample choice remains a threshold-selection problem.
In applications, this is normally handled through threshold paths and stability diagnostics. For
AGCA the natural diagnostics are eigenvalues, anchored variation explained, residual risk,
loadings, scores, and projector distances across a range of \(k\).

Replacing the unknown marginal transforms by ranks affects the estimator in two ways: the
direction of a fixed observation changes from \(G_i\) to \(\widehat G_i\), and the top-\(k_n\) set based on
\(\widehat R_i\) need not equal the oracle top-\(k_n\) set. The next theorem
shows that both effects are negligible at the consistency scale under primitive tail conditions.
Rank-based estimation of limiting angular (spectral) measures goes back to
\citet{einmahl1997estimating,einmahl2001nonparametric}; see also \citet{einmahl2009maximum}
and, for the multivariate rank-tail empirical process, \citet{einmahl2012m}.

\begin{theorem}[Rank-based consistency]
\label{thm:rank-consistency}
Assume that the margins \(F_1,\ldots,F_d\) and \(R_1\) are continuous, and that
Assumption~\ref{ass:erv} holds. Let \(k_n\to\infty\) and \(k_n/n\to0\). Then
\(\norm{\widehat\Sigma_{\mu,n}^{(k_n),\mathrm{emp}}-\Sigma_\mu}_{\mathrm{op}}\to0\)
in probability. In particular, the eigenvalues, anchored residual risks, anchored variation
explained when \(\tau_\mu>0\), and, under an eigengap, the rank-\(p\) AGCA projector obtained from
the rank-Pareto top-\(k_n\) observations are consistent for their population counterparts.
\end{theorem}

The proof, in Supplementary Section~\ref{app:primitive-margin-conditions}, shows that the average
rank-induced angular error over the relevant selected observations is \(o_p(1)\), and that the
symmetric difference between the oracle and rank-based selected sets is \(o_p(k_n)\). The latter
step is why Theorem~\ref{thm:rank-consistency} uses Assumption~\ref{ass:erv}: angular convergence
alone does not control how many radii lie in a shrinking relative band around the top-\(k_n\)
threshold. The average form of the angular control is essential, because the rank transform
distorts the top few order statistics of each margin at constant order.

\subsection{Oracle distributional inference and stability diagnostics}
\label{sec:stability-diagnostics}
\label{sec:oracle-inference}

The consistency result above is enough to justify the AGCA target, but it does not quantify
sampling variation. For the oracle estimator, distributional inference is particularly simple
because the anchored departure \(u_\mu(G)\) is bounded. The only additional condition needed here
is a second-order version of the angular limit at the \(k_n/n\) tail level.

With the oracle notation above, define 
\[
\bar F_R(r):=\Prob[R_1>r],\qquad
H_\mu(r)
:=
\E\!\left[u_\mu(G_1)u_\mu(G_1)^{\T}\mid R_1>r\right],
\]
and the angular-bias envelope
\[
\eta_\mu(t)
:=
\sup\left\{
\norm{H_\mu(r)-\Sigma_\mu}_{\mathrm{op}}
:\bar F_R(r)\le t
\right\},
\qquad 0<t<1.
\]

\begin{assumption}[Second-order angular bias]
\label{ass:second-order-angular-bias}
As \(n\to\infty\), \(\sqrt{k_n}\,\eta_\mu(2k_n/n)\to0\) for the intermediate sequence
\(k_n\to\infty\) with \(k_n/n\to0\).
\end{assumption}
Assumption~\ref{ass:second-order-angular-bias} assumes away the asymptotic bias. Such a condition is often assumed in classical extreme value statistics; see, e.g.
\citep[Chapter~3]{de2006extreme}. It places an upper bound for the intermediate sequence $k_n$.
\begin{theorem}[Oracle CLT for top-\(k\) AGCA]
\label{thm:oracle-agca-clt}
Assume the conditions of Theorem~\ref{thm:top-k-consistency} and
Assumption~\ref{ass:second-order-angular-bias}. Then
\[
\sqrt{k_n}
\left(
\widehat\Sigma_{\mu,n}^{(k_n),\mathrm{orc}}-\Sigma_\mu
\right)
\rightsquigarrow Z,
\]
where \(Z\) is a centered symmetric Gaussian matrix supported on
\(\mu^\perp\times\mu^\perp\). Moreover, for all symmetric \(A,B\in\R^{d\times d}\),
\(\Cov\!\left[\tr(AZ),\tr(BZ)\right]
=
\Cov\!\left[U_\mu^{\T}AU_\mu,U_\mu^{\T}BU_\mu\right]\).
Equivalently, given \(U_{\mu,1},\ldots,U_{\mu,m}\) independent copies of \(U_\mu\), \(Z\)
has the same law as the weak limit of
\[
m^{-1/2}\sum_{\ell=1}^m
\left(U_{\mu,\ell}U_{\mu,\ell}^{\T}-\Sigma_\mu\right).
\]
\end{theorem}

Theorem~\ref{thm:oracle-agca-clt} says that top-\(k_n\) selection is first-order free for the
bounded anchored second moment: the random radial threshold has the same first-order covariance
effect as the empirical tail count and cancels from the final covariance. This is the main
inferential simplification created by using the bounded spherical departure \(u_\mu(G)\).
If the second-order bias is not negligible and
\(\sqrt{k_n}(H_\mu(r_n)-\Sigma_\mu)\) converges along the relevant deterministic threshold
sequence, the same CLT has the corresponding deterministic bias shift.

Using the population projector and eigengap in
\eqref{eq:population-projector-eigengap}, let
\(\widehat\lambda_j^{\mathrm{orc}}\) and \(\widehat b_j^{\mathrm{orc}}\) denote the ordered
eigenvalues and associated eigenvectors of
\(\widehat\Sigma_{\mu,n}^{(k_n),\mathrm{orc}}\), and define the corresponding oracle estimators
\[
\widehat\tau_\mu^{\mathrm{orc}}:=\tr(\widehat\Sigma_{\mu,n}^{(k_n),\mathrm{orc}}),
\qquad
\widehat{\mathrm{AVE}}_{\mu,p}^{\mathrm{orc}}
:=
\frac{\sum_{j=1}^p\widehat\lambda_j^{\mathrm{orc}}}
{\widehat\tau_\mu^{\mathrm{orc}}},
\qquad
\widehat P_{\mu,p}^{\mathrm{orc}}
:=
\sum_{j=1}^p\widehat b_j^{\mathrm{orc}}
(\widehat b_j^{\mathrm{orc}})^{\T}.
\]

\begin{corollary}[Oracle scalar and spectral inference]
\label{cor:oracle-agca-spectral-clt}
Under the assumptions of Theorem~\ref{thm:oracle-agca-clt}, the following limits hold.
\begin{enumerate}
\item
\(\sqrt{k_n}\left(\widehat\tau_\mu^{\mathrm{orc}}-\tau_\mu\right)
\rightsquigarrow
N\!\left(0,\Var\!\left[\norm{U_\mu}_2^2\right]\right)\).
\item If \(\lambda_j\) is a simple eigenvalue, then
\[
\sqrt{k_n}\left(\widehat\lambda_j^{\mathrm{orc}}-\lambda_j\right)
\rightsquigarrow
N\!\left(0,\Var\!\left[(b_j^{\T}U_\mu)^2\right]\right).
\]
The corresponding limits are jointly Gaussian over any fixed set of simple eigenvalues, with
covariances
\(\Cov\!\left[(b_j^{\T}U_\mu)^2,(b_\ell^{\T}U_\mu)^2\right]\).
\item If \(1\le p\le d-2\), \(\tau_\mu>0\), and \(\Delta_{\mu,p}>0\), then
\[
\sqrt{k_n}
\left(
\widehat{\mathrm{AVE}}_{\mu,p}^{\mathrm{orc}}
-
\mathrm{AVE}_{\mu,p}
\right)
\rightsquigarrow
N(0,\sigma_{\mu,p}^2),
\]
where
\(\sigma_{\mu,p}^2
:=
\Var\!\left[\psi_{\mu,p}(U_\mu)\right]\) and
\(\psi_{\mu,p}(u)
:=
\tau_\mu^{-1}
\left(
u^{\T}P_{\mu,p}u
-
\mathrm{AVE}_{\mu,p}\norm{u}_2^2
\right)\).
\item If \(1\le p\le d-2\) and \(\Delta_{\mu,p}>0\), then
\[
\sqrt{k_n}
\left(\widehat P_{\mu,p}^{\mathrm{orc}}-P_{\mu,p}\right)
\rightsquigarrow
\mathcal D_{\mu,p}(Z),\quad \mathcal D_{\mu,p}(E)
:=
\sum_{j=1}^p\sum_{m=p+1}^{d-1}
\frac{b_m b_m^{\T} E b_j b_j^{\T}
+b_j b_j^{\T} E b_m b_m^{\T}}
{\lambda_j-\lambda_m},\quad E=E^{\T}.
\]
\end{enumerate}
\end{corollary}

The eigengap condition is an identifiability condition: if \(\lambda_p=\lambda_{p+1}\), the
rank-\(p\) space is not unique, and small gaps allow large rotations of individual eigenvectors.
Thus unstable loading vectors need not imply an unstable low-dimensional fitted subspace. For this
reason, comparisons across thresholds, bootstrap samples, or anchors use spectral projectors
whenever individual component orientation is not the target.

\begin{corollary}[Plug-in variance for oracle explained variation]
\label{cor:oracle-agca-plugin-ci}
Under the assumptions of Corollary~\ref{cor:oracle-agca-spectral-clt}, with \(1\le p\le d-2\),
\(\tau_\mu>0\), and \(\Delta_{\mu,p}>0\), define
\[
\widehat\psi_{\mu,p,i}^{\mathrm{orc}}
:=
(\widehat\tau_\mu^{\mathrm{orc}})^{-1}
\left[
U_{\mu,i}^{\T}
\widehat P_{\mu,p}^{\mathrm{orc}}
U_{\mu,i}
-
\widehat{\mathrm{AVE}}_{\mu,p}^{\mathrm{orc}}
\norm{U_{\mu,i}}_2^2
\right]
\]
and let \(\widehat\sigma_{\mu,p}^{2,\mathrm{orc}}\) be the empirical variance of
\(\widehat\psi_{\mu,p,i}^{\mathrm{orc}}\) over \(i\in\mathcal I_n^{(k_n)}\). Then
\(\widehat\sigma_{\mu,p}^{2,\mathrm{orc}}
\to
\sigma_{\mu,p}^2\) in probability.
Consequently,
\[
\widehat{\mathrm{AVE}}_{\mu,p}^{\mathrm{orc}}
\pm
z_{1-\gamma/2}
\frac{\widehat\sigma_{\mu,p}^{\mathrm{orc}}}{\sqrt{k_n}}
\]
is an asymptotically valid oracle \((1-\gamma)\) confidence interval for
\(\mathrm{AVE}_{\mu,p}\), with \(z_q\) the \(q\)-quantile of \(N(0,1)\). The same plug-in principle estimates the asymptotic variances for
\(\widehat\tau_\mu^{\mathrm{orc}}\) and simple oracle eigenvalues.
\end{corollary}

\subsection{Rank standardization beyond consistency}
\label{sec:rank-pareto-clt-scale}

The preceding CLT applies to the oracle sample. The proof of
Theorem~\ref{thm:rank-consistency} shows that rank-Pareto marginal standardization is negligible
at the consistency scale, but this \(o_p(1)\) control does not imply distributional equivalence at
the \(k_n^{-1/2}\) scale. With \(h_\mu(g):=u_\mu(g)u_\mu(g)^{\T}\),
\[
\sqrt{k_n}
\left(
\widehat\Sigma_{\mu,n}^{(k_n),\mathrm{emp}}
-
\widehat\Sigma_{\mu,n}^{(k_n),\mathrm{orc}}
\right)
=
\frac1{\sqrt{k_n}}
\sum_{i\in\mathcal I_n^{(k_n)}}
\{h_\mu(\widehat G_i)-h_\mu(G_i)\}+
\frac1{\sqrt{k_n}}
\left(
\sum_{i\in\widehat{\mathcal I}_n^{(k_n)}}
-
\sum_{i\in\mathcal I_n^{(k_n)}}
\right)
h_\mu(\widehat G_i).
\]
The first term is the angular perturbation induced by estimating the marginal Pareto scales; the
second is the selection perturbation from recomputing the top-\(k_n\) set after rank
standardization. Weighted empirical-tail approximations indicate that both terms can enter at the
same \(k_n^{-1/2}\) scale as the oracle fluctuation. Thus the oracle interval in
Corollary~\ref{cor:oracle-agca-plugin-ci} is a formal interval for the oracle estimator, not
automatically for the rank-Pareto estimator.

A full rank-based CLT would add a Gaussian marginal-estimation correction to
Theorem~\ref{thm:oracle-agca-clt}, under Euclidean-polar regular variation
(Assumption~\ref{ass:erv}), second-order marginal and joint-tail conditions, and differentiability
of the limiting tail-measure functional under componentwise marginal scalings. Supplementary
Section~\ref{app:conditional-rank-margin-clt} records the corresponding conditional delta-method
theorem, but we do not derive primitive rank-tail empirical-process conditions here. Supplementary
Section~\ref{app:simulations-10d-oracle-coverage} illustrates the issue: oracle intervals are
reasonably calibrated with oracle margins, whereas applying the same formula after rank-Pareto
standardization is not uniformly reliable. The rank-standardized simulations and empirical
analyses below therefore use bootstrap resampling, threshold paths, anchor checks, and projector
distances as finite-sample stability diagnostics.

\subsection{Plug-in reconstruction for tail simulation}
\label{sec:estimation-plugin-simulation}

The population simulation results in Section~\ref{sec:population-reconstruction-simulation} use
the rank-\(p\) component space through its orthogonal projector. For an angular direction
\(g\in\Sphere^{d-1}_{+}\), the plug-in AGCA reconstruction uses the empirical projector
\(\widehat P_{\mu,p}\) and is
\[
\widehat{\widetilde g}_{\mu,p}
:=
\frac{a_\mu(g)\mu+\widehat P_{\mu,p}u_\mu(g)}
{\sqrt{a_\mu(g)^2+\norm{\widehat P_{\mu,p}u_\mu(g)}_2^2}}.
\]
Applied to the selected empirical directions \(\widehat G_i\), this gives reconstructed angular
profiles that can be combined with a fitted radial tail model to simulate extremes. For example,
one draws an independent Pareto radial multiplier and attaches it to
\(\widehat{\widetilde G}_{\mu,p,i}\), or to a positive post-projected version if positivity is
enforced as in Supplementary Section~\ref{app:positive-post-projection}.

This plug-in step is controlled by the preceding consistency and oracle inference results.
Theorem~\ref{thm:top-k-consistency} gives consistency of the anchored second-moment
matrix in the oracle case, and Theorem~\ref{thm:rank-consistency} gives the corresponding
consistency when margins are rank-standardized. These results also give consistency of the
eigenvalue residual-risk curve. Under an eigengap, the leading spectral projector is a
continuous function of that matrix, so \(\widehat P_{\mu,p}\) is consistent; in the oracle case,
Corollary~\ref{cor:oracle-agca-spectral-clt}
also gives its first-order distribution. The reconstruction map above is continuous in the
projector on the anchored hemisphere. Remaining errors, such as radial tail-index estimation,
threshold selection, and Monte Carlo error, belong to the radial EVT and simulation layers rather
than to the AGCA eigensolution.

%% file: 04_empirics_portfolios.tex
\section{Empirical study: portfolio-loss extremes}
\label{sec:empirics-portfolios}

The empirical study examines whether daily equity-portfolio loss extremes have a low-dimensional
angular structure and whether that structure is useful for portfolio tail summaries. The object
is the angular law of tail directions obtained after marginal rank-Pareto standardization. Each
extreme day is represented by a radius and a direction on the positive sphere; AGCA is then fit
to the largest radii. A low-rank fit means that, conditional on a large joint portfolio-loss
event, the cross-sectional pattern of standardized losses is well approximated by a small number
of anchored geodesic coordinates. Supplementary Section~\ref{app:empirics-portfolios} gives
data-construction details and secondary diagnostics.

\subsection{Data and tail construction}
\label{sec:emp-port-data}

The data set consists of daily value-weighted Fama--French \(2\times 3\) portfolios
\citep{kenFrenchDataLibrary}. We use four size-based bivariate sorts: Size--Book-to-Market,
Size--Operating Profitability, Size--Investment, and Size--Momentum. Each sort contributes six
portfolios, giving \(d=24\) daily return series. We interpret losses as negative returns, transform
each loss margin to the pseudo-Pareto scale, and form Euclidean radii and directions. The main
AGCA fit uses the largest \(5\%\) radii from the complete daily panel beginning on July 2, 1973.
A modern anomaly-sorted portfolio panel from Open Source Asset Pricing
\citep{openSourceAssetPricing} is used as a robustness check and is reported in
Supplementary Section~\ref{app:emp-port-osap-robustness}.

The analysis uses the canonical anchor \(\mu_0=d^{-1/2}\one\), the balanced
complete-dependence direction. Loadings are therefore departures from a day on which all
standardized portfolio losses participate equally in the extreme direction. This fixed anchor is
used for the main interpretation; the Supplementary Material contains data-adaptive anchor checks.
Whenever shaded bootstrap 95\% diagnostic intervals (DIs) are shown below, they are conditional
stability summaries: after the rank-Pareto transform and the top-\(5\%\) radial sample are fixed,
we resample the selected directions with replacement, refit AGCA, and recompute the displayed
summaries. These DIs do not
include uncertainty from marginal rank estimation, threshold selection, or serial dependence in
daily losses; they are used to assess stability of the fitted angular cloud, while threshold and
anchor sensitivity are checked separately.

\subsection{Effective dimension and portfolio tail summaries}
\label{sec:emp-port-dimension-functionals}

The portfolio-functional check follows Section~\ref{sec:population-portfolio-var}. We use the
capped portfolio excess in \eqref{eq:population-capped-portfolio-excess-score}, evaluated as
\(f_{w,2\norm{w}_2,\norm{w}_2}\), and the normalized Pareto-limit VaR summary in
Corollary~\ref{cor:population-portfolio-var}. With rank-Pareto margins the empirical Pareto index
is \(\alpha=1\), so the VaR comparison is reported as the relative normalized VaR error
\(\exp\{|\log \widetilde C_{w,p}-\log C_w|\}-1\). We evaluate both summaries on the equal-weight
portfolio, all block-equal portfolios, 250 random long-only portfolios, and 250 random
limited-leverage portfolios with \(\norm{w}_1\le1.5\). The rank-level curves average over this
fixed evaluation set; Supplementary Figure~\ref{fig:app-emp-port-ff-class} reports the
class-level breakdown.

\begin{figure}[tbp]
\centering
\appinclude[width=0.95\linewidth]{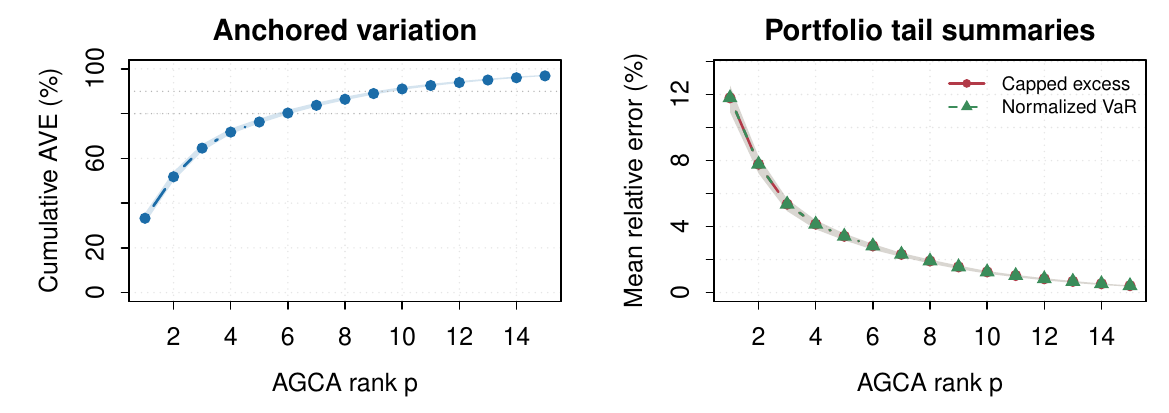}
\caption{Fama--French empirical summaries. Left: cumulative anchored variation explained at the
main \(5\%\) radial threshold. Right: average relative errors for capped excess and normalized
VaR over the portfolio evaluation class. Shaded bands are bootstrap 95\% diagnostic intervals, computed as pointwise \(2.5\%\) and \(97.5\%\) conditional bootstrap percentiles.}
\label{fig:emp-port-spectrum-functionals}
\end{figure}

Figure~\ref{fig:emp-port-spectrum-functionals} shows that the Fama--French portfolio tails have a
concentrated anchored spectrum. The first three AGCs explain \(64.6\%\) of anchored variation,
the first five explain \(76.3\%\), the first eight explain \(86.4\%\), and the first ten explain
\(91.1\%\). The same low-rank structure preserves the portfolio summaries in financially
interpretable units: rank five gives mean relative error \(3.4\%\) for both capped excess and
normalized VaR, rank eight reduces this to \(1.9\%\), rank ten to \(1.25\%\), and rank twelve to
\(0.84\%\). A \(1.25\%\) normalized VaR error means that the AGCA approximation changes the
Pareto-limit VaR level by about \(1.25\%\). The bootstrap DIs are narrow; for rank ten, the
pointwise DI is approximately \(1.14\%\)--\(1.30\%\). Supplementary
Figures~\ref{fig:app-emp-port-anchor} and~\ref{fig:app-emp-port-threshold} show that the
effective-dimension conclusion is stable under data-adaptive anchor choices and over tail
fractions from \(0.5\%\) to \(10\%\).

\subsection{Loading structure}
\label{sec:emp-port-loadings}

\begin{figure}[tbp]
\centering
\appinclude[width=0.95\linewidth]{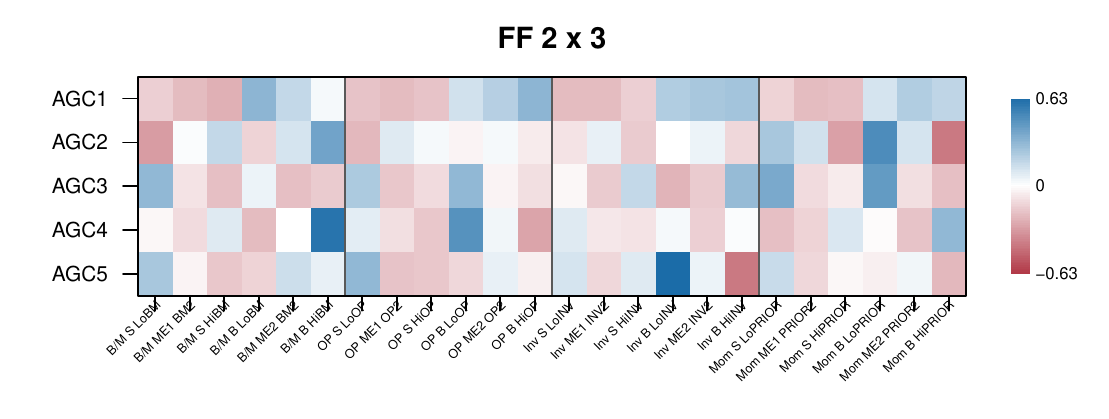}
\caption{Canonical-anchor AGCA loadings for the Fama--French panel. Columns are portfolios,
grouped by bivariate sort, and rows are the first five AGCA loading vectors. Red and blue denote
opposite signed departures from the balanced complete-dependence anchor. Loading signs are
oriented deterministically so that the largest absolute entry in each loading vector is positive.}
\label{fig:emp-port-ff-loadings}
\end{figure}

Figure~\ref{fig:emp-port-ff-loadings} gives the main loading display. The horizontal axis is the
portfolio universe and the rows are AGCA coordinates. This orientation emphasizes that a loading
is a cross-sectional contrast: it identifies which sorted portfolios are above or below the
balanced extreme direction when the corresponding AGCA score is positive. The first few loadings
are not single-asset spikes. They are structured contrasts across the four economically related
\(2\times 3\) blocks, so the low-rank approximation is not simply discarding all but a few
portfolios. Rather, AGCA is capturing repeated cross-sectional patterns in the directions of
joint portfolio losses. In AGC1, the main pattern is a size contrast repeated across the
bivariate sorts: the small-size portfolios tend to load negatively, while the large-size
portfolios tend to load positively. A positive AGC1 score therefore corresponds to an extreme
loss direction tilted toward relatively larger standardized losses among the large-stock
portfolios, after centering at balanced complete dependence.

The sign of a loading is arbitrary up to orientation. The canonical anchor fixes the
interpretation of the coordinate origin, while the deterministic sign convention fixes display
orientation. Data-driven anchors give similar scalar dimension diagnostics, but their loading
vectors are less directly interpretable because the coordinate origin is then a sample-dependent
tail direction rather than the balanced complete-dependence direction. Consequently, the
empirical content is in the blockwise relative signs and magnitudes.
Supplementary Figure~\ref{fig:app-emp-port-agc1-bootstrap-loadings} adds pointwise conditional
bootstrap DIs for AGC1; most displayed signs are stable in this diagnostic sense, with the
least stable entries concentrated near zero.

Because the unconstrained anchored subsphere may leave the positive orthant, the Supplementary
Material also reports the positive post-projection. In this application its effect on the scalar
functionals is small: the raw and post-projected curves are close, and the post-projection does
not change the rank-level conclusion. Supplementary
Section~\ref{app:emp-port-osap-robustness} reports the anomaly-panel robustness exercise.

The empirical conclusion is therefore twofold. First, the directions of daily portfolio-loss
extremes have a concentrated angular structure: in the \(d=24\) panel, ten AGCs explain about
\(91\%\) of anchored variation. Second, this angular reduction has practical content beyond the
descriptive spectrum. It preserves bounded capped-excess summaries and normalized VaR tail
constants across a broad set of portfolio weights, giving a direct empirical role for the
reconstruction theory in Section~\ref{sec:population-reconstruction-simulation}.

\FloatBarrier

%% file: 05_discussion.tex
\section{Discussion}
\label{sec:discussion}

Relative to the extremal-PCA \citep{Drees2021principal,drees2025asymptotic} and spherical-PCA
literatures \citep{fletcher2004principal,jung2012analysis}, AGCA occupies a deliberate niche. It
trades some geometric flexibility for a fixed, interpretable benchmark and for an oracle CLT based
on bounded angular departures. The benchmark makes scores and loadings read as departures from
balanced complete dependence, and the bounded tangent geometry makes the oracle AGCA CLT depend
on angular convergence and a second-order angular-bias condition rather than radial regular
variation or tail-moment assumptions. These features are hard to obtain jointly in the more
flexible manifold or multivariate regular variation approaches.
The restriction can be less parsimonious when the angular law has
curved geometry, strong face structure, or several active regimes, but it gives AGCA its
benchmark-relative interpretation, exact spectral solution, and clean inferential target.

%% file: appendix.tex
\section{Complementary theory}
\label{app:additional-theory}
\label{app:complementary-theory}

This supplementary section collects theoretical material that supports the construction of
AGCA. It records the geometry of the anchored coordinates, studies face- and axis-supported
angular laws, and discusses the canonical anchor, data-adaptive anchors, and anchor-free
spherical variants.

\subsection{Geometry of the anchored coordinates}
\label{app:anchored-geometry}

This subsection collects auxiliary coordinate identities used in the proofs. It also records that
the anchored coordinates agree with the Riemannian log map to first order near the anchor, while
\(u_\mu(g)\) remains bounded near the edge of the hemisphere.

\begin{lemma}[Positive directions lie in every interior hemisphere]
\label{lem:app-positive-hemisphere}
If \(\mu\in\Sphere^{d-1}_{++}\), then
\[
\Sphere^{d-1}_{+}\subset\mathcal H_\mu,
\qquad
\mathcal H_\mu:=\{g\in\Sphere^{d-1}:\ip{\mu}{g}>0\}.
\]
\end{lemma}

\begin{proof}
For \(g\in\Sphere^{d-1}_{+}\), all coordinates are nonnegative and at least one coordinate is
positive. Since every coordinate of \(\mu\) is positive, \(\ip{\mu}{g}>0\).
\end{proof}

The lemma is the reason the anchored coordinates are globally valid on the positive angular sample
space. For
\(g\in\mathcal H_\mu\), recall
\[
a_\mu(g):=\ip{\mu}{g},\qquad
u_\mu(g):=(I-\mu\mu^{\T})g,\qquad
v_\mu(g):=g/\ip{\mu}{g}-\mu.
\]
The inverse gnomonic map is
\[
\phi_\mu(v):=\frac{\mu+v}{(1+\norm{v}_2^2)^{1/2}},\qquad v\in\mu^\perp.
\]

\begin{proposition}[Anchored coordinate and membership identities]
\label{prop:app-coordinate-formulas}
For every \(g\in\Sphere^{d-1}_{+}\),
\[
g=\phi_\mu(v_\mu(g)),
\qquad
d_g(g,\mu)
=
\arccos(a_\mu(g))
=
\arcsin\norm{u_\mu(g)}_2
=
\arctan\norm{v_\mu(g)}_2.
\]
For any linear \(W\subset\mu^\perp\),
\[
M_\mu(W)=\phi_\mu(W),
\]
where \(\phi_\mu(W):=\{\phi_\mu(w):w\in W\}\). Consequently,
\[
g\in M_\mu(W)
\iff
v_\mu(g)\in W
\iff
u_\mu(g)\in W.
\]
\end{proposition}

\begin{proof}
For \(g\in\Sphere^{d-1}_{+}\), Lemma~\ref{lem:app-positive-hemisphere} gives \(a_\mu(g)>0\), and
\[
g=a_\mu(g)\mu+u_\mu(g),\qquad
a_\mu(g)^2+\norm{u_\mu(g)}_2^2=1.
\]
Since \(d_g(g,\mu)=\arccos(a_\mu(g))\), this gives the first three distance identities. The last
one follows from \(v_\mu(g)=u_\mu(g)/a_\mu(g)\) and
\(\tan(d_g(g,\mu))=\norm{u_\mu(g)}_2/a_\mu(g)\).
Since \(v_\mu(g)=u_\mu(g)/a_\mu(g)\), the identity \(g=\phi_\mu(v_\mu(g))\) follows from
\[
\phi_\mu(v_\mu(g))
=
\frac{\mu+u_\mu(g)/a_\mu(g)}
{(1+\norm{u_\mu(g)}_2^2/a_\mu(g)^2)^{1/2}}
=
a_\mu(g)\mu+u_\mu(g)
=g.
\]
If \(w\in W\), then \(\phi_\mu(w)\in\Sphere^{d-1}\), has positive inner product with \(\mu\),
and belongs to \(\Span(\{\mu\}\cup W)\), so \(\phi_\mu(w)\in M_\mu(W)\). Conversely, if
\(h\in M_\mu(W)\), then \(h=a(\mu+w)\) for some \(a>0\) and \(w\in W\). Since
\(h\in\Sphere^{d-1}\) and \(w\perp\mu\), \(a=(1+\norm{w}_2^2)^{-1/2}\), so
\(h=\phi_\mu(w)\). Hence \(M_\mu(W)=\phi_\mu(W)\). The membership equivalences follow because
\(g\in M_\mu(W)\) if and only if \(v_\mu(g)\in W\), and
\(u_\mu(g)=a_\mu(g)v_\mu(g)\).
\end{proof}

These identities are the coordinate facts used by the AGCA criterion. The membership equivalence
says that anchored great subspheres become linear subspaces in either \(u_\mu\)- or
\(v_\mu\)-coordinates. Proposition~\ref{prop:projection-identity} gives the additional residual
identity that turns spherical reconstruction into the spectral second-moment problem in the main
text.

\begin{proposition}[Log-map and gnomonic coordinate formulas]
\label{prop:app-log-gnomonic-formulas}
For \(v\in\mu^\perp\),
\[
a_\mu(\phi_\mu(v))=\frac{1}{(1+\norm{v}_2^2)^{1/2}},
\qquad
d_g(\phi_\mu(v),\mu)=\arctan\norm{v}_2.
\]
For \(g\in\Sphere^{d-1}_{+}\), let \(\Logmap_\mu(g)\) denote the Riemannian logarithm map on the
unit sphere at \(\mu\). It satisfies
\[
\Logmap_\mu(g)
=
\frac{\arctan\norm{v_\mu(g)}_2}{\norm{v_\mu(g)}_2}\,v_\mu(g)
=
\frac{\arcsin\norm{u_\mu(g)}_2}{\norm{u_\mu(g)}_2}\,u_\mu(g),
\]
with the prefactors interpreted by continuity as one at \(g=\mu\). Hence
\[
\Logmap_\mu(g)=v_\mu(g)+O\{\norm{v_\mu(g)}_2^3\}
=u_\mu(g)+O\{\norm{u_\mu(g)}_2^3\}
\quad\text{as }g\to\mu.
\]
Finally, for any linear \(W\subset\mu^\perp\),
\[
d_g(g,M_\mu(W))
=
\arctan\!\left[
\frac{\norm{v_\mu(g)-\Proj_Wv_\mu(g)}_2}
{(1+\norm{\Proj_Wv_\mu(g)}_2^2)^{1/2}}
\right].
\]
\end{proposition}

\begin{proof}
The identities for \(a_\mu(\phi_\mu(v))\) and \(d_g(\phi_\mu(v),\mu)\) follow from
\(\ip{\mu}{v}=0\), which gives
\[
\ip{\mu}{\phi_\mu(v)}
=
(1+\norm{v}_2^2)^{-1/2}.
\]
The distance formula follows because
\(\arccos((1+\norm{v}_2^2)^{-1/2})=\arctan\norm{v}_2\).
The log-map identity is the standard unit-sphere formula
\[
\Logmap_\mu(g)=\frac{\theta}{\sin\theta}(g-\cos(\theta)\mu),
\qquad
\theta=d_g(g,\mu),
\]
together with \(\cos\theta=a_\mu(g)\), \(\sin\theta=\norm{u_\mu(g)}_2\), and
Proposition~\ref{prop:app-coordinate-formulas}. The expansion follows from the Taylor series of
\(\arctan x/x\) and \(\arcsin x/x\).

For the gnomonic residual, write \(v=v_\mu(g)\), \(v_W=\Proj_Wv\), and
\(v_\perp=v-v_W\). Since \(u_\mu(g)=v/(1+\norm{v}_2^2)^{1/2}\) and
\(\Proj_Wu_\mu(g)=v_W/(1+\norm{v}_2^2)^{1/2}\),
Proposition~\ref{prop:projection-identity} gives
\[
\sin d_g(g,M_\mu(W))
=
\frac{\norm{v_\perp}_2}
{(1+\norm{v_W}_2^2+\norm{v_\perp}_2^2)^{1/2}}.
\]
Since \(d_g(g,M_\mu(W))<\pi/2\),
\[
d_g(g,M_\mu(W))
=
\arcsin\left(\frac{\norm{v_\perp}_2}
{(1+\norm{v_W}_2^2+\norm{v_\perp}_2^2)^{1/2}}\right).
\]
Using \(\arcsin x=\arctan(x/(1-x^2)^{1/2})\) gives the stated formula.
\end{proof}

These formulas place the anchored coordinates relative to standard spherical geometry. The
gnomonic map \(\phi_\mu\) gives a global coordinate map for the visible hemisphere, while the
log-map expansion explains why both anchored coordinates behave like usual tangent coordinates
near the anchor. Away from the anchor, the two coordinates differ: \(v_\mu(g)\) keeps the exact
gnomonic linearization, whereas \(u_\mu(g)\) remains bounded and is therefore the coordinate used
for the AGCA second moment.

The bounded loss \(\sin^2 d_g\) is chosen to preserve this exact projection geometry, not as a
generic replacement for squared geodesic distance. On the relevant hemisphere it is monotone in
the geodesic residual, comparable to \(d_g^2\), and has the same zero-risk sets. Its advantage is
that the residual becomes an ordinary Euclidean squared residual in the bounded coordinates
\(u_\mu(g)\), which is what leads to the eigensolution in the main text.

\begin{proposition}[Squared geodesic loss and bounded geodesic loss]
\label{prop:app-loss-comparability}
For any linear subspace \(W\subset\mu^\perp\) and any random \(G\in\Sphere^{d-1}_{+}\),
\[
\left(2/\pi\right)^2
d_g(G,M_\mu(W))^2
\le
\sin^2 d_g(G,M_\mu(W))
\le
d_g(G,M_\mu(W))^2
\]
almost surely. Consequently, zero risk under \(\sin^2 d_g\) is equivalent to exact geodesic
representation in \(M_\mu(W)\).
\end{proposition}

\begin{proof}
Because \(\mu\in M_\mu(W)\) and \(G\in\mathcal H_\mu\), \(d_g(G,M_\mu(W))<\pi/2\). The
inequality follows from \((2/\pi)t\le\sin t\le t\) on \([0,\pi/2]\). The zero-risk statement
follows because \(\sin t=0\) if and only if \(t=0\) on this interval.
\end{proof}

\subsection{Face and axis-supported angular laws}
\label{app:axis-supported-agca}

This subsection studies how AGCA behaves when some variables have asymptotically independent
extremes. At the angular level, such episodes may appear as mass on coordinate faces or near
coordinate axes: one subset of variables dominates the extreme direction while the remaining
coordinates are comparatively small. The goal here is to show what the AGCA population objects do
once such angular directions are part of the limiting law \(G\).

The results below make three simple points. First, face-supported angular laws enter AGCA through
bounded anchored departures, whose span determines the component rank. Second, if one variable
dominates the extreme direction, AGCA returns the corresponding contrast from the anchor toward
that variable, and this conclusion is stable when the mass is only near the axis. Third, if all
axes are represented symmetrically,
there is no preferred one-dimensional contrast, so the spectrum is tied rather than artificially
low-dimensional.

The first result makes precise the statement that coordinate faces do not cause a singularity for
AGCA. The bounded coordinate \(u_\mu(G)\) converts face-supported angular laws into ordinary
bounded tangent vectors, and the population component space is exactly the span of those tangent
departures. Thus asymptotically independent face support is compatible with the same
eigendecomposition used for jointly extreme angular laws.

\begin{proposition}[Anchored span of face-supported laws]
\label{prop:app-face-supported-span}
Let \(G\in\Sphere^{d-1}_{+}\), let \(U_\mu:=u_\mu(G)\), and let
\(\Sigma_\mu:=\E[U_\mu U_\mu^{\T}]\). Then
\(\operatorname{range}(\Sigma_\mu)
=
\Span\{u_\mu(g):g\in\operatorname{supp}(G)\}\).
Consequently, if \(G\) is supported on the coordinate face
\[
F_S:=\{g\in\Sphere^{d-1}_{+}:g_j=0\text{ for }j\notin S\},
\qquad S\subset\{1,\ldots,d\},
\]
then \(r_\mu=\rank(\Sigma_\mu)\le \min\{|S|,d-1\}\).
\end{proposition}

\begin{proof}
For any \(x\in\mu^\perp\),
\[
x^{\T}\Sigma_\mu x=\E[(x^{\T}U_\mu)^2].
\]
Hence \(x\in\ker(\Sigma_\mu)\) if and only if \(x^{\T}U_\mu=0\) almost surely. Since
\(u_\mu\) is continuous, this is equivalent to \(x\) being orthogonal to
\(u_\mu(g)\) for every \(g\in\operatorname{supp}(G)\). Thus
\(\ker(\Sigma_\mu)=\Span\{u_\mu(g):g\in\operatorname{supp}(G)\}^{\perp}\) inside
\(\mu^\perp\). Since \(\Sigma_\mu\) is symmetric on \(\mu^\perp\), this gives the range identity.
If \(G\) is supported on \(F_S\), then \(G\in E_S:=\Span\{e_j:j\in S\}\) almost surely, and
\(u_\mu(G)=(I-\mu\mu^{\T})G\) lies in \((I-\mu\mu^{\T})E_S\). This space has dimension at most
\(\min\{|S|,d-1\}\), so the same bound holds for \(\rank(\Sigma_\mu)\).
\end{proof}

The next result gives the population loading for the simplest asymptotically independent
episode: all angular mass is on one coordinate axis. In that case AGCA has a single nonzero
population component. With the canonical anchor, the loading has one positive coordinate for the
dominant variable and equal negative coordinates for the others, so it is exactly the contrast
from balanced co-movement toward that axis. The same conclusion is stable for laws concentrated
near, rather than exactly on, the axis.

\begin{proposition}[Axis loading and near-axis stability]
\label{prop:app-axis-loading}
Let \(e_j\) be the \(j\)th coordinate vector and let
\[
u_j:=u_\mu(e_j)=e_j-\mu_j\mu,\qquad \gamma_j:=\norm{u_j}_2^2=1-\mu_j^2.
\]
If \(G=e_j\) almost surely, then \(\Sigma_\mu=u_j u_j^{\T}\), \(r_\mu=1\),
and the unique rank-one loading, up to sign, is
\(b_j^\star:=u_j/\norm{u_j}_2\).
For the canonical anchor \(\mu_0=d^{-1/2}\one\),
\[
b_j^\star
=
\frac{e_j-d^{-1}\one}{\sqrt{1-1/d}}.
\]
More generally, suppose \(\E[\norm{G-e_j}_2^2]\le \delta^2\). Then
\[
\norm{\Sigma_\mu-u_j u_j^{\T}}_{\mathrm{op}}
\le
2\sqrt{\gamma_j}\,\delta+\delta^2.
\]
If \(\varepsilon_j:=2\sqrt{\gamma_j}\,\delta+\delta^2< \gamma_j/2\), the leading
rank-one population projector \(P_{\mu,1}\) of \(\Sigma_\mu\) is unique and satisfies
\[
\norm{P_{\mu,1}-b_j^\star b_j^{\star\T}}_{\mathrm{op}}
\le
\frac{2\varepsilon_j}{\gamma_j}.
\]
\end{proposition}

\begin{proof}
If \(G=e_j\) almost surely, then \(U_\mu=u_j\) almost surely, so
\(\Sigma_\mu=u_j u_j^{\T}\). Since \(\mu\in\Sphere^{d-1}_{++}\), \(\mu_j<1\) for \(d\ge2\), and
therefore \(\gamma_j>0\). The rank-one eigensolution and the canonical-anchor formula follow
immediately from \(u_{\mu_0}(e_j)=e_j-d^{-1}\one\).
For the near-axis statement, write \(V:=u_\mu(G)-u_j=(I-\mu\mu^{\T})(G-e_j)\). Since
\(I-\mu\mu^{\T}\) is an orthogonal projector,
\(\E[\norm{V}_2^2]\le\delta^2\). Also
\(\norm{\E V}_2\le\E\norm{V}_2\le\delta\). Expanding
\[
\Sigma_\mu
=
\E[(u_j+V)(u_j+V)^{\T}]
=
u_j u_j^{\T}+u_j(\E V)^{\T}+(\E V)u_j^{\T}+\E[V V^{\T}]
\]
gives
\[
\norm{\Sigma_\mu-u_j u_j^{\T}}_{\mathrm{op}}
\le
2\norm{u_j}_2\norm{\E V}_2+\E[\norm{V}_2^2]
\le
2\sqrt{\gamma_j}\,\delta+\delta^2.
\]
The rank-one projector bound is the Davis--Kahan bound applied to the perturbation of
\(u_j u_j^{\T}\), whose leading eigengap is \(\gamma_j\).
\end{proof}

The final result records an important negative case. AGCA can represent axis-supported angular
laws, but it does not force them to be low-dimensional. If a canonical-anchor angular law puts
equal mass on every coordinate axis, every axis contrast is equally important. The population
spectrum is then completely tied on \(\mu_0^\perp\), so individual loadings are not identified.
This is the appropriate diagnostic: there is finite axis-supported angular structure, but no
preferred one-dimensional component.

\begin{proposition}[Symmetric axis support has no preferred loading]
\label{prop:app-symmetric-axis-support}
Let \(\mu_0=d^{-1/2}\one\), and suppose
\(\Prob[G=e_j]=1/d\) for \(j=1,\ldots,d\). Then
\(\Sigma_{\mu_0}
=
\frac1d\left(I-\frac{\one\one^{\T}}{d}\right)\).
Hence all \(d-1\) nonzero eigenvalues are equal to \(1/d\), any orthonormal basis of
\(\mu_0^\perp\) is a population loading basis, and
\[
\mathrm{AVE}_{\mu_0,p}=\frac{p}{d-1},
\qquad p=0,\ldots,d-1.
\]
\end{proposition}

\begin{proof}
For the canonical anchor, \(u_{\mu_0}(e_j)=e_j-d^{-1}\one\). Therefore
\[
\Sigma_{\mu_0}
=
\frac1d\sum_{j=1}^d
(e_j-d^{-1}\one)(e_j-d^{-1}\one)^{\T}
=
\frac1d\left(I-\frac{\one\one^{\T}}{d}\right).
\]
The matrix in parentheses is the orthogonal projector onto \(\mu_0^\perp\). The eigenvalue and
anchored-variation statements follow directly.
\end{proof}

\subsection{Anchor choice and anchor-free variants}
\label{app:alternative-routes}

Once the anchor is fixed, Theorem~\ref{thm:population-eigensolution} identifies the optimal AGCA
component spaces by a single eigendecomposition. The substantive modeling choice is therefore
upstream: which benchmark should define the origin of the score system, and should one work in
the anchored great-subsphere family at all?

The anchor has two roles. First, it defines the zero-departure angular profile: a rank-zero
anchored approximation replaces every angular observation by \(\mu\). Second, it fixes the origin
of the component scores, so the loadings are contrasts describing how an extreme profile moves
away from that benchmark. A data-adaptive anchor may reduce reconstruction loss, but the fitted
contrasts then become relative to an estimated center or axis, which can weaken fixed-benchmark
interpretation and add anchor-estimation variability.

\begin{proposition}[Anchor-risk identities]
\label{prop:app-anchor-identities}
Let \(G\in\Sphere^{d-1}_{+}\), let \(M:=\E[GG^{\T}]\), and set
\(P_\mu:=I-\mu\mu^{\T}\). Then
\[
\Sigma_\mu=\E[u_\mu(G)u_\mu(G)^{\T}]=P_\mu M P_\mu
\]
on \(\mu^\perp\), and
\[
\tr(\Sigma_\mu)
=
1-\ip{\mu}{M\mu}
=
\mathcal R_\mu(\{0\})
=
\E[\sin^2 d_g(G,\mu)].
\]
\end{proposition}

\begin{proof}
Since \(u_\mu(G)=P_\mu G\), the first identity follows by taking expectations. Taking traces and
using \(\tr(M)=\E[\norm{G}_2^2]=1\) gives
\(\tr(\Sigma_\mu)=\tr(P_\mu M)=1-\ip{\mu}{M\mu}\). The remaining identities are the rank-zero case
of the risk decomposition in Theorem~\ref{thm:population-eigensolution}.
\end{proof}

These identities clarify what an empirical anchor can optimize. At rank zero, choosing an anchor
to minimize AGCA risk is equivalent to choosing \(\mu\) to maximize \(\ip{\mu}{M\mu}\), so the
anchor moves toward the dominant angular direction of the law.

\begin{proposition}[Anchor-free spherical eigensolution]
\label{prop:app-anchor-free-eigensolution}
Let \(G\in\Sphere^{d-1}_{+}\), let \(M:=\E[GG^{\T}]\), and let
\(L\subset\R^d\) be a linear subspace with \(\dim(L)=q\). Define the anchor-free bounded
geodesic risk
\[
\mathcal R_{\mathrm{af}}(L)
:=
\E[\sin^2 d_g(G,\Sphere^{d-1}\cap L)].
\]
If \(c_1,\ldots,c_d\) are orthonormal eigenvectors of \(M\), with eigenvalues
\(\alpha_1\ge\cdots\ge\alpha_d\), then
\(L_q^*:=\Span(c_1,\ldots,c_q)\)
minimizes \(\mathcal R_{\mathrm{af}}(L)\) over all \(q\)-dimensional linear subspaces of \(\R^d\), and
\(\mathcal R_{\mathrm{af}}(L_q^*)=\sum_{j=q+1}^d\alpha_j\).
\end{proposition}

\begin{proof}
For any \(g\in\Sphere^{d-1}\) and any linear subspace \(L\),
\[
\sin^2 d_g(g,\Sphere^{d-1}\cap L)
=
\norm{g-\Proj_L g}_2^2
=
1-\ip{g}{\Proj_Lg},
\]
where the projection onto the subsphere is \(\Proj_Lg/\norm{\Proj_Lg}_2\) when
\(\Proj_Lg\ne0\); the identity also holds by continuity if \(\Proj_Lg=0\). Taking expectations,
\[
\mathcal R_{\mathrm{af}}(L)=1-\tr(\Proj_LM).
\]
Thus minimizing \(\mathcal R_{\mathrm{af}}(L)\) is equivalent to maximizing \(\tr(\Proj_LM)\) over
\(q\)-dimensional \(L\). The Ky Fan maximum principle gives the span of the leading
\(q\) eigenvectors, and the minimized risk is the trailing eigensum.
\end{proof}

Proposition~\ref{prop:app-anchor-free-eigensolution} shows that the anchor-free problem is an
exact benchmark-free alternative, not an approximation to fixed-anchor AGCA. The \(q=1\) version
identifies the leading angular axis; when that axis is unique and can be oriented with positive
coordinates, it also gives a loss-matched self-anchored variant.

\begin{proposition}[Canonical anchor as maximal hemisphere margin]
\label{prop:app-canonical-anchor}
For \(\mu\in\Sphere^{d-1}_{++}\), define the hemisphere margin, the smallest
inner product between \(\mu\) and any positive-sphere direction, by
\[
m_+(\mu):=\inf_{g\in\Sphere^{d-1}_{+}}\ip{\mu}{g}.
\]
Then \(m_+(\mu)=\min_j\mu_j\), and
\(\max_{\mu\in\Sphere^{d-1}_{++}}m_+(\mu)=d^{-1/2}\),
with unique maximizer \(\mu_0\) in \eqref{eq:canonical-anchor}. Moreover, the same anchor
uniquely minimizes both the largest geodesic distance to the positive sphere and the largest
gnomonic coordinate norm:
\[
\argmin_{\mu\in\Sphere^{d-1}_{++}}
\sup_{g\in\Sphere^{d-1}_{+}}d_g(g,\mu)
=
\{\mu_0\},
\qquad
\argmin_{\mu\in\Sphere^{d-1}_{++}}
\sup_{g\in\Sphere^{d-1}_{+}}\norm{v_\mu(g)}_2
=
\{\mu_0\}.
\]
\end{proposition}

\begin{proof}
For any \(g\in\Sphere^{d-1}_{+}\),
\[
\ip{\mu}{g}
=\sum_{j=1}^d \mu_j g_j
\ge \left(\min_j\mu_j\right)\sum_{j=1}^d g_j
=\left(\min_j\mu_j\right)\norm{g}_1
\ge \min_j\mu_j,
\]
where the last inequality uses \(g\in\Sphere^{d-1}_{+}\), so
\(\norm{g}_1\ge\norm{g}_2=1\). Equality is attained at \(e_j\) for an index minimizing
\(\mu_j\). Hence \(m_+(\mu)=\min_j\mu_j\). If
\(m=\min_j\mu_j\), then \(1=\sum_{j=1}^d\mu_j^2\ge dm^2\), so
\(m\le d^{-1/2}\), with equality only when all coordinates are equal. The distance and gnomonic
claims follow because
\[
\sup_{g\in\Sphere^{d-1}_{+}}d_g(g,\mu)=\arccos(m_+(\mu)),\qquad
\sup_{g\in\Sphere^{d-1}_{+}}\norm{v_\mu(g)}_2
=
(m_+(\mu)^{-2}-1)^{1/2}.
\]
\end{proof}

The canonical anchor is therefore a nonrandom, permutation-symmetric benchmark with maximal
uniform angular margin over the positive sphere, not a sample center chosen to minimize empirical
loss. Loss-matched and spherical Fr\'echet anchors are useful sensitivity analyses when no
benchmark is privileged, but in samples they replace the fixed reference by an estimated one.

\section{Proofs}
\label{app:proofs}

This section gives the proofs of the mathematical statements in the main text, in the order in
which those statements appear.

\subsection{Anchored projection and population eigensolution}
\label{app:anchored-projection-population-eigensolution}

\begin{proof}[Proof of Proposition~\ref{prop:projection-identity}]
Let \(g\in\Sphere^{d-1}_{+}\) and write \(v:=v_\mu(g)\in\mu^\perp\). By
Proposition~\ref{prop:app-coordinate-formulas}, \(g=\phi_\mu(v)\) and
\(M_\mu(W)=\phi_\mu(W)\). Because \(\arccos\) is decreasing, geodesic projection onto
\(M_\mu(W)\) is equivalent to maximizing \(\ip{g}{\phi_\mu(w)}\) over \(w\in W\). Since
\(g=\phi_\mu(v)\),
\[
\ip{g}{\phi_\mu(w)}
=
\frac{1+\ip{v}{w}}
{(1+\norm{v}_2^2)^{1/2}(1+\norm{w}_2^2)^{1/2}}.
\]
Let \(v_W:=\Proj_Wv\). Since \(v-v_W\perp W\), \(\ip{v}{w}=\ip{v_W}{w}\). By
Cauchy--Schwarz in \(\R\times W\),
\[
1+\ip{v_W}{w}
\le
(1+\norm{v_W}_2^2)^{1/2}(1+\norm{w}_2^2)^{1/2},
\]
with equality if and only if \(w=v_W\). Dividing by the denominator in the preceding display
shows that \(\ip{g}{\phi_\mu(w)}\) is uniquely maximized over \(w\in W\) at
\(w=v_W=\Proj_Wv\). Hence the geodesic projection of \(g\) onto \(M_\mu(W)\) is
\[
\Proj_{M_\mu(W)}(g)=\phi_\mu(\Proj_Wv).
\]
At this maximizing point,
\[
\cos d_g(g,M_\mu(W))
=
\left(\frac{1+\norm{\Proj_Wv}_2^2}{1+\norm{v}_2^2}\right)^{1/2}.
\]
Therefore
\[
\sin^2 d_g(g,M_\mu(W))
=
\frac{\norm{v-\Proj_Wv}_2^2}{1+\norm{v}_2^2}.
\]
Since
\[
u_\mu(g)=\frac{v}{(1+\norm{v}_2^2)^{1/2}},
\qquad
\Proj_Wu_\mu(g)=\frac{\Proj_Wv}{(1+\norm{v}_2^2)^{1/2}},
\]
the identity \eqref{eq:projection-identity} follows.

Finally,
\[
\phi_\mu(\Proj_Wv)
=
\frac{\mu+\Proj_Wv}{(1+\norm{\Proj_Wv}_2^2)^{1/2}}.
\]
Substituting \(v=u_\mu(g)/a_\mu(g)\) gives
\[
\phi_\mu(\Proj_Wv)
=
\frac{a_\mu(g)\mu+\Proj_Wu_\mu(g)}
{(a_\mu(g)^2+\norm{\Proj_Wu_\mu(g)}_2^2)^{1/2}},
\]
which is \eqref{eq:projection-reconstruction}.
\end{proof}

\begin{proof}[Proof of Theorem~\ref{thm:population-eigensolution}]
By Proposition~\ref{prop:projection-identity},
\(\mathcal R_\mu(W)
=
\E[\norm{U_\mu-\Proj_WU_\mu}_2^2]\).
Expanding the squared norm and using symmetry and idempotence of \(\Proj_W\),
\[
\mathcal R_\mu(W)
=
\E[\norm{U_\mu}_2^2]-\E[\ip{U_\mu}{\Proj_WU_\mu}].
\]
The second term is
\[
\E[\ip{U_\mu}{\Proj_WU_\mu}]
=
\tr(\Proj_W\E[U_\mu U_\mu^{\T}])
=
\tr(\Proj_W\Sigma_\mu).
\]
Since \(\E[\norm{U_\mu}_2^2]=\tr(\Sigma_\mu)\), this proves
\(\mathcal R_\mu(W)
=
\tr((I-\Proj_W)\Sigma_\mu)\).

The first term \(\tr(\Sigma_\mu)\) does not depend on \(W\). Minimizing
\(\mathcal R_\mu(W)\) over \(p\)-dimensional subspaces is therefore equivalent to maximizing
\(\tr(\Proj_W\Sigma_\mu)\). By the Ky Fan maximum principle
\citep[see, e.g.,][]{hornJohnson2013matrix}, the maximum is attained by the span of the eigenvectors
corresponding to the largest \(p\) eigenvalues of \(\Sigma_\mu\).
For \(W_{\mu,p}:=\Span(b_1,\ldots,b_p)\),
\[
\mathcal R_\mu(W_{\mu,p})
=
\tr(\Sigma_\mu)-\sum_{j=1}^p\lambda_j
=
\sum_{j=p+1}^{d-1}\lambda_j.
\]
\end{proof}

\subsection{Rank, scores, and loading interpretation}
\label{app:population-rank-scores-loadings}

\begin{proof}[Proof of Corollary~\ref{cor:anchored-rank}]
By Proposition~\ref{prop:projection-identity}, \(\mathcal R_\mu(W)=0\) if and only if
\(U_\mu\in W\) almost surely. If this holds, then \(\operatorname{range}(\Sigma_\mu)\subset W\),
so \(\rank(\Sigma_\mu)\le\dim(W)\). Conversely, if \(y\) is orthogonal to
\(\operatorname{range}(\Sigma_\mu)\), then
\(0=\ip{y}{\Sigma_\mu y}=\E[\ip{y}{U_\mu}^2]\),
so \(\ip{y}{U_\mu}=0\) almost surely. Hence \(U_\mu\in\operatorname{range}(\Sigma_\mu)\) almost
surely, and any \(W\supset\operatorname{range}(\Sigma_\mu)\) has zero risk. At the minimal
dimension \(p=\rank(\Sigma_\mu)\), \(W\) must equal \(\operatorname{range}(\Sigma_\mu)\).
\end{proof}

\begin{proof}[Proof of Proposition~\ref{prop:loading-log-ratio}]
If \(\eta=0\), then \(c_\eta(t)\equiv\mu\), and the derivative is zero. Suppose \(\eta\neq0\)
and write \(r:=\norm{\eta}_2\). On the unit sphere,
\[
c_\eta(t)
=
\Expmap_\mu(t\eta)
=
\cos(tr)\mu+\sin(tr)\eta/r.
\]
Thus \(c_\eta(0)=\mu\) and \(c_\eta'(0)=\eta\). Since \(\mu_j,\mu_k>0\), the map
\(t\mapsto \log\{c_{\eta,j}(t)/c_{\eta,k}(t)\}\)
is differentiable at \(t=0\), and the chain rule gives
\[
\left.\frac{\mathrm d}{\mathrm dt}\right|_{t=0}
\log\frac{c_{\eta,j}(t)}{c_{\eta,k}(t)}
=
\frac{c_{\eta,j}'(0)}{c_{\eta,j}(0)}
-
\frac{c_{\eta,k}'(0)}{c_{\eta,k}(0)}
=
\frac{\eta_j}{\mu_j}-\frac{\eta_k}{\mu_k}.
\]
\end{proof}

\subsection{Low-rank reconstruction and tail simulation}
\label{app:population-reconstruction-tail-simulation}

\begin{proof}[Proof of Proposition~\ref{prop:population-agca-reconstruction}]
By definition, \(\widetilde G_{\mu,p}=\Proj_{M_\mu(W_{\mu,p})}(G)\). Applying
Proposition~\ref{prop:projection-identity} with \(g=G\) and \(W=W_{\mu,p}\) gives
\[
\sin^2 d_g(G,\widetilde G_{\mu,p})
=
\sin^2 d_g(G,M_\mu(W_{\mu,p}))
=
\norm{u_\mu(G)-\Proj_{W_{\mu,p}}u_\mu(G)}_2^2.
\]
Taking expectations and using \eqref{eq:population-rank-diagnostics} gives
\[
\E[\sin^2 d_g(G,\widetilde G_{\mu,p})]
=
\rho_{\mu,p}
=
\sum_{j=p+1}^{d-1}\lambda_j.
\]
Finally, let \(\theta=d_g(G,\widetilde G_{\mu,p})\). Since
\(\mu\in M_\mu(W_{\mu,p})\) and \(G\in\mathcal H_\mu\), the projection distance satisfies
\(\theta\le d_g(G,\mu)<\pi/2\). Hence
\[
\norm{G-\widetilde G_{\mu,p}}_2^2
=
2(1-\cos\theta)
=
\frac{2}{1+\cos\theta}\sin^2\theta
\le
2\sin^2\theta.
\]
Taking expectations gives the displayed equality for
\(\E[\norm{G-\widetilde G_{\mu,p}}_2^2]\), and the inequality follows from
\(\E[\sin^2 d_g(G,\widetilde G_{\mu,p})]=\rho_{\mu,p}\).
\end{proof}

\begin{proof}[Proof of Proposition~\ref{prop:population-tail-simulation-bound}]
Assumption~\ref{ass:erv} implies
\[
\mathcal L\!\left(X/r\,\middle|\,R>r\right)
\rightsquigarrow
\mathcal L(P_\alpha G),
\]
where \(P_\alpha\) is independent of \(G\) and
\(\Prob[P_\alpha>x]=x^{-\alpha}\), \(x\ge1\). Since bounded-Lipschitz distance metrizes weak
convergence on \(\R^d\), \(\varepsilon_r\to0\).

It remains to bound
\(d_{\mathrm{BL}}(\mathcal L(P_\alpha G),\mathcal L(P_\alpha\widetilde G_{\mu,p}))\). Let
\[
D_{\mu,p}
:=
\norm{u_\mu(G)-\Proj_{W_{\mu,p}}u_\mu(G)}_2.
\]
By Proposition~\ref{prop:population-agca-reconstruction},
\(\E[D_{\mu,p}^2]=\rho_{\mu,p}\). If \(f\) satisfies
\(\norm{f}_\infty\le1\) and \(\Lip(f)\le1\), then
\[
\abs{f(P_\alpha G)-f(P_\alpha\widetilde G_{\mu,p})}
\le
\min\{2,\ P_\alpha\norm{G-\widetilde G_{\mu,p}}_2\}.
\]
For \(0<\beta\le1\), \(\min\{2,z\}\le2^{1-\beta}z^\beta\) for \(z\ge0\). Also, because
\(d_g(G,\widetilde G_{\mu,p})\le d_g(G,\mu)<\pi/2\),
\[
\norm{G-\widetilde G_{\mu,p}}_2^2
=
2\{1-\cos d_g(G,\widetilde G_{\mu,p})\}
\le
2\sin^2 d_g(G,\widetilde G_{\mu,p})
=
2D_{\mu,p}^2.
\]
Hence
\[
\abs{\E[f(P_\alpha G)]-\E[f(P_\alpha\widetilde G_{\mu,p})]}
\le
2^{1-\beta/2}
\E[P_\alpha^\beta]\E[D_{\mu,p}^\beta].
\]
The independence of \(P_\alpha\) and \(G\) gives the factorization. Since \(\beta<\alpha\),
\[
\E[P_\alpha^\beta]
=
\frac{\alpha}{\alpha-\beta},
\]
and Lyapunov's inequality gives
\[
\E[D_{\mu,p}^\beta]
\le
\{\E[D_{\mu,p}^2]\}^{\beta/2}
=
\rho_{\mu,p}^{\beta/2}.
\]
Taking the supremum over the bounded-Lipschitz class yields
\[
d_{\mathrm{BL}}(\mathcal L(P_\alpha G),\mathcal L(P_\alpha\widetilde G_{\mu,p}))
\le
2^{1-\beta/2}\frac{\alpha}{\alpha-\beta}\rho_{\mu,p}^{\beta/2}.
\]
The stated inequality follows by the triangle inequality with \(\varepsilon_r\).
\end{proof}

\begin{proof}[Proof of Proposition~\ref{prop:population-homogeneous-tail-scores}]
Write \(W:=X/R\). Since \(h(W)\le H_h\), the event \(\{h(X)>rx\}\) is contained in
\(\{R>r\}\) whenever \(x\ge H_h\). Therefore
\[
\frac{\Prob[h(X)>rx]}{\Prob[R>r]}
=
\Prob[(R/r)h(W)>x\mid R>r],
\qquad x\ge H_h.
\]
By Assumption~\ref{ass:erv} and the continuous mapping theorem,
\[
(R/r)h(W)\mid R>r
\rightsquigarrow
P_\alpha h(G).
\]
Writing \(A_h:=h(G)\), independence of \(P_\alpha\) and \(G\) gives, for \(x\ge H_h\),
\[
\Prob[P_\alpha h(G)>x]
=
\E[\Prob[P_\alpha>x/A_h\mid A_h]]
=
x^{-\alpha}\E[A_h^\alpha],
\]
with the conditional probability interpreted as zero when \(A_h=0\). The limiting survival
function is continuous on \([H_h,\infty)\). Uniform convergence on compact subintervals follows
from weak convergence, and the remaining tail is controlled by tightness; hence
\(\eta_{r,h}\to0\).

Similarly, with \(\widetilde A_{h,p}:=h(\widetilde G_{\mu,p})\),
\[
\Prob[h(P_\alpha\widetilde G_{\mu,p})>x]
=
\E[\Prob[P_\alpha>x/\widetilde A_{h,p}\mid \widetilde A_{h,p}]]
=
x^{-\alpha}\E[\widetilde A_{h,p}^\alpha],
\qquad x\ge H_h.
\]

It remains to bound the angular constants. If \(0<\alpha\le1\), the map
\(a\mapsto a^\alpha\) is \(\alpha\)-Hölder on \([0,\infty)\), so
\[
\abs{C_h-\widetilde C_{h,p}}
\le
\E[\abs{h(G)-h(\widetilde G_{\mu,p})}^{\alpha}]
\le
L_h^\alpha
\{\E[\norm{G-\widetilde G_{\mu,p}}_2^2]\}^{\alpha/2}.
\]
If \(\alpha>1\), the map \(a\mapsto a^\alpha\) is Lipschitz on \([0,H_h]\) with constant
\(\alpha H_h^{\alpha-1}\), and hence
\[
\abs{C_h-\widetilde C_{h,p}}
\le
\alpha H_h^{\alpha-1}L_h
\E[\norm{G-\widetilde G_{\mu,p}}_2].
\]
Using Proposition~\ref{prop:population-agca-reconstruction} in both cases gives
\[
\abs{C_h-\widetilde C_{h,p}}
\le
K_{\alpha,h}(2\rho_{\mu,p})^{\gamma/2}.
\]
Finally, for \(x\ge H_h\),
\[
\abs{
\frac{\Prob[h(X)>rx]}{\Prob[R>r]}
-
\Prob[h(P_\alpha\widetilde G_{\mu,p})>x]
}
\le
\eta_{r,h}+x^{-\alpha}\abs{C_h-\widetilde C_{h,p}}.
\]
Taking the supremum over \(x\ge H_h\) gives the stated bound.
\end{proof}

\begin{proof}[Proof of Proposition~\ref{prop:population-tv-obstruction}]
By Proposition~\ref{prop:population-agca-reconstruction},
\(\widetilde G_{\mu,p}\in M_\mu(W_{\mu,p})\) almost surely, and
\(\norm{P_\alpha\widetilde G_{\mu,p}}_2=P_\alpha\ge1\). Therefore
\(\Prob[P_\alpha\widetilde G_{\mu,p}\in\mathcal C_{\mu,p}]=1\). Similarly,
\(P_\alpha G\in\mathcal C_{\mu,p}\) if and only if
\(G\in M_\mu(W_{\mu,p})\), because \(P_\alpha\ge1\) and
\((P_\alpha G)/\norm{P_\alpha G}_2=G\). Thus
\[
\Prob[P_\alpha G\in\mathcal C_{\mu,p}]
=
\Prob[G\in M_\mu(W_{\mu,p})].
\]
Taking \(A=\mathcal C_{\mu,p}\) in the supremum over Borel sets gives the lower bound.
\end{proof}

\subsection{Portfolio tail functionals}
\label{app:population-portfolio-tail-functionals}

\begin{proof}[Proof of Corollary~\ref{cor:population-capped-portfolio-excess}]
The map \(s\mapsto \min\{1,s_+\}\) is bounded by one and one-Lipschitz. Therefore
\[
\norm{f_{w,t,L}}_\infty\le1,
\qquad
\Lip(f_{w,t,L})
\le
\norm{w}_2/L
\le
\norm{w}_1/L
\le1.
\]
The claim follows by applying
Proposition~\ref{prop:population-tail-simulation-bound} with \(f=f_{w,t,L}\).
\end{proof}

\begin{proof}[Proof of Corollary~\ref{cor:population-portfolio-var}]
Apply Proposition~\ref{prop:population-homogeneous-tail-scores} to
\(h_w(z):=(w^{\T}z)_+\). This score is nonnegative and positively homogeneous. On
\(\Sphere^{d-1}\),
\[
h_w(s)\le\norm{w}_2\le\norm{w}_1\le L,
\qquad
\abs{h_w(s)-h_w(t)}
\le
\abs{w^{\T}(s-t)}
\le
L\norm{s-t}_2.
\]
The tail identities for all \(x\ge L\) and the bound on
\(\abs{C_w-\widetilde C_{w,p}}\) therefore follow from the proposition, with
\[
K_{\alpha,h_w}
\le
\begin{cases}
L^\alpha, & 0<\alpha\le1,\\
\alpha L^\alpha, & \alpha>1.
\end{cases}
\]
The VaR formulas follow by solving
\(x^{-\alpha}C=u\); the condition
\(u<L^{-\alpha}\min(C_w,\widetilde C_{w,p})\) ensures that both resulting quantiles are at least
\(L\), where the tail identities apply. The ratio identity is immediate from the formulas. Finally,
if \(\min(C_w,\widetilde C_{w,p})\ge c\), then
\[
\abs{\log q_w(u)-\log\widetilde q_{w,p}(u)}
=
\alpha^{-1}\abs{\log C_w-\log\widetilde C_{w,p}}
\le
(\alpha c)^{-1}\abs{C_w-\widetilde C_{w,p}},
\]
and the angular-constant bound gives the stated logarithmic VaR bound.
\end{proof}

\subsection{Positive post-projection for tail simulation}
\label{app:positive-post-projection}

The population reconstruction \(\widetilde G_{\mu,p}\), defined in
\eqref{eq:population-reconstructed-angular}, is a unit vector in the fitted anchored
great-subsphere model
\[
M_\mu(W_{\mu,p})
=
\Span(\{\mu\}\cup W_{\mu,p})\cap\mathcal H_\mu,
\]
but this model itself need not be contained in the positive orthant. For simulation on the
positive cone, one can enforce positivity after the AGCA reconstruction by applying a
deterministic positive post-projection. This correction preserves the low-rank AGCA step as the
primary approximation, but it no longer preserves exact membership in \(M_\mu(W_{\mu,p})\). The
point of the following results is to quantify precisely the additional price paid for this
positivity correction.

For \(z\in\R^d\), let \(z_+\) denote the coordinatewise positive part. When
\(z_+\ne0\), define \(\mathcal P_+(z):=z_+/\norm{z_+}_2\).
Since \(\widetilde G_{\mu,p}\in\mathcal H_\mu\) and \(\mu\in\Sphere^{d-1}_{++}\),
\((\widetilde G_{\mu,p})_+\ne0\) almost surely. Indeed, if all coordinates of
\(\widetilde G_{\mu,p}\) were nonpositive, then
\(\ip{\mu}{\widetilde G_{\mu,p}}\le0\), contradicting
\(\widetilde G_{\mu,p}\in\mathcal H_\mu\). Thus the positive post-projected reconstruction
\[
\widetilde G_{\mu,p}^+
:=
\mathcal P_+(\widetilde G_{\mu,p})
=
(\widetilde G_{\mu,p})_+/\norm{(\widetilde G_{\mu,p})_+}_2
\]
is well defined and takes values in \(\Sphere^{d-1}_{+}\). Define the positivity-correction cost
\[
\delta_{\mu,p}^+
:=
\E\left[\norm{\widetilde G_{\mu,p}-\widetilde G_{\mu,p}^+}_2^2\right],
\qquad
\omega_{\mu,p}^+
:=
(2\rho_{\mu,p})^{1/2}+(\delta_{\mu,p}^+)^{1/2}.
\]
The quantity \(\omega_{\mu,p}^+\) is the \(L^2\) angular error budget for positive simulation:
the first term is the AGCA reconstruction error from
Proposition~\ref{prop:population-agca-reconstruction}, and the second term is the extra movement
needed to return the reconstructed direction to the positive sphere.

\begin{proposition}[Positive post-projected reconstruction]
\label{prop:app-positive-post-projection-reconstruction}
For \(p=0,\ldots,d-1\), \(\widetilde G_{\mu,p}^+\in\Sphere^{d-1}_{+}\) almost surely, and
\[
\E\left[\norm{G-\widetilde G_{\mu,p}^+}_2^2\right]
\le
(\omega_{\mu,p}^+)^2.
\]
\end{proposition}

\begin{proof}
The preceding paragraph proves that \(\widetilde G_{\mu,p}^+\) is well defined and belongs to
\(\Sphere^{d-1}_{+}\). By Minkowski's inequality,
\[
\sqrt{\E\left[\norm{G-\widetilde G_{\mu,p}^+}_2^2\right]}
\le
\sqrt{\E\left[\norm{G-\widetilde G_{\mu,p}}_2^2\right]}
+
\sqrt{\E\left[\norm{\widetilde G_{\mu,p}-\widetilde G_{\mu,p}^+}_2^2\right]}.
\]
The first term is bounded by \((2\rho_{\mu,p})^{1/2}\) by
Proposition~\ref{prop:population-agca-reconstruction}, and the second term is
\((\delta_{\mu,p}^+)^{1/2}\) by definition. Squaring gives the final display.
\end{proof}

\begin{proposition}[Positive post-projected tail simulation]
\label{prop:app-positive-tail-simulation-bound}
Suppose Assumption~\ref{ass:erv} holds, and let \(\varepsilon_r\) be defined by
\eqref{eq:population-tail-threshold-error}. Then \(\varepsilon_r\to0\). Moreover, for every
\(0<\beta\le1\) with \(\beta<\alpha\),
\[
d_{\mathrm{BL}}\!\left(
\mathcal L\!\left(X/r\,\middle|\, R>r\right),
\mathcal L(P_\alpha\widetilde G_{\mu,p}^+)
\right)
\le
\varepsilon_r
+
2^{1-\beta}\frac{\alpha}{\alpha-\beta}\,
(\omega_{\mu,p}^+)^\beta.
\]
\end{proposition}

\begin{proof}
As in the proof of Proposition~\ref{prop:population-tail-simulation-bound},
\(\varepsilon_r\to0\), and it remains to bound the distance between
\(\mathcal L(P_\alpha G)\) and \(\mathcal L(P_\alpha\widetilde G_{\mu,p}^+)\). If
\(\norm{f}_\infty\le1\) and \(\Lip(f)\le1\), then
\[
\abs{f(P_\alpha G)-f(P_\alpha\widetilde G_{\mu,p}^+)}
\le
\min\{2,\ P_\alpha\norm{G-\widetilde G_{\mu,p}^+}_2\}.
\]
For \(0<\beta\le1\), \(\min\{2,z\}\le2^{1-\beta}z^\beta\). Since \(P_\alpha\) is independent of
the angular variables and \(\beta<\alpha\),
\[
\abs{\E[f(P_\alpha G)]-\E[f(P_\alpha\widetilde G_{\mu,p}^+)]}
\le
2^{1-\beta}\E[P_\alpha^\beta]\E[\norm{G-\widetilde G_{\mu,p}^+}_2^\beta]
\le
2^{1-\beta}\frac{\alpha}{\alpha-\beta}(\omega_{\mu,p}^+)^\beta,
\]
where the last step uses Lyapunov's inequality and
Proposition~\ref{prop:app-positive-post-projection-reconstruction}. Taking the supremum over the
bounded-Lipschitz class and applying the triangle inequality with \(\varepsilon_r\) gives the
claim.
\end{proof}

\begin{proposition}[Positive post-projected homogeneous tail scores]
\label{prop:app-positive-homogeneous-tail-scores}
Suppose Assumption~\ref{ass:erv} holds, and let \(h:\R^d\to[0,\infty)\) satisfy the conditions
used in Proposition~\ref{prop:population-homogeneous-tail-scores}. Let \(C_h\), \(\eta_{r,h}\),
\(\gamma\), and \(K_{\alpha,h}\) be the corresponding quantities from that setup, and define
\(\widetilde C_{h,p}^+
:=
\E[h(\widetilde G_{\mu,p}^+)^\alpha]\).
Then \(\eta_{r,h}\to0\) as \(r\to\infty\). Moreover, for every \(x\ge H_h\),
\[
\Prob[h(P_\alpha\widetilde G_{\mu,p}^+)>x]
=
x^{-\alpha}\widetilde C_{h,p}^+,\qquad
\sup_{x\ge H_h}
\abs{
\frac{\Prob[h(X)>rx]}{\Prob[R>r]}
-
\Prob[h(P_\alpha\widetilde G_{\mu,p}^+)>x]
}
\le
\eta_{r,h}
+
H_h^{-\alpha}K_{\alpha,h}(\omega_{\mu,p}^+)^\gamma.
\]
\end{proposition}

\begin{proof}
The convergence \(\eta_{r,h}\to0\) is the threshold part of
Proposition~\ref{prop:population-homogeneous-tail-scores}. The tail identity follows from the
independence of \(P_\alpha\) and \(G\), since
\(\widetilde G_{\mu,p}^+\) is a measurable function of \(G\):
\[
\Prob[h(P_\alpha\widetilde G_{\mu,p}^+)>x]
=
\E\!\left[\Prob[P_\alpha>x/h(\widetilde G_{\mu,p}^+)\mid \widetilde G_{\mu,p}^+]\right]
=
x^{-\alpha}\E[h(\widetilde G_{\mu,p}^+)^\alpha],
\]
with the conditional probability interpreted as zero when \(h(\widetilde G_{\mu,p}^+)=0\).
It remains to compare the angular constants. If \(0<\alpha\le1\), then
\[
\abs{C_h-\widetilde C_{h,p}^+}
\le
L_h^\alpha
\E[\norm{G-\widetilde G_{\mu,p}^+}_2^\alpha]
\le
L_h^\alpha(\omega_{\mu,p}^+)^\alpha.
\]
If \(\alpha>1\), the map \(a\mapsto a^\alpha\) is Lipschitz on \([0,H_h]\) with constant
\(\alpha H_h^{\alpha-1}\), and hence
\[
\abs{C_h-\widetilde C_{h,p}^+}
\le
\alpha H_h^{\alpha-1}L_h
\E[\norm{G-\widetilde G_{\mu,p}^+}_2]
\le
\alpha H_h^{\alpha-1}L_h\,\omega_{\mu,p}^+.
\]
Together these bounds give
\[
\abs{C_h-\widetilde C_{h,p}^+}
\le
K_{\alpha,h}(\omega_{\mu,p}^+)^\gamma.
\]
The final display follows from the definition of \(\eta_{r,h}\), the tail identity above, and
\(x^{-\alpha}\le H_h^{-\alpha}\) for \(x\ge H_h\).
\end{proof}

The positive post-projection therefore solves the cone-membership problem without changing the
radial model or the threshold errors \(\varepsilon_r\) and \(\eta_{r,h}\). Its cost is explicit:
the angular terms in the main simulation bounds are replaced by the combined budget
\(\omega_{\mu,p}^+\). It does not, however, remove the support mismatch in total variation. The
simulator remains supported on a transformed low-dimensional set, so full-law approximation over
all Borel events still requires either a correctly supported angular law or an additional residual
model around the low-rank reconstruction.

\subsection{Empirical eigensolution and oracle-based consistency}
\label{app:empirical-eigensolution-oracle-consistency}

\begin{proof}[Proof of Proposition~\ref{prop:empirical-eigensolution}]
By Proposition~\ref{prop:projection-identity},
\[
\mathcal R_{\mu,\mathcal I}(W)
=
\frac1{\abs{\mathcal I}}\sum_{i\in\mathcal I}
\norm{u_\mu(g_i)-\Proj_Wu_\mu(g_i)}_2^2.
\]
Expanding the squared norm and using symmetry and idempotence of \(\Proj_W\) gives
\[
\mathcal R_{\mu,\mathcal I}(W)
=
\tr(\Sigma_{\mu,\mathcal I})
-
\tr(\Proj_W\Sigma_{\mu,\mathcal I})
=
\tr((I-\Proj_W)\Sigma_{\mu,\mathcal I}).
\]
The first term is fixed in \(W\), so minimizing the empirical risk over \(p\)-dimensional
subspaces is equivalent to maximizing
\(\tr(\Proj_W\Sigma_{\mu,\mathcal I})\). The Ky Fan maximum principle gives a
maximizer as the span of the leading \(p\) eigenvectors, and substituting that maximizer into the
risk identity yields the trailing eigensum.
\end{proof}

\begin{proof}[Proof of Theorem~\ref{thm:top-k-consistency}]
The proof uses a deterministic-threshold sandwich. The random top-\(k_n\) set is trapped, with
probability tending to one, between exceedance sets over two deterministic high thresholds whose
expected sizes are slightly below and above \(k_n\). Angular convergence for deterministic
threshold exceedances can then be transferred to the random top-\(k_n\) selection.

We first prove this transfer at the level of empirical angular averages. Specifically, for every
bounded continuous matrix-valued function \(h\) on \(\Sphere^{d-1}_{+}\), we show that
\[
\frac1{k_n}\sum_{i\in\mathcal I_n^{(k_n)}}h(G_i)
\to
\E[h(G)]
\qquad\text{in probability}.
\]
Taking \(h(g)=u_\mu(g)u_\mu(g)^{\T}\), which is bounded and continuous because \(u_\mu\) is
linear and \(g\in\Sphere^{d-1}_{+}\), gives the consistency of
\(\widehat\Sigma_{\mu,n}^{(k_n),\mathrm{orc}}\), and the remaining claims follow from standard
eigenvalue and spectral-projector perturbation bounds.

Let \(B\) bound \(\norm{h(g)}_{\mathrm{op}}\). For \(0<\varepsilon<1\), choose deterministic
thresholds \(u_n^-\) and \(u_n^+\) such that
\[
\Prob[R_1>u_n^-]=(1+\varepsilon)k_n/n,
\qquad
\Prob[R_1>u_n^+]=(1-\varepsilon)k_n/n.
\]
Because \(k_n/n\to0\), \(u_n^\pm\to\infty\). Let
\(\mathcal E_n^\pm:=\{i:R_i>u_n^\pm\}\). The binomial law of large numbers gives
\(|\mathcal E_n^-|/k_n\to1+\varepsilon\) and
\(|\mathcal E_n^+|/k_n\to1-\varepsilon\). With probability tending to one,
\(\mathcal E_n^+\subset\mathcal I_n^{(k_n)}\subset\mathcal E_n^-\).
By the angular convergence in Assumption~\ref{ass:angular-limit} and a triangular-array law of
large numbers for bounded variables,
\[
\frac1{|\mathcal E_n^\pm|}\sum_{i\in\mathcal E_n^\pm}h(G_i)
\to
\E[h(G)]
\qquad\text{in probability}.
\]
On the event \(\mathcal E_n^+\subset\mathcal I_n^{(k_n)}\subset\mathcal E_n^-\), write
\(A_I=k_n^{-1}\sum_{i\in\mathcal I_n^{(k_n)}}h(G_i)\) and
\(A_\pm=|\mathcal E_n^\pm|^{-1}\sum_{i\in\mathcal E_n^\pm}h(G_i)\). Since
\(\norm{h}_{\mathrm{op}}\le B\),
\[
\norm{A_I-A_-}_{\mathrm{op}}
\le
2B\,\frac{|\mathcal E_n^-|-k_n}{|\mathcal E_n^-|},
\qquad
\norm{A_I-A_+}_{\mathrm{op}}
\le
2B\,\frac{k_n-|\mathcal E_n^+|}{k_n}.
\]
The two fractions converge in probability to \(\varepsilon/(1+\varepsilon)\) and
\(\varepsilon\), respectively. Hence the top-\(k_n\) average has all its subsequential
probability limits within an \(O(\varepsilon)\) neighborhood of \(\E[h(G)]\).
Letting \(\varepsilon\downarrow0\) proves the claim.
\end{proof}

\subsection{Rank-based consistency}
\label{app:primitive-margin-conditions}

This subsection contains the rank-margin part of Section~\ref{sec:estimation-topk-margins}.
Proposition~\ref{prop:rank-margin-stability} first transfers average angular accuracy and
selected-set stability to covariance consistency. We then verify those two stability requirements
from primitive assumptions and prove Theorem~\ref{thm:rank-consistency}.

\begin{proposition}[Stability under marginal standardization error]
\label{prop:rank-margin-stability}
With the oracle and rank-Pareto quantities in \eqref{eq:sample-anchored-second-moments}, suppose
that
\[
\frac1{k_n}
\sum_{i\in\mathcal I_n^{(k_n)}\cup\widehat{\mathcal I}_n^{(k_n)}}
\norm{\widehat G_i-G_i}_2
\to0
\qquad\text{in probability}
\]
and
\(
\abs{\mathcal I_n^{(k_n)}\triangle\widehat{\mathcal I}_n^{(k_n)}}/k_n
\to0\)
in probability, where \(\triangle\) denotes symmetric difference.
Then
\[
\norm{
\widehat\Sigma_{\mu,n}^{(k_n),\mathrm{emp}}
-
\widehat\Sigma_{\mu,n}^{(k_n),\mathrm{orc}}
}_{\mathrm{op}}
\to0
\qquad\text{in probability}.
\]
Hence, whenever the oracle top-\(k_n\) estimator is consistent, the rank-based estimator is
consistent for the same \(\Sigma_\mu\).
\end{proposition}

\begin{proof}[Proof of Proposition~\ref{prop:rank-margin-stability}]
The map \(g\mapsto u_\mu(g)\) is nonexpansive and \(\norm{u_\mu(g)}_2\le1\). Hence
\[
\norm{u_\mu(\widehat g)u_\mu(\widehat g)^{\T}-u_\mu(g)u_\mu(g)^{\T}}_{\mathrm{op}}
\le
2\norm{\widehat g-g}_2.
\]
On the common selected set, the assumed average directional error gives
\[
\frac1{k_n}\sum_{i\in\mathcal I_n^{(k_n)}\cap\widehat{\mathcal I}_n^{(k_n)}}
\norm{u_\mu(\widehat G_i)u_\mu(\widehat G_i)^{\T}-u_\mu(G_i)u_\mu(G_i)^{\T}}_{\mathrm{op}}
\le
\frac2{k_n}
\sum_{i\in\mathcal I_n^{(k_n)}\cup\widehat{\mathcal I}_n^{(k_n)}}
\norm{\widehat G_i-G_i}_2
=o_p(1).
\]
The symmetric difference of the selected sets contributes at most
\(2\abs{\mathcal I_n^{(k_n)}\triangle\widehat{\mathcal I}_n^{(k_n)}}/k_n\),
because each summand has operator norm at most one. This proves the covariance transfer.
\end{proof}

We now verify the two stability requirements under primitive assumptions. They are handled
separately and under different assumptions. The average angular condition holds under continuity
of the margins alone
(Proposition~\ref{prop:average-angular}); the argument is a weighted Glivenko--Cantelli ratio
bound for the uniform empirical distribution function in the spirit of
\citet{wellner1978limit} and \citet[Chapter~10]{shorack1986empirical}. The selected-set
condition is where Assumption~\ref{ass:erv} enters, through an anti-concentration bound for the
radii near the top-\(k_n\) threshold (Proposition~\ref{prop:selected-set}).

Throughout this subsection the conditions of Theorem~\ref{thm:rank-consistency} are in
force: the margins \(F_1,\ldots,F_d\) and the distribution of \(R_1\) are continuous,
Assumption~\ref{ass:erv} holds, \(k_n\to\infty\), and \(k_n/n\to0\). Write
\(\bar F_R(r):=\Prob[R_1>r]\). Under Assumption~\ref{ass:erv}, \(\bar F_R\) is regularly
varying with index \(-\alpha\); beyond the angular limit already used in
Theorem~\ref{thm:top-k-consistency}, only this consequence of the assumption is used below.
With the standard Pareto margins of Section~\ref{sec:population-angular-anchor} the index is
necessarily \(\alpha=1\), since \(R_1\ge X_{11}\) gives \(\bar F_R(r)\ge r^{-1}\) while
\(R_1>r\) forces \(\max_jX_{1j}>r/\sqrt d\) and hence \(\bar F_R(r)\le d^{3/2}r^{-1}\); we
keep the general index \(\alpha\) in the statements. Because the margins are continuous, the
\(R_i\) are almost surely distinct, so \(\mathcal I_n^{(k_n)}\) is almost surely unique. The
\(\widehat R_i\) take finitely many values and can tie with positive probability; fix any
measurable tie-breaking rule under which \(|\widehat{\mathcal I}_n^{(k_n)}|=k_n\) and
\(\widehat{\mathcal I}_n^{(k_n)}\) contains every index with \(\widehat R_i\) strictly above
its \(k_n\)-th largest value. The results below hold for every such rule. Write
\(\mathcal U_n:=\mathcal I_n^{(k_n)}\cup\widehat{\mathcal I}_n^{(k_n)}\) for the union of the
oracle and rank-based selected sets.

We use an exact representation of the rank transform. Fix a margin \(j\) and set
\(V_{ij}:=1-F_j(Y_{ij})\), so that \(V_{1j},\ldots,V_{nj}\) are i.i.d.\ standard uniform and
\(X_{ij}=1/V_{ij}\), with large \(Y_{ij}\) corresponding to small \(V_{ij}\). Let
\[
\widehat F_{j,n}(v):=\frac1n\sum_{m=1}^n\mathbb I\{V_{mj}\le v\}
\]
be the empirical distribution function of these uniforms. Since ranks are almost surely unique,
\[
n+1-\operatorname{rank}(Y_{ij})
=1+\#\{m:Y_{mj}>Y_{ij}\}
=1+\#\{m:V_{mj}<V_{ij}\}
=n\widehat F_{j,n}(V_{ij}),
\]
and therefore, exactly,
\begin{equation}
\label{eq:rank-ratio}
\frac{\widehat X_{ij}}{X_{ij}}
=\frac{(n+1)\,V_{ij}}{n\,\widehat F_{j,n}(V_{ij})}.
\end{equation}
The componentwise relative error of the rank transform is thus the ratio of a uniform empirical
distribution function to the identity, evaluated at the observation's own uniform. The marginal
step controls this ratio above a slowly shrinking level.

\begin{lemma}[Ratio uniformity above a shrinking level]
\label{lem:ratio-gc}
Let \(V_1,\ldots,V_n\) be i.i.d.\ standard uniform with empirical distribution function
\(\widehat F_n\), and let \(c_n\in(0,1)\) satisfy \(nc_n\to\infty\). Then
\[
\sup_{c_n\le v\le1}\abs{\frac{\widehat F_n(v)}{v}-1}\to0
\qquad\text{in probability.}
\]
\end{lemma}

\begin{proof}
This is classical \citep[Chapter~10]{shorack1986empirical}; we include the short argument for
completeness. Fix \(\theta\in(0,1)\) and put \(v_\ell:=c_n(1+\theta)^\ell\wedge1\) for
\(\ell=0,1,\ldots,L_n\) with \(L_n:=\lceil\log(1/c_n)/\log(1+\theta)\rceil\), so that the grid
reaches \(1\). For \(v\in[v_\ell,v_{\ell+1}]\), monotonicity gives
\[
\frac{\widehat F_n(v)}{v}\le\frac{\widehat F_n(v_{\ell+1})}{v_\ell}
\le(1+\theta)\frac{\widehat F_n(v_{\ell+1})}{v_{\ell+1}},
\qquad
\frac{\widehat F_n(v)}{v}\ge\frac1{1+\theta}\,\frac{\widehat F_n(v_\ell)}{v_\ell},
\]
so if \(\abs{\widehat F_n(v_\ell)/v_\ell-1}\le\theta\) for every \(\ell\), the supremum in the
statement is at most \((1+\theta)^2-1\le3\theta\). At a fixed \(v\), \(n\widehat F_n(v)\) is
binomial with parameters \(n\) and \(v\), and Bernstein's inequality gives, for
\(0<\theta\le1\),
\[
\Prob\left[\abs{\widehat F_n(v)-v}>\theta v\right]
\le2\exp\left(-\frac{(\theta nv)^2}{2nv+\tfrac23\theta nv}\right)
\le2\exp\left(-\tfrac38\,\theta^2nv\right).
\]
With \(\beta_n:=\tfrac38\theta^2nc_n\to\infty\) and \((1+\theta)^\ell\ge1+\ell\theta\),
\[
\sum_{\ell=0}^{L_n}\Prob\left[\abs{\widehat F_n(v_\ell)/v_\ell-1}>\theta\right]
\le2\sum_{\ell\ge0}e^{-\beta_n(1+\ell\theta)}
=\frac{2e^{-\beta_n}}{1-e^{-\beta_n\theta}}\to0.
\]
Hence the supremum exceeds \(3\theta\) with probability tending to zero, for every
\(\theta\in(0,1)\).
\end{proof}

Fix the level \(c_n:=\sqrt{k_n}/n\), so that \(nc_n=\sqrt{k_n}\to\infty\) and
\(nc_n=o(k_n)\), and define the bad set
\[
J_n:=\left\{i\le n:\min_{1\le j\le d}V_{ij}<c_n\right\},
\]
the observations having at least one coordinate among the roughly \(\sqrt{k_n}\) most extreme
of its margin. Set
\[
\varepsilon_n^\ast
:=\max_{1\le j\le d}\ \sup_{c_n\le v\le1}
\abs{\frac{n\,\widehat F_{j,n}(v)}{(n+1)\,v}-1}.
\]
Since
\(n\widehat F_{j,n}(v)/((n+1)v)-1
=\tfrac n{n+1}\bigl(\widehat F_{j,n}(v)/v-1\bigr)-\tfrac1{n+1}\),
Lemma~\ref{lem:ratio-gc} applied to each margin and a union bound over the \(d\) margins give
\(\varepsilon_n^\ast\to0\) in probability; each margin's uniforms are i.i.d., and no
independence across margins is needed.

\begin{lemma}[Componentwise accuracy off the bad set]
\label{lem:componentwise}
On the event \(\{\varepsilon_n^\ast\le\tfrac12\}\), every \(i\notin J_n\) satisfies
\[
\max_{1\le j\le d}\abs{\frac{\widehat X_{ij}}{X_{ij}}-1}\le2\varepsilon_n^\ast,
\qquad
\abs{\widehat R_i-R_i}\le2\varepsilon_n^\ast R_i,
\qquad
\norm{\widehat G_i-G_i}_2\le4\varepsilon_n^\ast.
\]
Moreover \(\E\abs{J_n}\le d\sqrt{k_n}\), so \(\abs{J_n}/k_n\to0\) in probability.
\end{lemma}

\begin{proof}
For \(i\notin J_n\) and every \(j\), \(V_{ij}\ge c_n\), so by the definition of
\(\varepsilon_n^\ast\) the ratio \(n\widehat F_{j,n}(V_{ij})/((n+1)V_{ij})\) lies in
\([1-\varepsilon_n^\ast,1+\varepsilon_n^\ast]\). By \eqref{eq:rank-ratio},
\(\widehat X_{ij}/X_{ij}\) is the reciprocal of this ratio, so, using
\(\varepsilon_n^\ast\le\tfrac12\),
\[
\abs{\frac{\widehat X_{ij}}{X_{ij}}-1}
\le\frac{\varepsilon_n^\ast}{1-\varepsilon_n^\ast}
\le2\varepsilon_n^\ast.
\]
Consequently
\[
\norm{\widehat X_i-X_i}_2^2
=\sum_{j=1}^dX_{ij}^2\left(\frac{\widehat X_{ij}}{X_{ij}}-1\right)^{\!2}
\le\left(2\varepsilon_n^\ast\right)^2R_i^2,
\]
which gives \(\abs{\widehat R_i-R_i}\le\norm{\widehat X_i-X_i}_2\le2\varepsilon_n^\ast R_i\)
and, by the normalization inequality
\(\norm{a/\norm a_2-b/\norm b_2}_2\le2\norm{a-b}_2/\norm a_2\), the bound
\(\norm{\widehat G_i-G_i}_2\le4\varepsilon_n^\ast\). Finally
\(\E\abs{J_n}\le\sum_{j=1}^dn\,\Prob[V_{1j}<c_n]=dnc_n=d\sqrt{k_n}\), and Markov's inequality
gives \(\abs{J_n}/k_n\to0\) in probability since \(\sqrt{k_n}/k_n\to0\).
\end{proof}

\begin{proposition}[Average angular accuracy from continuity alone]
\label{prop:average-angular}
Let \(F_1,\ldots,F_d\) be continuous and \(k_n\to\infty\), \(k_n/n\to0\). Then
\[
\frac1{k_n}\sum_{i\in\mathcal U_n}\norm{\widehat G_i-G_i}_2\to0
\qquad\text{in probability.}
\]
No assumption on the joint tail is used.
\end{proposition}

\begin{proof}
\(\abs{\mathcal U_n}\le2k_n\), and \(\norm{\widehat G_i-G_i}_2\le2\) always, being a
difference of unit vectors. Splitting the sum over \(\mathcal U_n\) at \(J_n\) and applying
Lemma~\ref{lem:componentwise} on \(\{\varepsilon_n^\ast\le\tfrac12\}\),
\[
\frac1{k_n}\sum_{i\in\mathcal U_n}\norm{\widehat G_i-G_i}_2
\le\frac{2\abs{J_n}}{k_n}+\frac{2k_n}{k_n}\cdot4\varepsilon_n^\ast
=\frac{2\abs{J_n}}{k_n}+8\varepsilon_n^\ast
\to0
\qquad\text{in probability.}
\qedhere
\]
\end{proof}

The supremum analogue of the average angular condition fails. For the largest observation of
margin \(j\), \(\operatorname{rank}(Y_{ij})=n\) gives \(\widehat X_{ij}=n+1\), while
\(X_{ij}\) is the reciprocal of the minimal uniform \(V_{(1),j}\), and \(nV_{(1),j}\)
converges weakly to a standard exponential variable; hence \(\widehat X_{ij}/X_{ij}\)
converges weakly to that exponential variable rather than to \(1\). The rank transform thus
distorts the top few order statistics of each margin at constant order, and such observations
are typically selected. Under the level \(c_n=\sqrt{k_n}/n\) there are only
\(O_p(\sqrt{k_n})\) of them, which is why the averaged condition of
Proposition~\ref{prop:rank-margin-stability} absorbs them while a uniform condition cannot.
This is the reason Proposition~\ref{prop:rank-margin-stability} is stated in average form.

\begin{lemma}[Threshold location and band count]
\label{lem:radial-band}
Define \(r_n\) by \(\bar F_R(r_n)=k_n/n\), let \(t_n:=R_{n-k_n+1:n}\) be the \(k_n\)-th
largest radius, and for \(\delta\in(0,\tfrac12)\) let
\[
B_n(\delta):=\#\left\{i\le n:t_n(1-\delta)\le R_i<t_n\right\}.
\]
Then \(t_n/r_n\to1\) in probability, and for each fixed \(\delta\in(0,\tfrac12)\),
\[
\frac{B_n(\delta)}{k_n}\le(1-2\delta)^{-\alpha}-1+o_p(1).
\]
\end{lemma}

\begin{proof}
The first statement is the classical consistency of intermediate order statistics under
regular variation; see \citet[Chapter~2]{de2006extreme}. For the band count, let
\(N_n(r):=\#\{i\le n:R_i>r\}\). For fixed \(s>0\),
\(\E[N_n(r_ns)]/k_n=\bar F_R(r_ns)/\bar F_R(r_n)\to s^{-\alpha}\) by the uniform convergence
theorem for regularly varying functions \citep[Theorem~1.5.2]{bingham1987regular}, and
\(\Var[N_n(r_ns)]\le\E[N_n(r_ns)]\), so Chebyshev's inequality gives
\(N_n(r_ns)/k_n\to s^{-\alpha}\) in probability. Since the distribution of \(R_1\) is
continuous, the \(R_i\) are almost surely distinct, so \(\#\{i:R_i\ge t_n\}=k_n\). On the
event \(\{t_n\ge r_n(1-\delta)\}\), whose probability tends to one by the first statement,
every index counted by \(B_n(\delta)\) has
\(R_i\ge t_n(1-\delta)\ge r_n(1-\delta)^2\ge r_n(1-2\delta)\), and the same holds for the
\(k_n\) indices with \(R_i\ge t_n\); hence
\[
B_n(\delta)\le\#\{i:R_i\ge r_n(1-2\delta)\}-k_n
=N_n\bigl(r_n(1-2\delta)\bigr)-k_n
=k_n\left[(1-2\delta)^{-\alpha}-1+o_p(1)\right],
\]
where the middle equality holds almost surely because the threshold \(r_n(1-2\delta)\) is
deterministic and the distribution of \(R_1\) is continuous.
\end{proof}

\begin{proposition}[Selected-set stability under radial regular variation]
\label{prop:selected-set}
Under the conditions of Theorem~\ref{thm:rank-consistency},
\(\abs{\mathcal I_n^{(k_n)}\triangle\widehat{\mathcal I}_n^{(k_n)}}/k_n\to0\) in probability.
\end{proposition}

\begin{proof}
Both sets have cardinality \(k_n\), so
\(\abs{\mathcal I_n^{(k_n)}\triangle\widehat{\mathcal I}_n^{(k_n)}}
=2\abs{\mathcal I_n^{(k_n)}\setminus\widehat{\mathcal I}_n^{(k_n)}}\). Fix \(\eta>0\) and
choose \(\delta_0\in(0,\tfrac12)\) with \(2[(1-2\delta_0)^{-\alpha}-1]<\eta/2\). Since
\(\varepsilon_n^\ast\to0\) in probability, the event \(\{4\varepsilon_n^\ast\le\delta_0\}\)
has probability tending to one; work on this event, on which in particular
\(\varepsilon_n^\ast\le\tfrac12\), so Lemma~\ref{lem:componentwise} applies. Call \(i\)
\emph{good} if \(i\notin J_n\).

Every good \(i\in\mathcal I_n^{(k_n)}\) has \(R_i\ge t_n\) and hence
\(\widehat R_i\ge(1-2\varepsilon_n^\ast)R_i\ge(1-2\varepsilon_n^\ast)t_n\). Therefore
\[
\abs{\mathcal I_n^{(k_n)}\setminus\widehat{\mathcal I}_n^{(k_n)}}
\le\abs{J_n}
+\#\left(\{i:\widehat R_i\ge(1-2\varepsilon_n^\ast)t_n\}
\setminus\widehat{\mathcal I}_n^{(k_n)}\right).
\]
Because \(\widehat{\mathcal I}_n^{(k_n)}\) consists of \(k_n\) indices with the largest
\(\widehat R_i\), under any tie-breaking rule the number of indices at level
\(\ge(1-2\varepsilon_n^\ast)t_n\) that it fails to contain is at most
\(\widehat N_n-k_n\) when
\(\widehat N_n:=\#\{i:\widehat R_i\ge(1-2\varepsilon_n^\ast)t_n\}\) exceeds \(k_n\), and zero
otherwise. A good \(i\) with \(\widehat R_i\ge(1-2\varepsilon_n^\ast)t_n\) has
\[
R_i\ge\frac{\widehat R_i}{1+2\varepsilon_n^\ast}
\ge t_n\,\frac{1-2\varepsilon_n^\ast}{1+2\varepsilon_n^\ast}
\ge t_n(1-4\varepsilon_n^\ast),
\]
so, using \(\#\{i:R_i\ge t_n\}=k_n\) almost surely,
\[
\widehat N_n\le\abs{J_n}+\#\{i:R_i\ge t_n(1-4\varepsilon_n^\ast)\}
=\abs{J_n}+k_n+B_n(4\varepsilon_n^\ast).
\]
Combining the two displays,
\[
\abs{\mathcal I_n^{(k_n)}\triangle\widehat{\mathcal I}_n^{(k_n)}}
\le2\left(2\abs{J_n}+B_n(4\varepsilon_n^\ast)\right),
\]
and on the working event \(B_n(4\varepsilon_n^\ast)\le B_n(\delta_0)\) by monotonicity of the
band in its width. Lemma~\ref{lem:radial-band} and the choice of \(\delta_0\) give
\(\limsup_n\Prob[2B_n(4\varepsilon_n^\ast)/k_n>\eta/2]=0\), while
\(4\abs{J_n}/k_n\to0\) in probability by Lemma~\ref{lem:componentwise}. Hence
\(\Prob[\abs{\mathcal I_n^{(k_n)}\triangle\widehat{\mathcal I}_n^{(k_n)}}/k_n>\eta]\to0\).
\end{proof}

\begin{proof}[Proof of Theorem~\ref{thm:rank-consistency}]
Propositions~\ref{prop:average-angular} and~\ref{prop:selected-set} verify the two conditions
of Proposition~\ref{prop:rank-margin-stability}, which yields
\(\norm{\widehat\Sigma_{\mu,n}^{(k_n),\mathrm{emp}}
-\widehat\Sigma_{\mu,n}^{(k_n),\mathrm{orc}}}_{\mathrm{op}}\to0\) in probability. Since
Assumption~\ref{ass:erv} implies Assumption~\ref{ass:angular-limit},
Theorem~\ref{thm:top-k-consistency} gives
\(\norm{\widehat\Sigma_{\mu,n}^{(k_n),\mathrm{orc}}-\Sigma_\mu}_{\mathrm{op}}\to0\) in
probability, and the triangle inequality gives consistency of
\(\widehat\Sigma_{\mu,n}^{(k_n),\mathrm{emp}}\). The spectral consequences follow from the
same eigenvalue and spectral-projector perturbation bounds as in the proof of
Theorem~\ref{thm:top-k-consistency}.
\end{proof}

The radial step uses regular variation of \(\bar F_R\) only through the vanishing band ratio
\[
\lim_{\delta\downarrow0}\limsup_{r\to\infty}
\frac{\bar F_R(r(1-\delta))}{\bar F_R(r)}=1.
\]
Under Assumption~\ref{ass:angular-limit} alone, the standard Pareto margins give only the
two-sided bound \(r\bar F_R(r)\in[1,d^{3/2}]\), which permits \(\bar F_R\) to fall by a fixed
factor over relative bands of vanishing width; a fixed fraction of the top-\(k_n\) radii could
then sit in such a band, and a rank perturbation of vanishing relative size could reshuffle a
non-vanishing fraction of the selected set. We do not claim that the selected-set condition fails
under Assumption~\ref{ass:angular-limit}, only that its verification requires some radial
anti-concentration, for which Assumption~\ref{ass:erv} is the natural primitive condition.

The proof gives more than \(o_p(1)\): since \(g\mapsto u_\mu(g)u_\mu(g)^{\T}\) is bounded by
one in operator norm and \(2\)-Lipschitz on \(\Sphere^{d-1}_{+}\), the bounds above combine to
\[
\norm{\widehat\Sigma_{\mu,n}^{(k_n),\mathrm{emp}}
-\widehat\Sigma_{\mu,n}^{(k_n),\mathrm{orc}}}_{\mathrm{op}}
=O_p\left(\varepsilon_n^\ast+k_n^{-1/2}
+\frac{B_n(4\varepsilon_n^\ast)}{k_n}\right),
\]
with sharper weighted-approximation bounds for \(\varepsilon_n^\ast\), of order
\((nc_n)^{-1/2}=k_n^{-1/4}\) up to logarithmic factors, available from
\citet{shorack1986empirical}. We do not pursue rates here because the \(k_n^{-1/2}\) scale is
governed by the separate program of Section~\ref{app:conditional-rank-margin-clt}.

\subsection{Oracle CLT and plug-in inference}
\label{app:oracle-clt-plugin-inference}

\begin{lemma}[Oracle tail empirical process]
\label{lem:oracle-tail-empirical-process}
Let \(\mathcal H\) be the finite class consisting of the constant function
\(h_0(g):=1\) and the coordinate functions of
\(h_\mu(g):=u_\mu(g)u_\mu(g)^{\T}\).
For \(x\in[1/2,2]\), let \(r_n(x)\) be a deterministic threshold with
\(\bar F_R(r_n(x))=xk_n/n\), and define, for scalar \(\psi\in\mathcal H\),
\[
\mathbb G_n(\psi,x)
:=
\frac1{\sqrt{k_n}}\sum_{i=1}^n
\left\{
\psi(G_i)\mathbb I\{R_i>r_n(x)\}
-
\E[\psi(G_1)\mathbb I\{R_1>r_n(x)\}]
\right\}.
\]
Under the assumptions of Theorem~\ref{thm:top-k-consistency},
\(\mathbb G_n\rightsquigarrow W\) in
\(\ell^\infty(\mathcal H\times[1/2,2])\), where \(W\) is centered Gaussian with
covariance
\[
\Cov[W(\psi,x),W(\varphi,x')]
=
\min(x,x')\E[\psi(G)\varphi(G)].
\]
\end{lemma}

\begin{proof}
For fixed \((\psi,x)\) and \((\varphi,x')\), put \(y:=\min(x,x')\). Since the two
exceedance events intersect in \(\{R_1>r_n(y)\}\),
\[
\frac n{k_n}
\E[\psi(G_1)\varphi(G_1)\mathbb I\{R_1>r_n(y)\}]
=
y\,\E[\psi(G_1)\varphi(G_1)\mid R_1>r_n(y)]
\to
y\,\E[\psi(G)\varphi(G)].
\]
The product of the centering terms is \(O(k_n/n)\) after multiplication by \(n/k_n\).
Because all functions in \(\mathcal H\) are bounded, the row-wise Lindeberg condition is
immediate, and the finite-dimensional convergence follows from the Lindeberg--Feller theorem
and the Cramer--Wold device.

It remains only to control the index \(x\). The class
\(\{(r,g)\mapsto \psi(g)\mathbb I\{r>r_n(x)\}:x\in[1/2,2],\psi\in\mathcal H\}\)
has a bounded envelope after the \(n/k_n\) row scaling and is generated by a finite class of
bounded angular functions multiplied by a monotone one-dimensional threshold class. A regular
grid in \(x\) gives brackets whose squared \(L_2\)-size is proportional to the grid width,
uniformly in \(n\), and the envelope Lindeberg condition follows from \(k_n/n\to0\). The
triangular-array bracketing central limit theorem therefore gives tightness. The limiting
semimetric is proportional to \(\abs{x-x'}^{1/2}\) on each fixed angular coordinate, so the
limit admits continuous sample paths in \(x\).
\end{proof}

\begin{lemma}[Oracle random threshold expansion]
\label{lem:oracle-random-threshold-expansion}
Let \(R_{(1)}\ge\cdots\ge R_{(n)}\) be the descending order statistics of
\(R_1,\ldots,R_n\), and set
\(\widehat x_n:=\frac n{k_n}\bar F_R(R_{(k_n+1)})\).
Under the assumptions of Theorem~\ref{thm:top-k-consistency},
\(\widehat x_n\to1\) in probability and
\[
\sqrt{k_n}(\widehat x_n-1)
=
-\mathbb G_n(h_0,1)+o_p(1).
\]
Moreover, with probability tending to one, the oracle top-\(k_n\) set equals
\(\{i:R_i>r_n(\widehat x_n)\}\).
\end{lemma}

\begin{proof}
The continuity of \(R_1\) implies that the radii are almost surely distinct. Hence the set
\(\{i:R_i>R_{(k_n+1)}\}\) is exactly the top-\(k_n\) set. By the definition of
\(\widehat x_n\), this set can be written as \(\{i:R_i>r_n(\widehat x_n)\}\), up to null
events caused by flat parts of the survival function.

For fixed \(\delta\in(0,1/2)\), the event \(\widehat x_n>1+\delta\) is the event that at most
\(k_n\) observations exceed the deterministic threshold \(r_n(1+\delta)\). The exceedance count
is binomial with mean \((1+\delta)k_n\), so Chebyshev's inequality sends this probability to
zero. The event \(\widehat x_n<1-\delta\) is handled in the same way. Thus
\(\widehat x_n\to1\) in probability.
Let \(N_n(x):=\sum_{i=1}^n\mathbb I\{R_i>r_n(x)\}\). On the event above,
\(N_n(\widehat x_n)=k_n\), while
\(\mathbb G_n(h_0,x)=[N_n(x)-xk_n]/\sqrt{k_n}\).
Therefore
\[
\sqrt{k_n}(1-\widehat x_n)=\mathbb G_n(h_0,\widehat x_n).
\]
Lemma~\ref{lem:oracle-tail-empirical-process}, its asymptotic equicontinuity in \(x\), and
\(\widehat x_n\to1\) give
\(\mathbb G_n(h_0,\widehat x_n)=\mathbb G_n(h_0,1)+o_p(1)\), which proves the expansion.
\end{proof}

\begin{proof}[Proof of Theorem~\ref{thm:oracle-agca-clt}]
Let
\[
S_n(x):=\frac1{k_n}\sum_{i=1}^n h_\mu(G_i)\mathbb I\{R_i>r_n(x)\}.
\]
By Lemma~\ref{lem:oracle-random-threshold-expansion},
\(\widehat\Sigma_{\mu,n}^{(k_n),\mathrm{orc}}=S_n(\widehat x_n)\) with probability tending to
one. Also
\(\E[S_n(x)]
=
x\,H_\mu(r_n(x))\).
On \(\{\widehat x_n\in[1/2,2]\}\),
\[
\begin{aligned}
\sqrt{k_n}\left(S_n(\widehat x_n)-\Sigma_\mu\right)
={}&
\mathbb G_n(h_\mu,\widehat x_n)
+
\sqrt{k_n}(\widehat x_n-1)H_\mu(r_n(\widehat x_n))\\
&+
\sqrt{k_n}\{H_\mu(r_n(\widehat x_n))-\Sigma_\mu\}.
\end{aligned}
\]
Since \(\bar F_R(r_n(\widehat x_n))\le2k_n/n\) with probability tending to one, the last term is
bounded in operator norm by
\(\sqrt{k_n}\eta_\mu(2k_n/n)=o(1)\). The same bound gives
\(H_\mu(r_n(\widehat x_n))\to\Sigma_\mu\). Lemma~\ref{lem:oracle-tail-empirical-process} and
asymptotic equicontinuity yield
\(\mathbb G_n(h_\mu,\widehat x_n)=\mathbb G_n(h_\mu,1)+o_p(1)\), while
Lemma~\ref{lem:oracle-random-threshold-expansion} gives
\[
\sqrt{k_n}(\widehat x_n-1)H_\mu(r_n(\widehat x_n))
=
-\mathbb G_n(h_0,1)\Sigma_\mu+o_p(1).
\]
Thus, entrywise,
\[
\sqrt{k_n}
\left(
\widehat\Sigma_{\mu,n}^{(k_n),\mathrm{orc}}-\Sigma_\mu
\right)
=
\mathbb G_n(h_\mu,1)-\mathbb G_n(h_0,1)\Sigma_\mu+o_p(1)
\rightsquigarrow
Z.
\]

It remains to identify the covariance. For symmetric \(A,B\), write
\(W_A:=\tr(AW(h_\mu,1))\). Lemma~\ref{lem:oracle-tail-empirical-process} gives
\[
\Cov[W_A,W_B]
=
\E[(U_\mu^{\T}AU_\mu)(U_\mu^{\T}BU_\mu)],
\qquad
\Cov[W_A,W(h_0,1)]
=
\E[U_\mu^{\T}AU_\mu]
=
\tr(A\Sigma_\mu),
\]
and \(\Var[W(h_0,1)]=1\). Therefore
\[
\Cov[\tr(AZ),\tr(BZ)]
=
\Cov[U_\mu^{\T}AU_\mu,U_\mu^{\T}BU_\mu].
\]
This is exactly the covariance of
\(\tr(A(U_\mu U_\mu^{\T}-\Sigma_\mu))\) and
\(\tr(B(U_\mu U_\mu^{\T}-\Sigma_\mu))\) for a single draw of \(U_\mu\), which gives the
i.i.d.\ weak-limit representation. Since \(U_\mu\in\mu^\perp\) almost surely, \(Z\) is supported
on symmetric matrices acting on \(\mu^\perp\).
\end{proof}

\begin{proof}[Proof of Corollary~\ref{cor:oracle-agca-spectral-clt}]
Write
\(E_n:=\widehat\Sigma_{\mu,n}^{(k_n),\mathrm{orc}}-\Sigma_\mu\). By
Theorem~\ref{thm:oracle-agca-clt} and fixed dimension,
\(\norm{E_n}_{\mathrm{op}}=O_p(k_n^{-1/2})\).
The total anchored variation is linear:
\[
\sqrt{k_n}(\widehat\tau_\mu^{\mathrm{orc}}-\tau_\mu)
=
\tr(\sqrt{k_n}E_n),
\]
so the first limit follows from Theorem~\ref{thm:oracle-agca-clt} with \(A=I\), since
\(U_\mu^{\T}IU_\mu=\norm{U_\mu}_2^2\).

For a simple eigenvalue \(\lambda_j\), standard symmetric-matrix perturbation theory gives
\[
\widehat\lambda_j^{\mathrm{orc}}-\lambda_j
=
b_j^{\T}E_nb_j+O_p(\norm{E_n}_{\mathrm{op}}^2),
\]
where the constant depends on the isolation gap around \(\lambda_j\). Multiplying by
\(\sqrt{k_n}\) makes the remainder negligible, and the stated normal limit follows from
Theorem~\ref{thm:oracle-agca-clt} with \(A=b_jb_j^{\T}\). Joint convergence over any fixed set
of simple eigenvalues follows from joint convergence of the corresponding linear forms.

Assume now that \(\Delta_{\mu,p}>0\). Let
\(S_{\mu,p}:=\sum_{j=1}^p\lambda_j\) and
\(\widehat S_{\mu,p}^{\mathrm{orc}}:=\sum_{j=1}^p\widehat\lambda_j^{\mathrm{orc}}\). The map
taking a symmetric matrix near \(\Sigma_\mu\) to the sum of its leading \(p\) eigenvalues is
differentiable at \(\Sigma_\mu\), with derivative \(E\mapsto\tr(P_{\mu,p}E)\). Hence
\[
\sqrt{k_n}\left(\widehat S_{\mu,p}^{\mathrm{orc}}-S_{\mu,p}\right)
=
\tr(P_{\mu,p}\sqrt{k_n}E_n)+o_p(1).
\]
Applying the delta method to
\(\mathrm{AVE}_{\mu,p}=S_{\mu,p}/\tau_\mu\) gives the influence function
\[
\tau_\mu^{-1}
\left(
U_\mu^{\T}P_{\mu,p}U_\mu
-
\mathrm{AVE}_{\mu,p}\norm{U_\mu}_2^2
\right),
\]
which proves the AVE limit and variance formula.

Finally, the spectral projector map is differentiable under the eigengap. Its derivative at
\(\Sigma_\mu\) in the direction of a symmetric matrix \(E\) is
\[
\mathcal D_{\mu,p}(E)
=
\sum_{j=1}^p\sum_{m=p+1}^{d-1}
\frac{
b_m b_m^{\T} E b_j b_j^{\T}
+
b_j b_j^{\T} E b_m b_m^{\T}
}
{\lambda_j-\lambda_m}.
\]
For completeness, this follows from the resolvent representation of the spectral projector and a
first-order resolvent expansion around a contour separating
\(\{\lambda_1,\ldots,\lambda_p\}\) from the rest of the spectrum. Applying the continuous
mapping theorem to \(\sqrt{k_n}E_n\rightsquigarrow Z\) gives the projector limit.
\end{proof}

\begin{proof}[Proof of Corollary~\ref{cor:oracle-agca-plugin-ci}]
The proof of Theorem~\ref{thm:top-k-consistency} shows that oracle top-\(k_n\) averages of any
fixed bounded continuous function of \(G_i\) converge in probability to the corresponding
expectation under the limiting angular law. This applies to the functions needed for the first
and second moments of
\[
\psi_{\mu,p}(u)
=
\tau_\mu^{-1}
\left(
u^{\T}P_{\mu,p}u-\mathrm{AVE}_{\mu,p}\norm{u}_2^2
\right).
\]
By Theorem~\ref{thm:top-k-consistency} and the eigengap assumption,
\(\widehat\tau_\mu^{\mathrm{orc}}\to\tau_\mu\),
\(\widehat{\mathrm{AVE}}_{\mu,p}^{\mathrm{orc}}\to\mathrm{AVE}_{\mu,p}\), and
\(\widehat P_{\mu,p}^{\mathrm{orc}}\to P_{\mu,p}\) in probability. Since \(\tau_\mu>0\) and
\(\norm{u}_2\le1\), the map from
\((\tau,\mathrm{AVE},P,u)\) to
\(\tau^{-1}(u^{\T}Pu-\mathrm{AVE}\norm{u}_2^2)\) is uniformly Lipschitz on a neighborhood of
the truth. The empirical first and second moments of
\(\widehat\psi_{\mu,p,i}^{\mathrm{orc}}\) therefore converge to the corresponding moments of
\(\psi_{\mu,p}(U_\mu)\), and the empirical variance converges to
\(\sigma_{\mu,p}^2\). The confidence interval follows from
Corollary~\ref{cor:oracle-agca-spectral-clt} and Slutsky's theorem. The same argument applies
to the plug-in variances for \(\widehat\tau_\mu^{\mathrm{orc}}\) and simple oracle eigenvalues.
\end{proof}

Two boundary cases are worth recording. If the angular law is exactly rank \(p\), so that
\(U_\mu\in W_{\mu,p}\) almost surely, then \(P_{\mu,p}U_\mu=U_\mu\),
\(\mathrm{AVE}_{\mu,p}=1\), and the influence function in
Corollary~\ref{cor:oracle-agca-spectral-clt} is identically zero. Hence
\(\sigma_{\mu,p}^2=0\), and the first-order normal interval for
\(\mathrm{AVE}_{\mu,p}\) degenerates. Likewise, if \(b_j^{\T}U_\mu=0\) almost surely for an
eigen-direction \(b_j\), then the simple-eigenvalue variance
\(\Var[(b_j^{\T}U_\mu)^2]\) is zero. These are boundary cases of the population angular law, not
failures of the oracle CLT; they only say that the corresponding first-order fluctuation vanishes.

The same proof also gives explicit plug-in variance estimators for the other scalar limits in
Corollary~\ref{cor:oracle-agca-spectral-clt}. If
\(\widehat v_\tau^{\mathrm{orc}}\) is the empirical variance of
\(\norm{U_{\mu,i}}_2^2\) over \(i\in\mathcal I_n^{(k_n)}\), then
\(\widehat v_\tau^{\mathrm{orc}}\to\Var[\norm{U_\mu}_2^2]\) in probability. If
\(\lambda_j\) is simple and \(\widehat v_{\lambda_j}^{\mathrm{orc}}\) is the empirical variance
of \(\{(\widehat b_j^{\mathrm{orc}})^{\T}U_{\mu,i}\}^2\) over
\(i\in\mathcal I_n^{(k_n)}\), then
\(\widehat v_{\lambda_j}^{\mathrm{orc}}\to\Var[(b_j^{\T}U_\mu)^2]\) in probability. Joint
covariances over any fixed set of simple eigenvalues are estimated analogously by the empirical
covariances of these plug-in squared scores.

\subsection{A conditional rank-margin CLT}
\label{app:conditional-rank-margin-clt}

Rank-tail empirical-process expansions of this kind underlie rank-based inference for
extreme-value copulas and stable tail-dependence functions
\citep[e.g.,][]{genest2009rank,einmahl2012m}. The following result states the corresponding AGCA
consequence in a high-level form. It isolates the part specific to AGCA: once the rank-Pareto tail
empirical measure has a Gaussian expansion and the induced top-\(k_n\) angular functional is
differentiable, the same spectral delta method used for the oracle CLT applies.

Let \(E_+:=\R^d_+\setminus\{0\}\), \(s(x):=x/\norm{x}_2\), and
\(\varphi_\mu(x):=u_\mu(s(x))u_\mu(s(x))^{\T}\). For \(q>0\), write
\(A_q:=\{x\in E_+:\norm{x}_2>q\}\). If \(m\) is a Radon measure on \(E_+\) and a threshold
\(q(m)\) can be chosen so that \(m(A_{q(m)})=1\), define
\[
\Phi_\mu(m):=\int_{A_{q(m)}}\varphi_\mu(x)\,dm(x).
\]

\begin{theorem}[Conditional rank-margin CLT]
\label{thm:conditional-rank-margin-clt}
Assume Assumption~\ref{ass:erv}. Let \(a_n:=n/k_n\) and define the rank-tail empirical measure
\[
\widehat\nu_n^{\mathrm{rank}}
:=
\frac1{k_n}\sum_{i=1}^n\delta_{\widehat X_i/a_n}.
\]
Suppose that there exist a Radon tail measure \(\nu\), a normed linear space \(\mathbb D\) of
signed measures, a linear subspace \(\mathbb D_0\subset\mathbb D\), and a centered Gaussian
element \(\mathbb W_\nu^{\mathrm{rank}}\in\mathbb D\), with
\(\mathbb W_\nu^{\mathrm{rank}}\in\mathbb D_0\) almost surely, such that
\[
\sqrt{k_n}
\left(
\widehat\nu_n^{\mathrm{rank}}-\nu
\right)
\rightsquigarrow
\mathbb W_\nu^{\mathrm{rank}}
\qquad\text{in }\mathbb D.
\]
Assume further that \(q(\widehat\nu_n^{\mathrm{rank}})\) can be chosen so that
\(A_{q(\widehat\nu_n^{\mathrm{rank}})}\) contains the rank-Pareto top-\(k_n\) observations with
probability tending to one, that \(q(\nu)\) satisfies \(\nu(A_{q(\nu)})=1\) and
\(\nu(\partial A_{q(\nu)})=0\), and that \(\Phi_\mu\) is Hadamard differentiable at \(\nu\)
tangentially to \(\mathbb D_0\), with continuous linear derivative
\(\dot\Phi_{\mu,\nu}\). Then
\[
\sqrt{k_n}
\left(
\widehat\Sigma_{\mu,n}^{(k_n),\mathrm{emp}}-\Sigma_\mu
\right)
\rightsquigarrow
Z_\mu^{\mathrm{rank}}
:=
\dot\Phi_{\mu,\nu}(\mathbb W_\nu^{\mathrm{rank}}),
\]
where \(Z_\mu^{\mathrm{rank}}\) is a centered Gaussian symmetric matrix. Consequently, if
\(1\le p\le d-2\), \(\tau_\mu>0\), and \(\Delta_{\mu,p}>0\), then
\(\widehat{\mathrm{AVE}}_{\mu,p}\) and
\(\widehat P_{\mu,p}\), formed from
\(\widehat\Sigma_{\mu,n}^{(k_n),\mathrm{emp}}\), satisfy
\[
\sqrt{k_n}
\left(
\widehat{\mathrm{AVE}}_{\mu,p}-\mathrm{AVE}_{\mu,p}
\right)
\rightsquigarrow
\tau_\mu^{-1}
\left\{
\tr(P_{\mu,p}Z_\mu^{\mathrm{rank}})
-
\mathrm{AVE}_{\mu,p}\tr(Z_\mu^{\mathrm{rank}})
\right\},
\]
and
\[
\sqrt{k_n}
\left(
\widehat P_{\mu,p}-P_{\mu,p}
\right)
\rightsquigarrow
\mathcal D_{\mu,p}(Z_\mu^{\mathrm{rank}}),
\]
with \(\mathcal D_{\mu,p}\) as in Corollary~\ref{cor:oracle-agca-spectral-clt}. The same
replacement of \(Z\) by \(Z_\mu^{\mathrm{rank}}\) gives the limits for total anchored variation
and simple eigenvalues.
\end{theorem}
The theorem makes two high-level assumptions. For tail empirical measures with estimated Pareto margins, Gaussian expansions of the assumed
type have been studied in the literature
\citep{einmahl1997estimating,einmahl2001nonparametric,einmahl2009maximum}. The expansion
assumption should be attainable along that route under suitable second-order marginal and joint-tail
conditions. The differentiability assumption is more delicate: when the limiting
angular law places mass on the axes, \(\Phi_\mu\) can respond non-smoothly to componentwise
marginal scalings. The two assumptions together take care of the face and near-axis regimes that the bounded departures \(u_\mu\)
accommodate at the oracle level.
\begin{proof}
Put \(\widehat m_n:=\widehat\nu_n^{\mathrm{rank}}\). By the assumed choice of
\(q(\widehat m_n)\), the set \(A_{q(\widehat m_n)}\) contains exactly the observations selected
by the rank-Pareto top-\(k_n\) rule, with probability tending to one. Therefore
\[
\Phi_\mu(\widehat m_n)
=
\frac1{k_n}
\sum_{i\in\widehat{\mathcal I}_n^{(k_n)}}
u_\mu(\widehat G_i)u_\mu(\widehat G_i)^{\T}
=
\widehat\Sigma_{\mu,n}^{(k_n),\mathrm{emp}}
\]
with probability tending to one. Under Assumption~\ref{ass:erv}, the limiting tail measure has a
polar representation whose conditional angular law on any radial exceedance set \(A_q\) is the
angular law of \(G\). Since \(\nu(A_{q(\nu)})=1\), this gives
\[
\Phi_\mu(\nu)
=
\int_{A_{q(\nu)}}\varphi_\mu(x)\,d\nu(x)
=
\E[u_\mu(G)u_\mu(G)^{\T}]
=
\Sigma_\mu.
\]
The functional delta method applied to the assumed weak convergence of
\(\widehat m_n\) and to the Hadamard differentiability of \(\Phi_\mu\) gives
\[
\sqrt{k_n}
\left(
\Phi_\mu(\widehat m_n)-\Phi_\mu(\nu)
\right)
\rightsquigarrow
\dot\Phi_{\mu,\nu}(\mathbb W_\nu^{\mathrm{rank}}).
\]
The first two displayed identities identify the left-hand side with
\(\sqrt{k_n}(\widehat\Sigma_{\mu,n}^{(k_n),\mathrm{emp}}-\Sigma_\mu)\), except on an event
whose probability tends to zero. Because \(\dot\Phi_{\mu,\nu}\) is continuous and linear and
\(\mathbb W_\nu^{\mathrm{rank}}\) is centered Gaussian, the limit
\(Z_\mu^{\mathrm{rank}}\) is centered Gaussian.

The scalar and spectral conclusions are the same perturbation steps used in the proof of
Corollary~\ref{cor:oracle-agca-spectral-clt}. The trace map is linear, simple eigenvalues are
first-order differentiable with derivative \(E\mapsto b_j^{\T}Eb_j\), the explained-variation map
has derivative
\[
E\mapsto
\tau_\mu^{-1}
\left\{
\tr(P_{\mu,p}E)-\mathrm{AVE}_{\mu,p}\tr(E)
\right\},
\]
and the rank-\(p\) spectral projector has derivative
\(\mathcal D_{\mu,p}\) under \(\Delta_{\mu,p}>0\). Applying the continuous mapping theorem to
\(\sqrt{k_n}(\widehat\Sigma_{\mu,n}^{(k_n),\mathrm{emp}}-\Sigma_\mu)
\rightsquigarrow Z_\mu^{\mathrm{rank}}\) proves all displayed limits.
\end{proof}

%% file: appendix_simulations.tex
\section{Simulation analysis}
\label{app:simulations}

This supplementary section provides a simulation analysis used to illustrate AGCA. To make the
geometry visible, we begin with two three-dimensional designs. Model~1 is a low-dimensional angular
law whose selected directions form a slightly bent cloud on the positive sphere. Model~2 adds a
variable-specific near-axis regime, so that one coordinate can dominate the angular direction
while remaining a finite point on \(\Sphere^2_+\). We then add a ten-dimensional design in which
eight variables share a lower-rank angular mechanism and two variables have asymptotically
independent near-axis regimes. The diagnostics below focus on visualization, explained variation,
scores, loadings, anchor choice, bootstrap stability, finite-sample sensitivity, and robustness of
the canonical-anchor summary beyond the visible three-dimensional case. For the ten-dimensional
design, we also include a repeated-sampling coverage check for the oracle plug-in intervals,
comparing the formal oracle-margin setting with the same interval formula applied after
rank-Pareto marginal standardization.

Throughout, the canonical anchor is \(\mu_0=d^{-1/2}\one_d\) in the relevant dimension. We compare
it with the empirical Fr\'echet and principal anchors in \eqref{eq:data-adaptive-anchors}, computed
from the selected directions. Bootstrap 95\% diagnostic intervals (DIs) are conditional
summaries: after fixing the selected angular directions, we resample those directions with
replacement, refit AGCA at the same rank and anchor, and align loading signs to the main-sample
loadings.

\subsection{Data-generating mechanisms}
\label{app:simulations-dgp-details}

Let \(Y_{\theta,q}\) be a \(q\)-variate logistic max-stable vector with unit Fr\'echet margins,
\begin{equation}
\label{eq:app-sim-logistic-frechet}
\Prob[Y_{\theta,q,1}\le y_1,\ldots,Y_{\theta,q,q}\le y_q]
=
\exp\left[-\left(\sum_{j=1}^q y_j^{-1/\theta}\right)^{\theta}\right],
\qquad y_j>0,\quad 0<\theta<1.
\end{equation}
The corresponding logistic-Pareto block is
\begin{equation}
\label{eq:app-sim-logistic-pareto}
Z_{\theta,q,j}:=(1-\exp(-1/Y_{\theta,q,j}))^{-1},
\qquad j=1,\ldots,q,
\end{equation}
so each margin is standard Pareto. For this specific model, \(\eta_\mu(t)=O(t^{\kappa})\) with
\(\kappa=\min\{1,(1-\theta)/\theta\}\). Therefore, Assumption \ref{ass:second-order-angular-bias} holds whenever
\(k_n=o(n^{2\kappa/(2\kappa+1)})\); for \(\theta\le1/2\) this is \(k_n=o(n^{2/3})\), up to a
logarithmic refinement at \(\theta=1/2\).

Let \(P\) denote an independent standard Pareto variable and
let \(E_i\) collect independent standard exponential noise variables. The light-tailed noise
keeps finite samples nonsingular, but is asymptotically negligible relative to the Pareto
signals.

\begin{table}[tbp]
\centering
\small
\begin{tabular}{@{}p{0.15\linewidth}p{0.12\linewidth}p{0.43\linewidth}p{0.20\linewidth}@{}}
\toprule
Design & \(n,k\) & Extremal mechanism & Near-axis part \\
\midrule
Model 1 & \(2400,120\) &
Trivariate logistic-Pareto block embedded through three positive rays. Two rays generate the
dominant \(X_1\)-versus-\(X_2\) contrast; a weaker third ray bends the cloud through \(X_3\). &
None by design. \\
Model 2 & \(2400,120\) &
Bivariate logistic-Pareto block for \(X_1,X_2\), plus an independent Pareto source for \(X_3\). &
\(X_3\) near-axis regime. \\
10D model & \(10000,500\) &
Five-ray logistic-Pareto block on \(X_1,\ldots,X_8\) lying in a two-dimensional anchored
surface, plus independent Pareto sources for \(X_9\) and \(X_{10}\). Margins are rank-Pareto
standardized before thresholding. &
\(X_9\) and \(X_{10}\) near-axis regimes. \\
\bottomrule
\end{tabular}
\caption{Simulation designs. Selected extremes are the \(k\) largest Euclidean radii after the
stated preprocessing, with directions normalized to the positive sphere.}
\label{tab:app-sim-design}
\end{table}

For the two three-dimensional designs, thresholding and angular normalization are applied directly
to the generated vectors \(X_i\); only the ten-dimensional design uses the rank-Pareto
preprocessing in \eqref{eq:app-sim-rank-pareto}.

For Model~1, with \(\theta_1=0.50\) and \(\tau_1=0.30\),
\begin{equation}
\label{eq:app-sim-3d-model1-dgp}
X_i=A_1Z_i+\tau_1E_i,
\qquad Z_i=(Z_{i1},Z_{i2},Z_{i3})^{\T}\sim Z_{\theta_1,3}.
\end{equation}
Let
\[
b_1=(1,-1,0)^{\T}/\sqrt2,\qquad
b_2=(1,1,-2)^{\T}/\sqrt6,\qquad
\mu_3=3^{-1/2}\one_3.
\]
With \(\rho=0.25\), \(\gamma=(1+\rho)/2\), \(\eta=0.45\), and \(\kappa=0.45\), the ray matrix is
\[
A_1=(a_1,a_2,a_3),\qquad
a_1=
\begin{pmatrix}1\\ \rho\\ \gamma\end{pmatrix},
\quad
a_2=
\begin{pmatrix}\rho\\ 1\\ \gamma\end{pmatrix},
\quad
a_3=\kappa(\mu_3-\eta b_2).
\]
The first two rays create the dominant \(b_1\) contrast and sit slightly below the canonical
anchor in the \(X_3\) direction. The third ray has larger relative \(X_3\) participation, so the
selected angular cloud is bent rather than horizontal, with \(\mu_0\) visually inside the main
cloud.

For Model~2, with \(\theta_2=0.25\) and \(\tau_2=0.30\),
\begin{equation}
\label{eq:app-sim-3d-model2-dgp}
X_i
=
\left(Z_{i1},Z_{i2},P_{i3}\right)^{\T}+\tau_2E_i,
\qquad Z_i=(Z_{i1},Z_{i2})^{\T}\sim Z_{\theta_2,2}.
\end{equation}
The independent Pareto source \(P_{i3}\) creates the variable-specific near-axis regime. In the
displayed realization, \(72\) selected extremes are shared \(X_1\)-\(X_2\) episodes and \(48\)
are \(X_3\)-axis episodes.

For the ten-dimensional model, with \(\theta_{10}=0.45\) and \(\tau_{10}=0.25\), let
\(\mu_8=8^{-1/2}\one_8\) and
\[
h_1=\left(-1,-\frac57,-\frac37,-\frac17,\frac17,\frac37,\frac57,1\right)^{\T},
\qquad
h_2=(1,1,1,1,-1,-1,-1,-1)^{\T}.
\]
Let \(B=(b_1,b_2)\) be the Gram--Schmidt orthonormalization of \((h_1,h_2)\). For the rows
\(s_\ell^{\T}\), \(\ell=1,\ldots,5\), of
\[
S=
\begin{pmatrix}
-0.42&-0.30&0.00&0.28&0.48\\
-0.26&0.30&-0.36&0.24&-0.02
\end{pmatrix}^{\T},
\]
define shared rays \(a_\ell:=(\mu_8+Bs_\ell)/\norm{\mu_8+Bs_\ell}_2\in\Sphere^7_+\),
\(\ell=1,\ldots,5\), and \(A_8=(a_1,\ldots,a_5)\). With
\(Z_i=(Z_{i1},\ldots,Z_{i5})^{\T}\sim Z_{\theta_{10},5}\), the raw vector is
\begin{equation}
\label{eq:app-sim-10d-dgp}
X_i^\star
=
\left(A_8Z_i,\; P_{i9},\; P_{i,10}\right)+\tau_{10}E_i
\in\R^{10}.
\end{equation}
The shared \(X_1,\ldots,X_8\) mechanism lies in a two-dimensional anchored surface before
normalization, whereas \(P_{i9}\) and \(P_{i,10}\) create two variable-specific near-axis regimes.
As in empirical work with unknown margins, the AGCA fit is applied after the marginal rank-Pareto
transform
\begin{equation}
\label{eq:app-sim-rank-pareto}
\widetilde X_{ij}
=
\frac{n+1}{n+1-\operatorname{rank}(X_{ij}^\star;X_{1j}^\star,\ldots,X_{nj}^\star)},
\qquad j=1,\ldots,10,
\end{equation}
where ranks are increasing within each margin. The selected extremes are the \(k=500\) largest
radii of the transformed vectors \(\widetilde X_i\).

\subsection{Sphere geometry}
\label{app:simulations-sphere-geometry}

Figure~\ref{fig:app-sim-3d-sphere-overlay} shows the geometric object fitted by AGCA. The
points are the selected extreme directions. The canonical fit uses the fixed balanced anchor;
the Fr\'echet fit recenters the geometry at the sample spherical center. Only the first AGC is
shown for each anchor.
In Model~1, the selected directions form a bent but still essentially one-dimensional angular
cloud. The canonical and Fr\'echet anchors are close, and their first AGCs trace similar
directions through the cloud. In Model~2, the first AGC runs from the balanced interior toward
the \(X_3\)-axis part of the sphere. This is the geometric sense in which AGCA handles
asymptotic independence: a near-axis episode is a large but bounded anchored displacement, not a
point outside the coordinate system.

\begin{figure}[tbp]
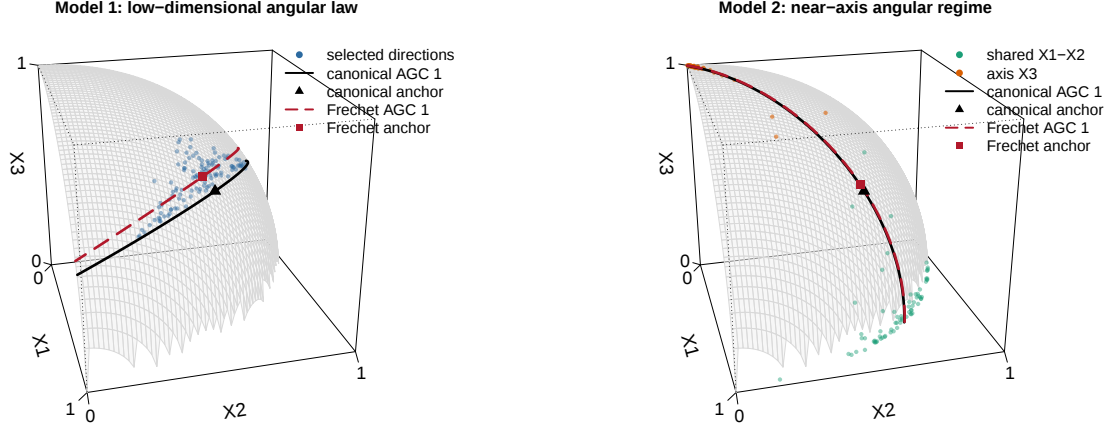

\centering
\begin{minipage}[t]{0.48\linewidth}
\centering
\appinclude[width=\linewidth]{fig_sim3d_model1_sphere_overlay.pdf}
\end{minipage}\hfill
\begin{minipage}[t]{0.48\linewidth}
\centering
\appinclude[width=\linewidth]{fig_sim3d_model2_sphere_overlay.pdf}
\end{minipage}
\caption{Sphere-level geometry in the two three-dimensional simulations. The left panel shows
Model~1 and the right panel shows Model~2. Each panel displays selected extreme directions, the
canonical anchor and first AGC, and the Fr\'echet anchor and first AGC.}
\label{fig:app-sim-3d-sphere-overlay}
\end{figure}

\subsection{Explained variation and scores}
\label{app:simulations-variation-scores}

Figure~\ref{fig:app-sim-3d-variation} reports cumulative anchored variation explained. In
Model~1, the first canonical AGC explains \(82.0\%\) of anchored variation and the rank-one
residual risk is \(0.0085\). The lower percentage relative to an exactly rank-one construction
is intentional: the cloud is bent by the third ray, but the leading direction is still sharply
identified, with absolute alignment \(0.999996\) with the designed \(b_1\) contrast. In Model~2,
the first canonical AGC explains \(93.3\%\) of anchored variation and has absolute alignment
\(0.999977\) with the \(X_3\)-axis contrast.

Figure~\ref{fig:app-sim-3d-scores} gives the score representation of the same geometry. In
Model~1, the selected extremes spread primarily along AGC~1, with a smaller second-coordinate
spread caused by the \(X_3\) bend. In Model~2, the first score separates shared \(X_1\)-\(X_2\)
episodes from \(X_3\)-axis episodes, while the second score records residual variation inside
the shared block.

\begin{figure}[tbp]
\centering
\begin{minipage}[t]{0.48\linewidth}
\centering
\textbf{(A) Model 1}\par\medskip
\appinclude[width=\linewidth]{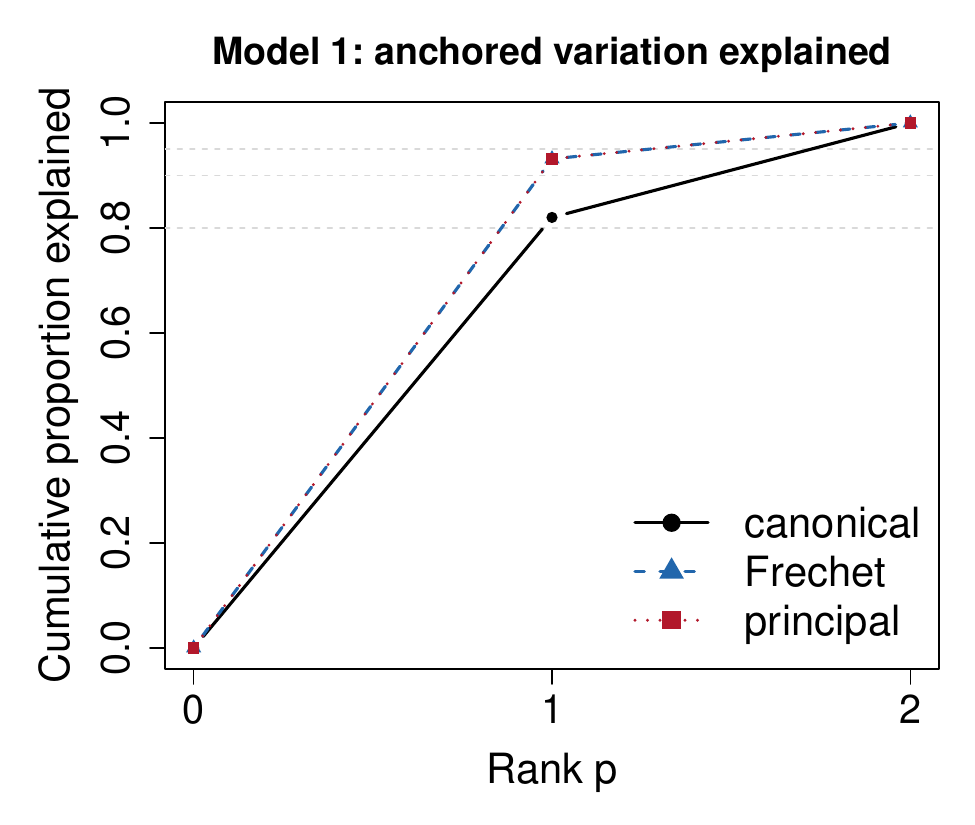}
\end{minipage}\hfill
\begin{minipage}[t]{0.48\linewidth}
\centering
\textbf{(B) Model 2}\par\medskip
\appinclude[width=\linewidth]{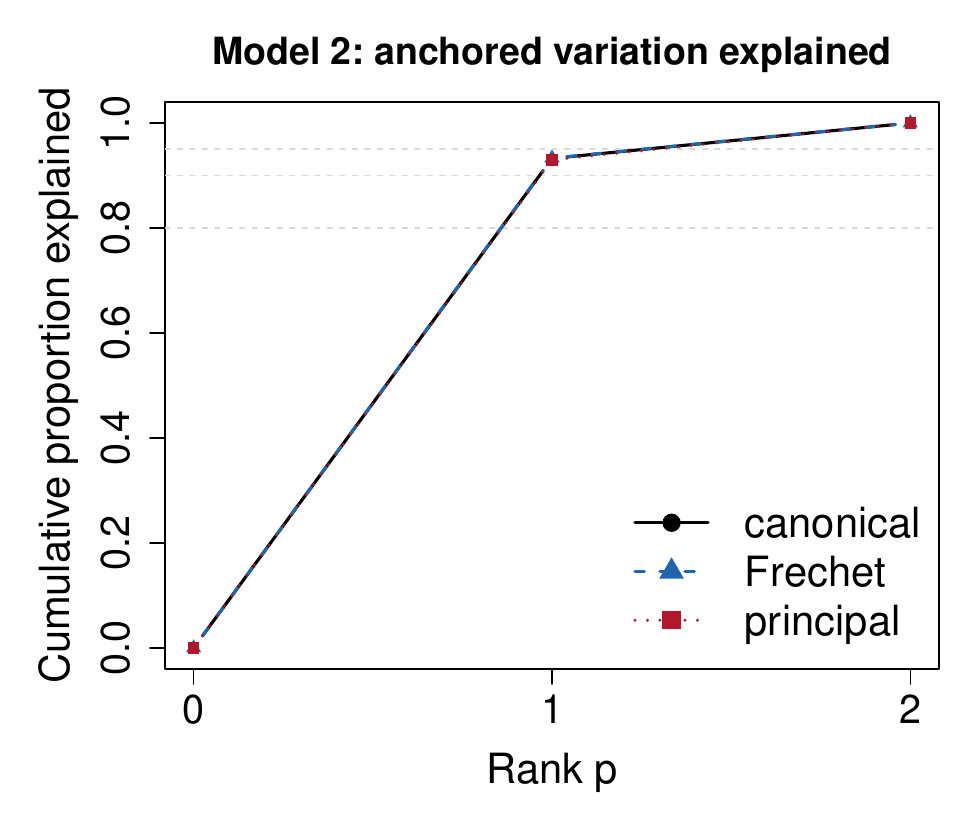}
\end{minipage}
\caption{Anchored variation explained in the three-dimensional simulations. Each panel compares
the canonical, Fr\'echet, and principal anchors.}
\label{fig:app-sim-3d-variation}
\end{figure}

\begin{figure}[tbp]
\centering
\begin{minipage}[t]{0.48\linewidth}
\centering
\textbf{(A) Model 1}\par\medskip
\appinclude[width=\linewidth]{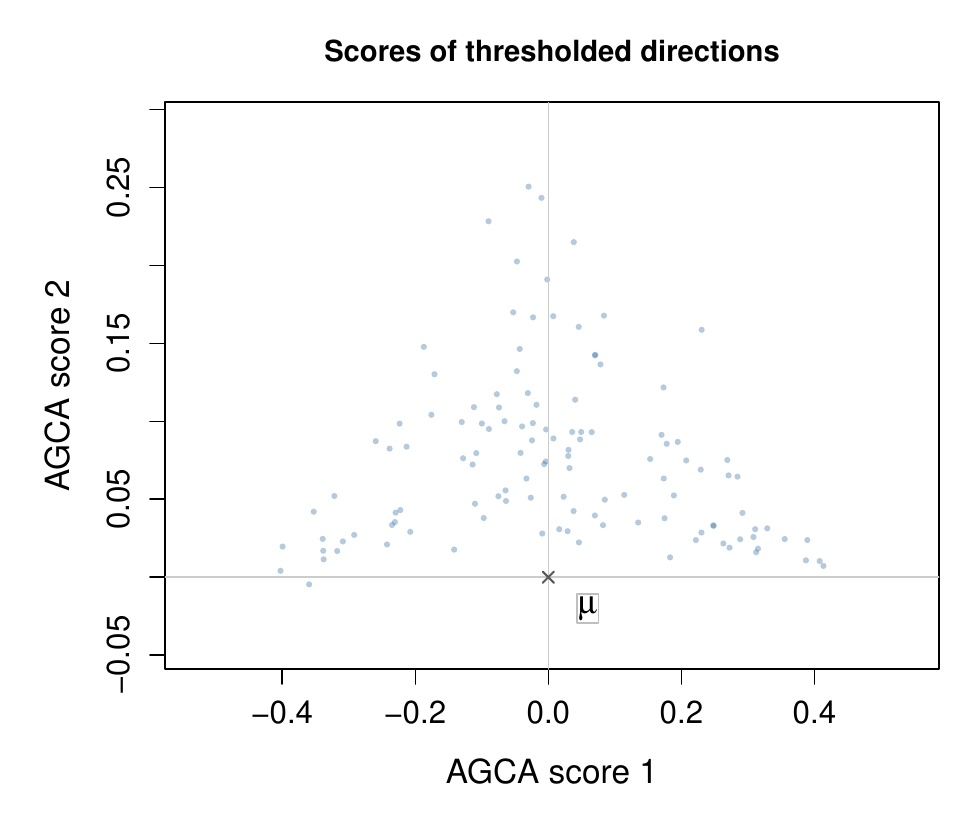}
\end{minipage}\hfill
\begin{minipage}[t]{0.48\linewidth}
\centering
\textbf{(B) Model 2}\par\medskip
\appinclude[width=\linewidth]{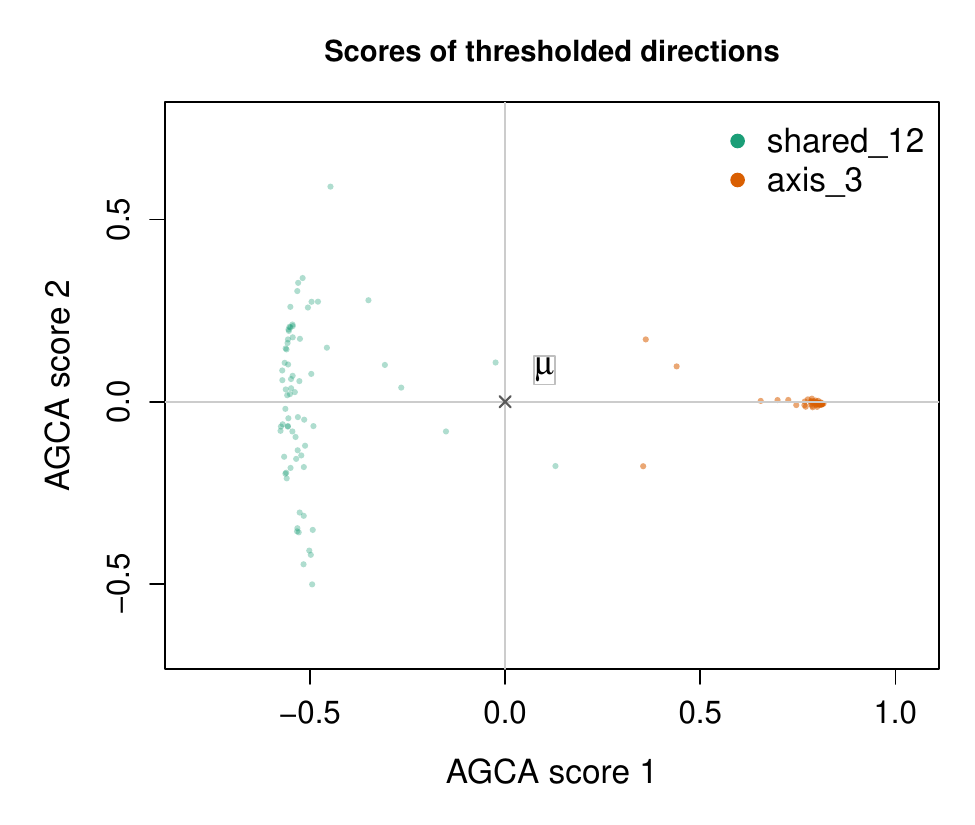}
\end{minipage}
\caption{Canonical-anchor scores for the selected extremes.}
\label{fig:app-sim-3d-scores}
\end{figure}

\subsection{Loadings}
\label{app:simulations-loadings}

Figure~\ref{fig:app-sim-3d-loadings} shows the canonical-anchor loadings. In Model~1, AGC~1 is
the \(X_1\)-versus-\(X_2\) contrast, while AGC~2 captures the weaker \(X_3\) bending direction.
In Model~2, AGC~1 is positive for \(X_3\) and negative for \(X_1,X_2\); it separates the
variable-specific near-axis extremes from the shared block. AGC~2 then describes the residual
\(X_1\)-versus-\(X_2\) contrast. The bootstrap diagnostic intervals (DIs) are narrow relative to the leading
loading magnitudes, so the displayed contrasts are stable in the selected samples.

\begin{figure}[tbp]
\centering
\begin{minipage}[t]{0.235\linewidth}
\centering
\appinclude[width=\linewidth]{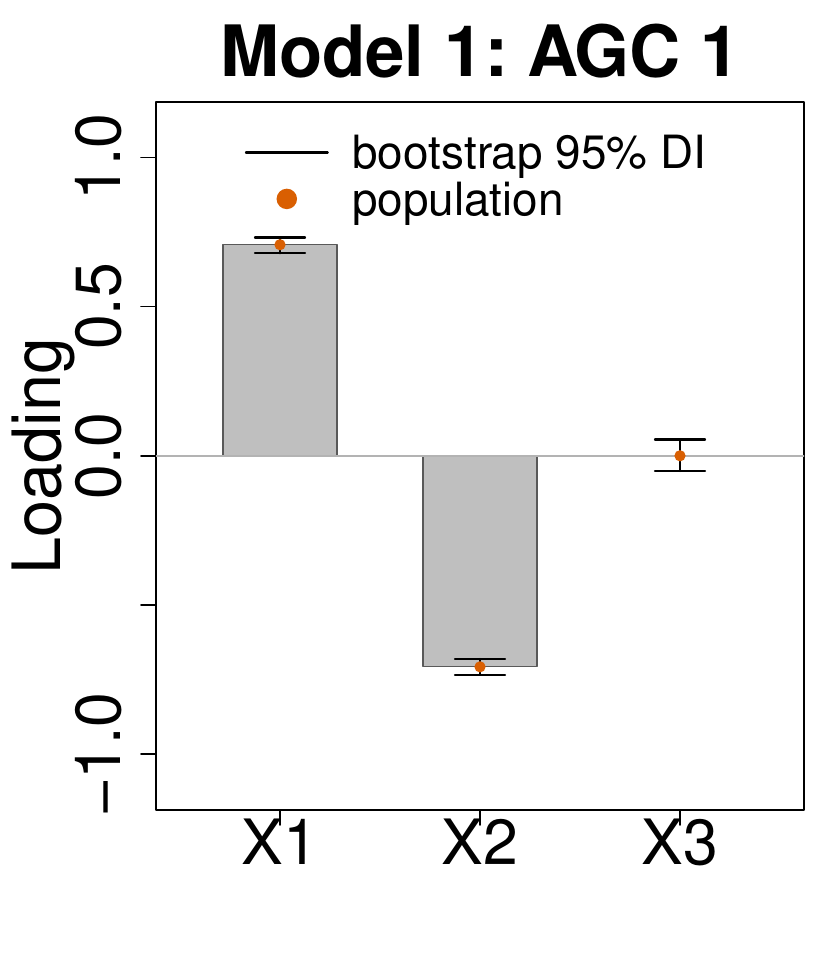}
\end{minipage}\hfill
\begin{minipage}[t]{0.235\linewidth}
\centering
\appinclude[width=\linewidth]{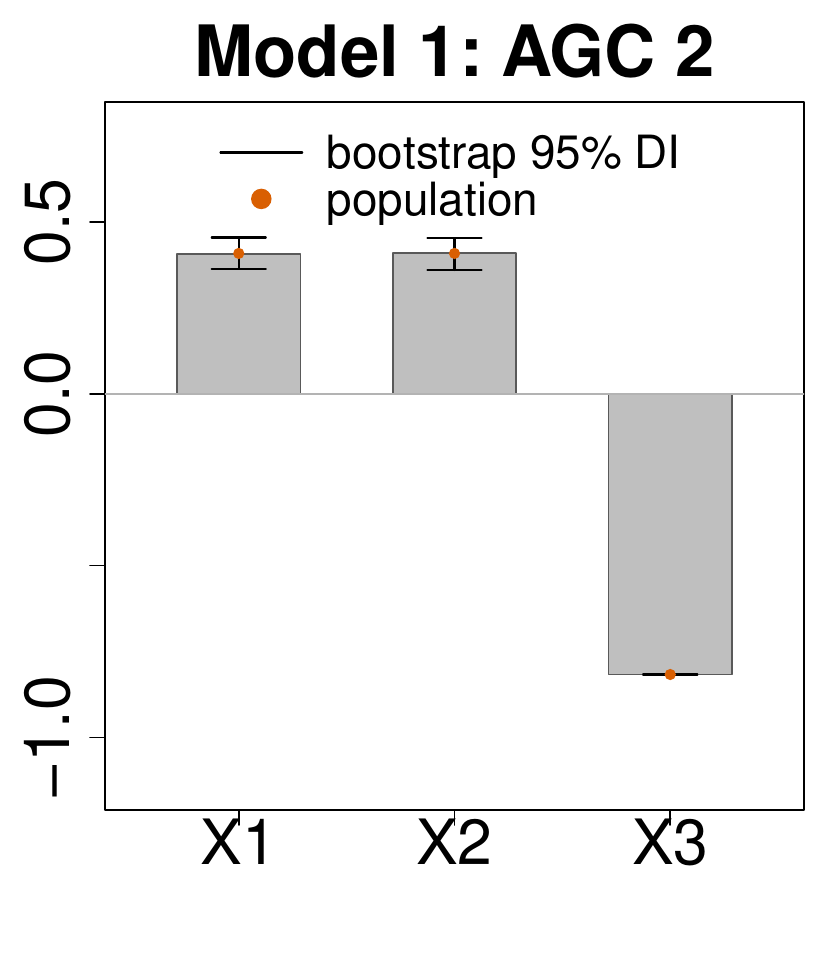}
\end{minipage}\hfill
\begin{minipage}[t]{0.235\linewidth}
\centering
\appinclude[width=\linewidth]{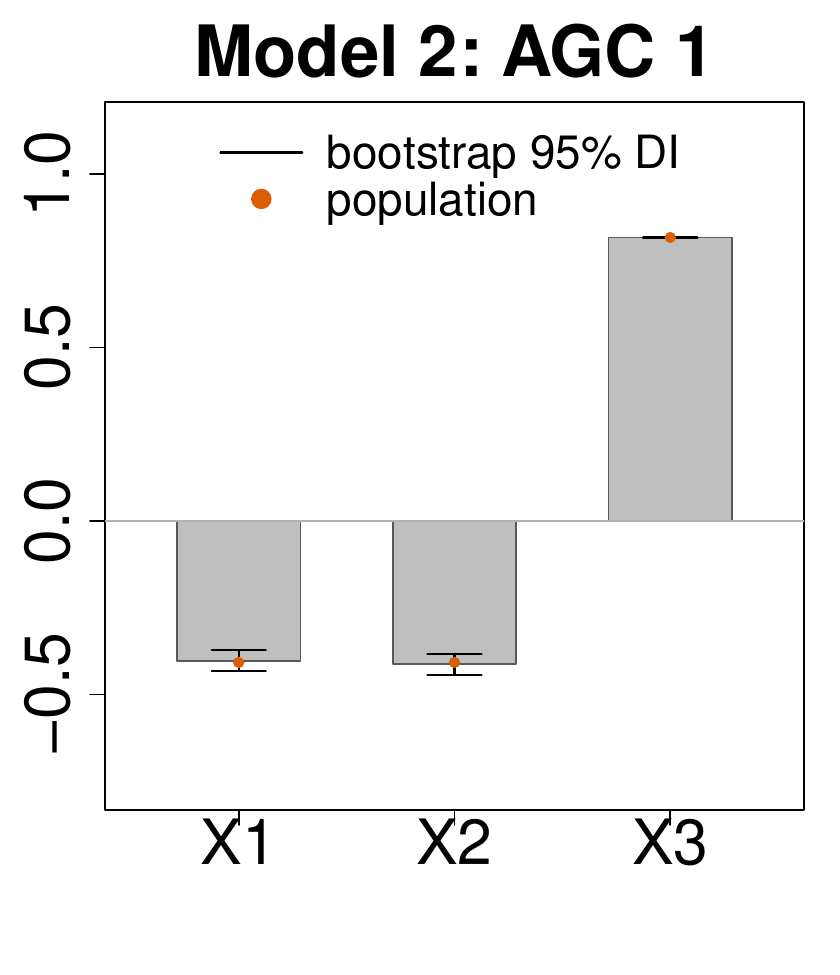}
\end{minipage}\hfill
\begin{minipage}[t]{0.235\linewidth}
\centering
\appinclude[width=\linewidth]{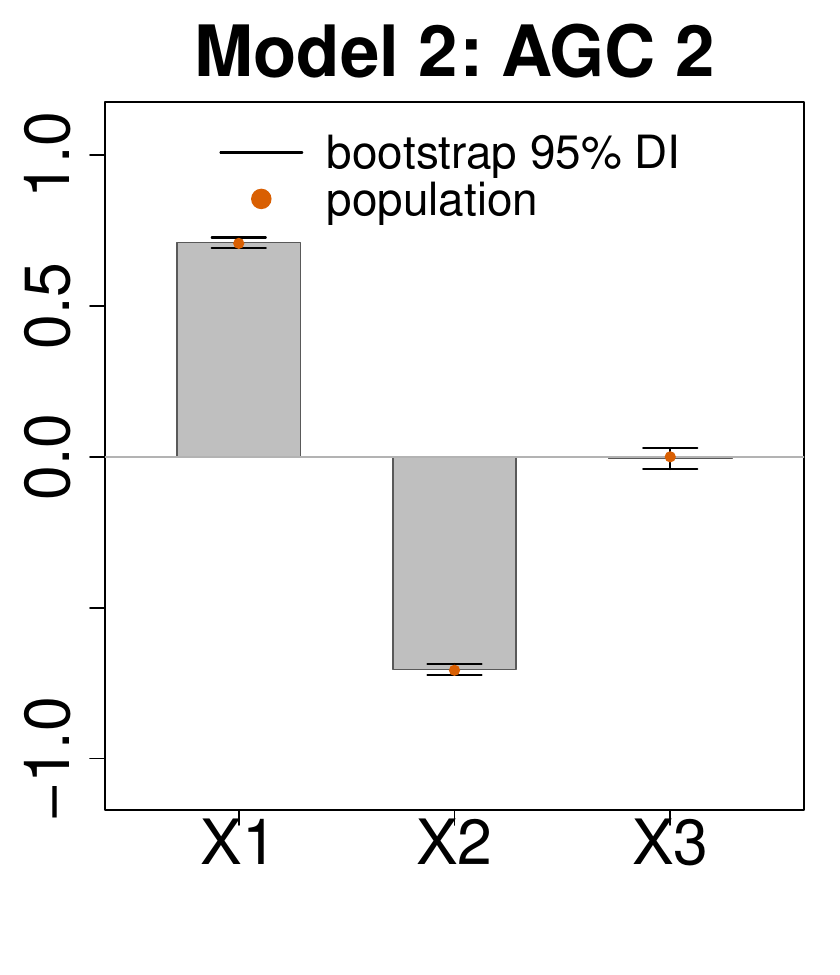}
\end{minipage}
\caption{Canonical-anchor loadings. From left to right: Model~1 AGC~1, Model~1 AGC~2,
Model~2 AGC~1, and Model~2 AGC~2. Bars show sample loadings, vertical intervals show bootstrap
95\% diagnostic intervals (DIs), and dots show population loadings matched to the sample AGCs.}
\label{fig:app-sim-3d-loadings}
\end{figure}

\subsection{Anchor choice}
\label{app:simulations-anchor-choice}

The canonical anchor fixes the interpretation as departure from balanced complete dependence.
Data-adaptive anchors can give similar scalar fits, but they change the zero point of the loading
and score coordinates. Figure~\ref{fig:app-sim-3d-anchor-distance} reports the geodesic distances
from the canonical anchor. In Model~1, the Fr\'echet and principal anchors are both about
\(0.078\) radians from \(\mu_0\). In Model~2, the Fr\'echet anchor remains close \((0.029)\), while
the principal anchor moves farther \((0.296)\) because the leading angular direction is pulled
toward the \(X_3\) axis. The sphere-level Fr\'echet fit is already displayed in
Figure~\ref{fig:app-sim-3d-sphere-overlay}; Figure~\ref{fig:app-sim-3d-frechet-loadings} gives
the corresponding Fr\'echet-anchor loadings. The dominant directions are comparable to the
canonical directions, but their interpretation is sample-centered rather than benchmark-relative.
This distinction is the main reason the canonical anchor remains the default in the paper.

\begin{figure}[tbp]
\centering
\begin{minipage}[t]{0.36\linewidth}
\centering
\textbf{(A) Model 1}\par\medskip
\appinclude[width=\linewidth]{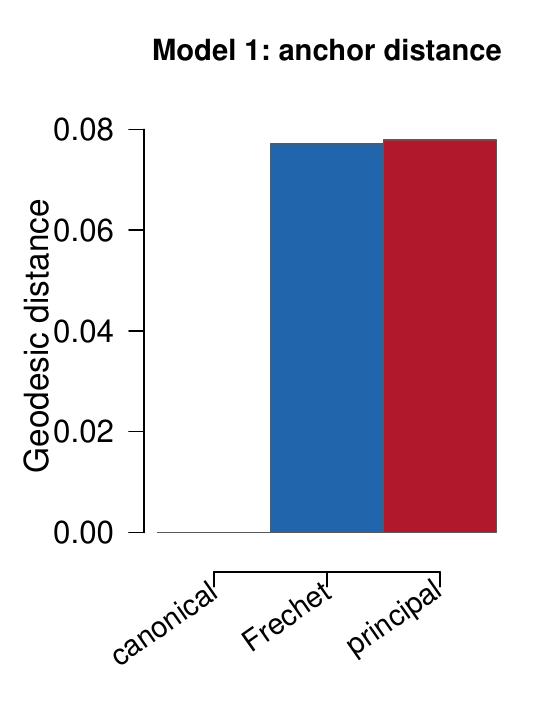}
\end{minipage}\hfill
\begin{minipage}[t]{0.36\linewidth}
\centering
\textbf{(B) Model 2}\par\medskip
\appinclude[width=\linewidth]{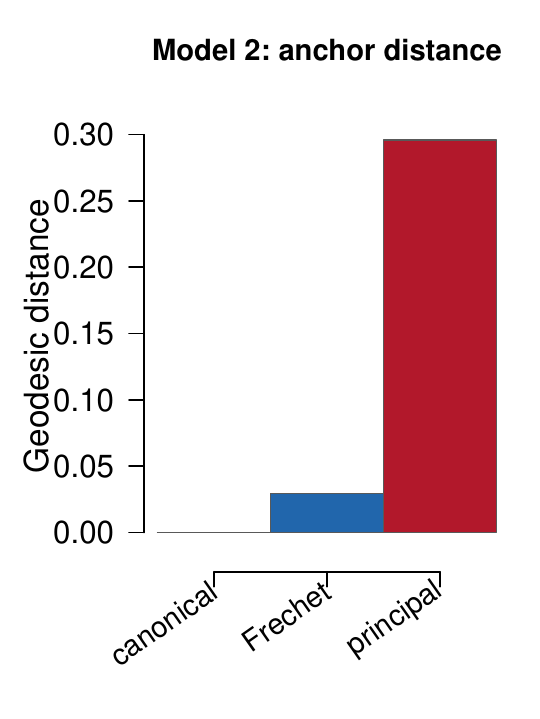}
\end{minipage}
\caption{Anchor sensitivity distances. Distances are spherical geodesic distances from the
canonical anchor.}
\label{fig:app-sim-3d-anchor-distance}
\end{figure}

\begin{figure}[tbp]
\centering
\begin{minipage}[t]{0.235\linewidth}
\centering
\appinclude[width=\linewidth]{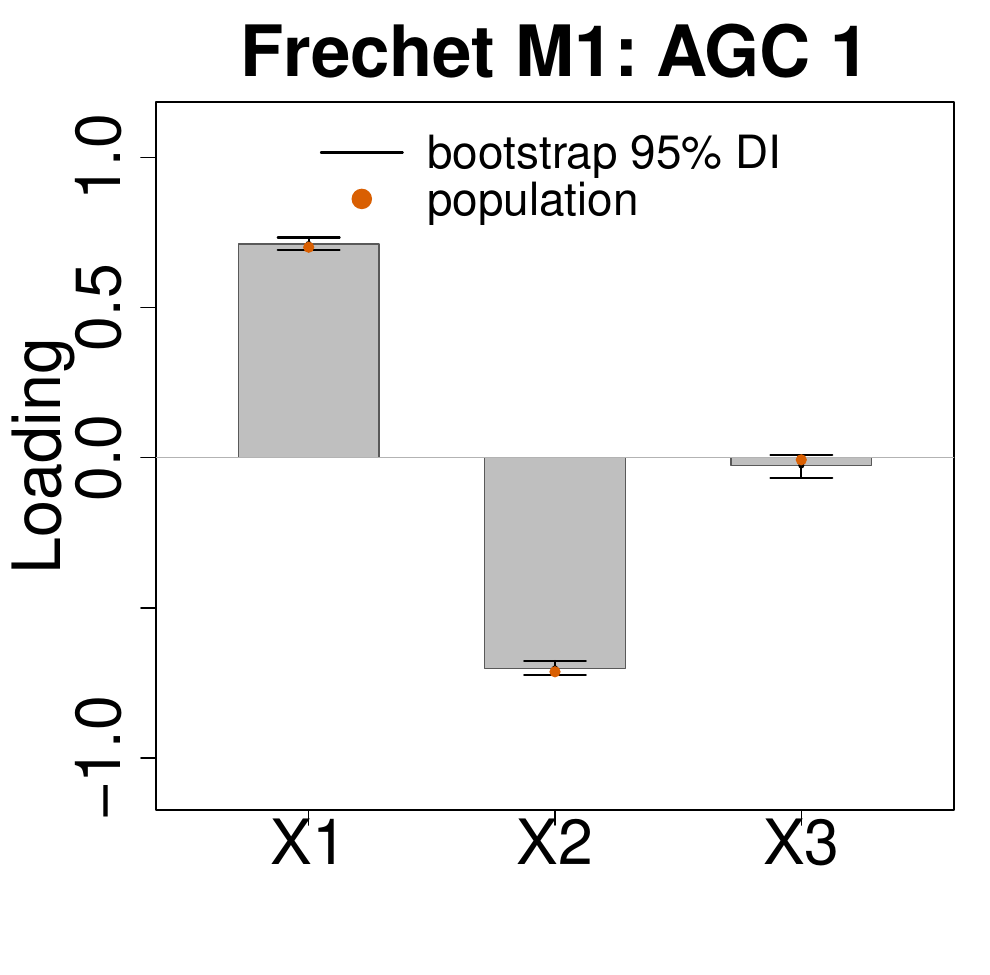}
\end{minipage}\hfill
\begin{minipage}[t]{0.235\linewidth}
\centering
\appinclude[width=\linewidth]{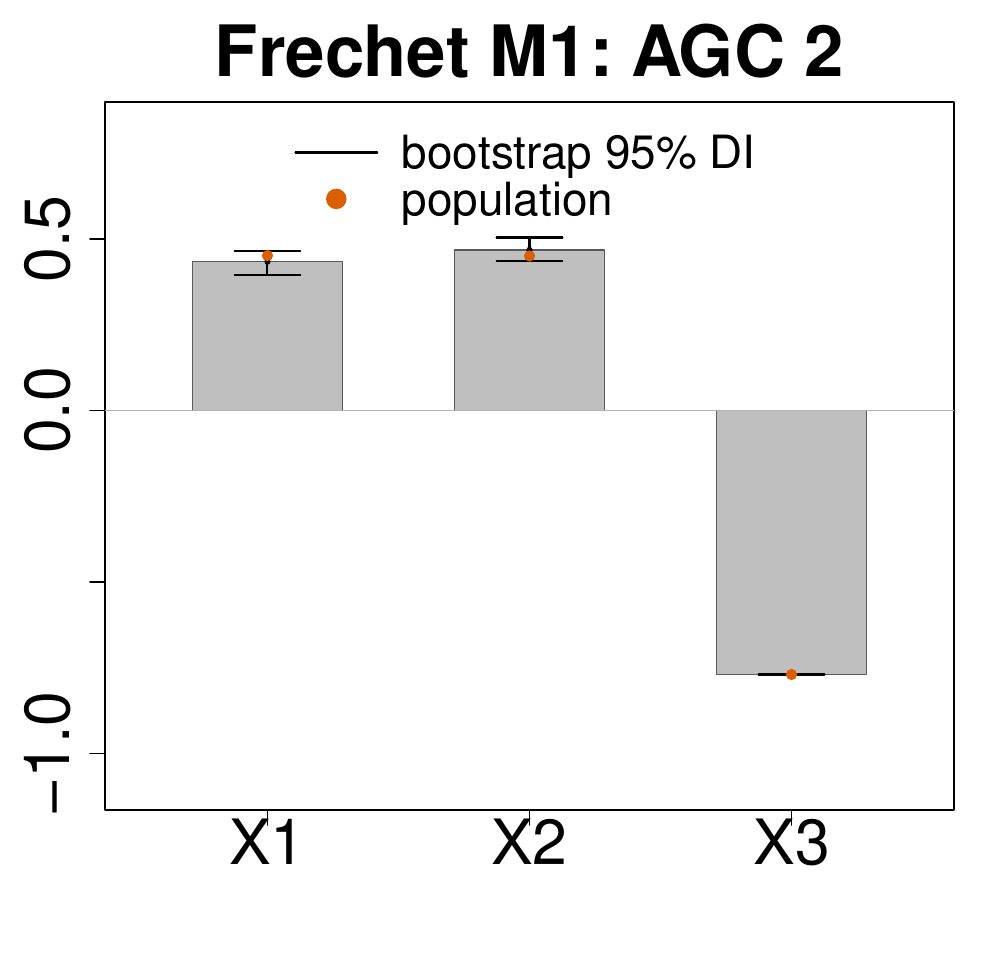}
\end{minipage}\hfill
\begin{minipage}[t]{0.235\linewidth}
\centering
\appinclude[width=\linewidth]{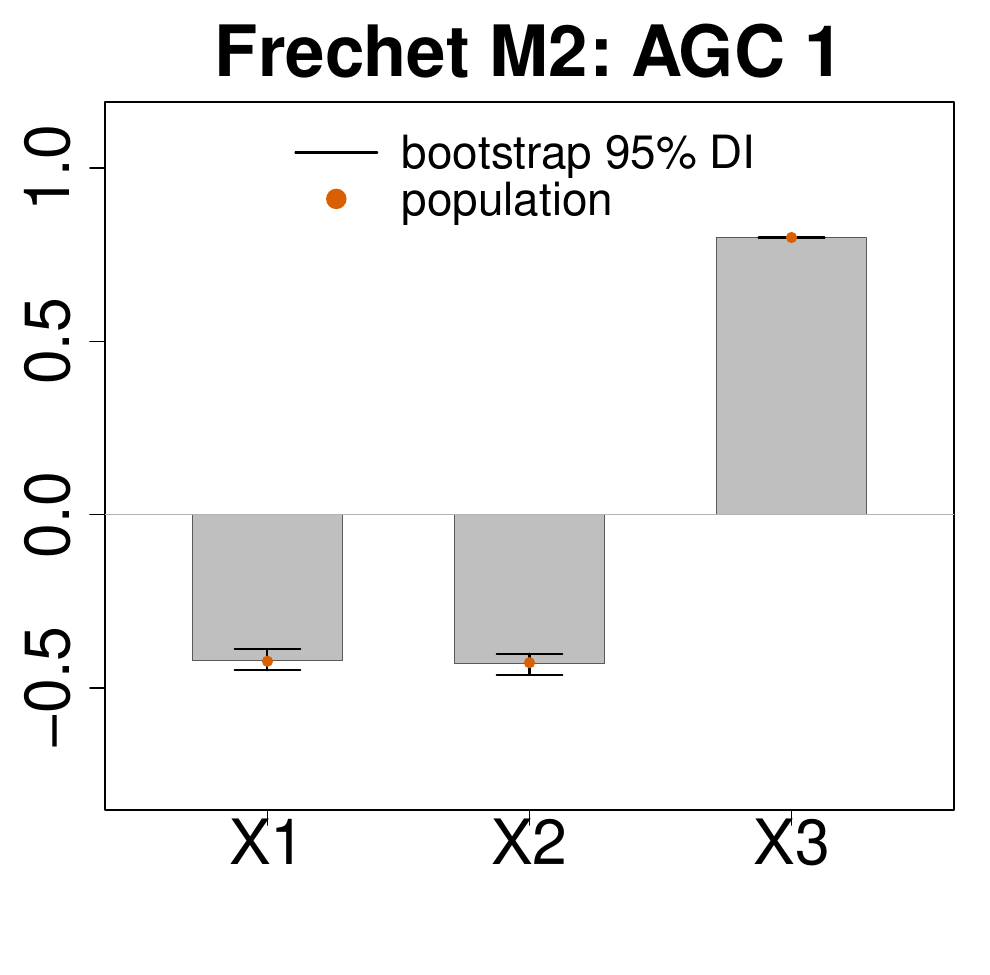}
\end{minipage}\hfill
\begin{minipage}[t]{0.235\linewidth}
\centering
\appinclude[width=\linewidth]{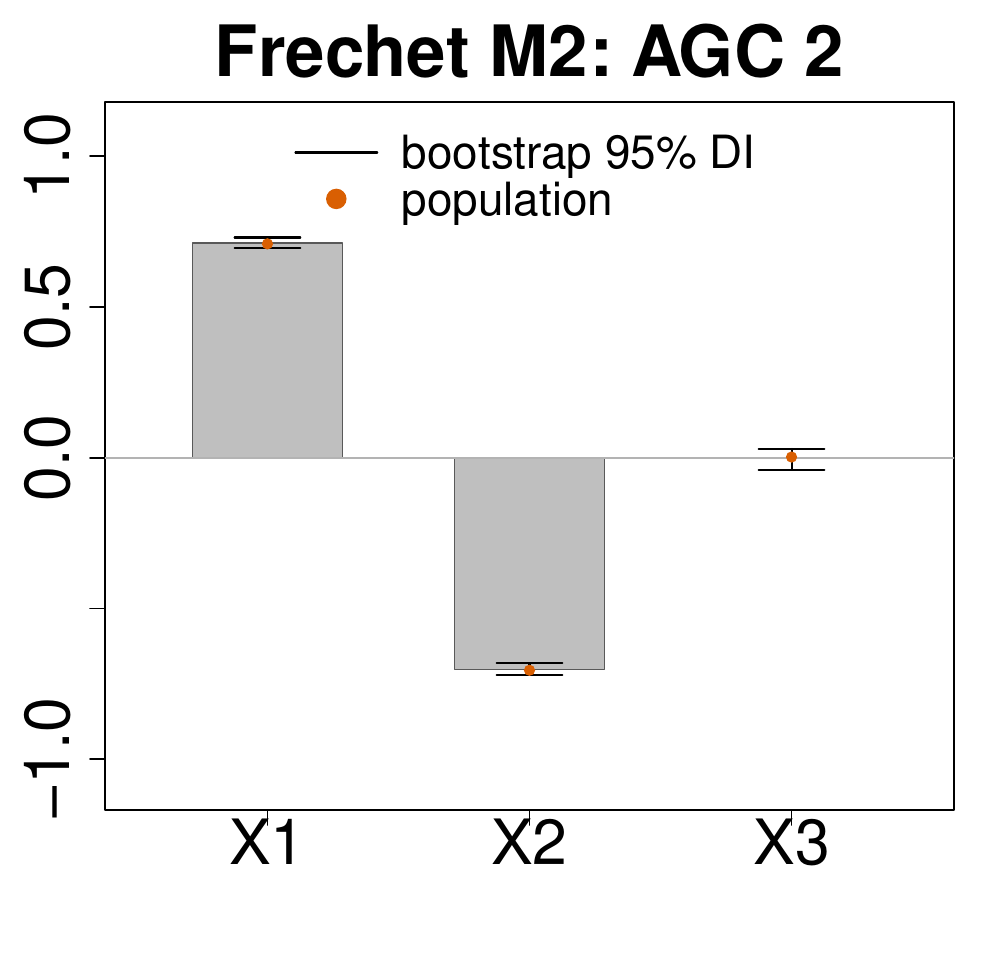}
\end{minipage}
\caption{Fr\'echet-anchor loadings. From left to right: Model~1 AGC~1, Model~1 AGC~2,
Model~2 AGC~1, and Model~2 AGC~2. Bars show sample loadings, vertical intervals show bootstrap
95\% diagnostic intervals (DIs), and dots show population loadings matched to the sample AGCs under the
Fr\'echet anchor.}
\label{fig:app-sim-3d-frechet-loadings}
\end{figure}

\subsection{Finite-sample sensitivity}
\label{app:simulations-finite-sample}

Figures~\ref{fig:app-sim-3d-finite-n-recovery}
and~\ref{fig:app-sim-3d-finite-n-variation-risk} repeat both three-dimensional designs over
increasing \(n\), keeping the tail fraction fixed. The recovery figure contains two diagnostics.
The loading-alignment curves compare the estimated leading loading with the known target contrast:
the \(X_1\)-versus-\(X_2\) contrast in Model~1 and the \(X_3\)-axis contrast in Model~2. The
projector-distance curves compare the estimated leading one-dimensional AGCA space with the
corresponding target space, so smaller values mean better subspace recovery. In both designs the
median alignment is already high at the displayed sample sizes and the projector distance
decreases as \(n\) grows; Model~2 is slightly more variable because the selected sample mixes a
shared \(X_1\)-\(X_2\) regime with a separate near-axis \(X_3\) regime.

\begin{figure}[tbp]
\centering
\appinclude[width=0.82\linewidth]{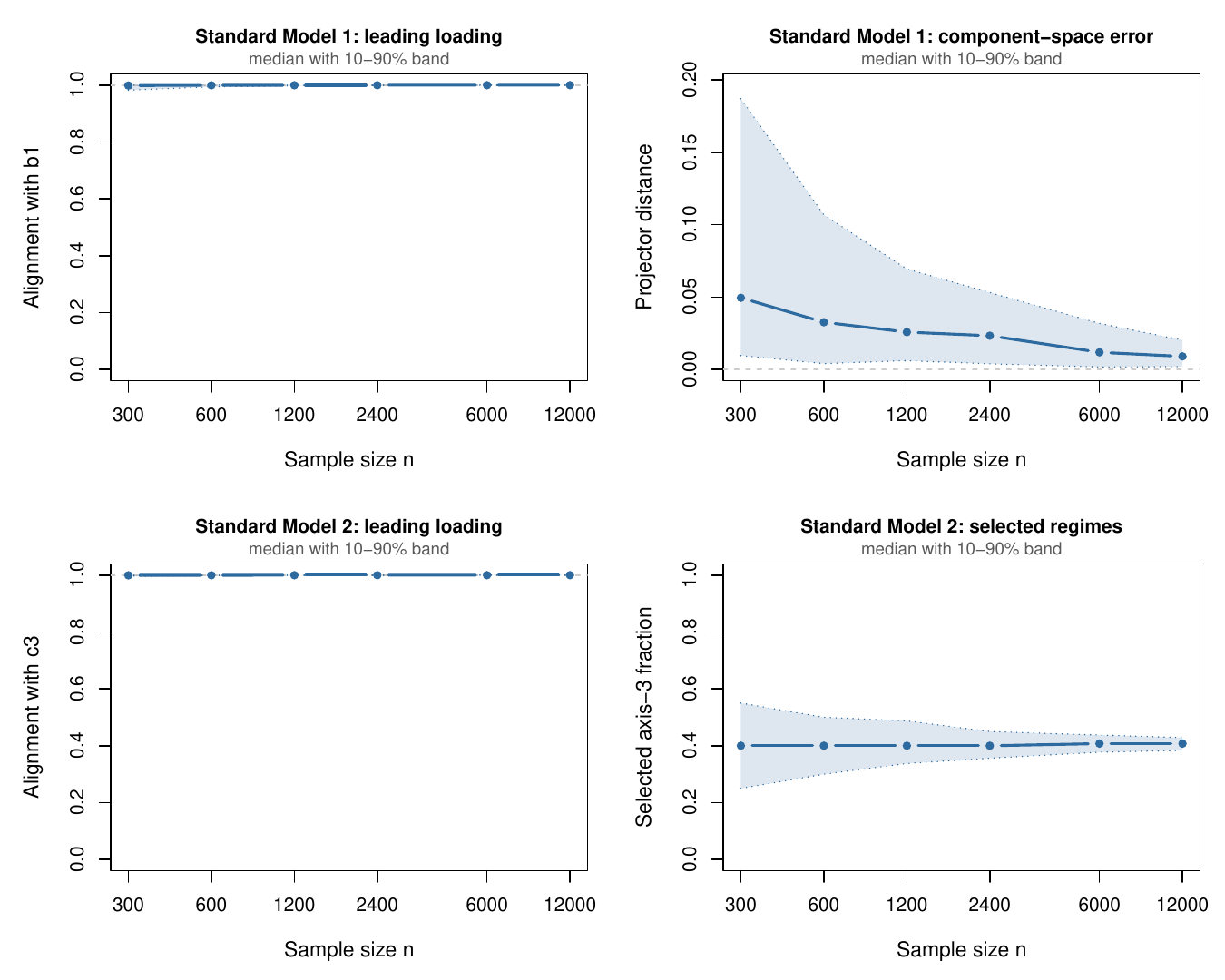}
\caption{Finite-sample recovery diagnostics for the two three-dimensional simulations. Curves
report medians with 10--90\% bands across repeated samples at a fixed tail fraction.}
\label{fig:app-sim-3d-finite-n-recovery}
\end{figure}

Figure~\ref{fig:app-sim-3d-finite-n-variation-risk} reports the scalar fit diagnostics along the
same finite-sample path. The anchored-variation curves show the fraction of total anchored
variation explained by the leading AGC, while the residual-risk curves show the average squared
anchored-coordinate error left after the rank-one fit. These quantities stabilize near the
main-sample values used in the preceding displays: Model~1 has a deliberately nonzero residual
because the angular cloud is slightly bent, whereas Model~2 has high rank-one explained variation
because the dominant contrast separates shared \(X_1\)-\(X_2\) episodes from the \(X_3\)-axis
episodes. The bands shrink with \(n\), confirming that the visual conclusions are not artifacts of
one favorable selected sample.

\begin{figure}[tbp]
\centering
\appinclude[width=0.82\linewidth]{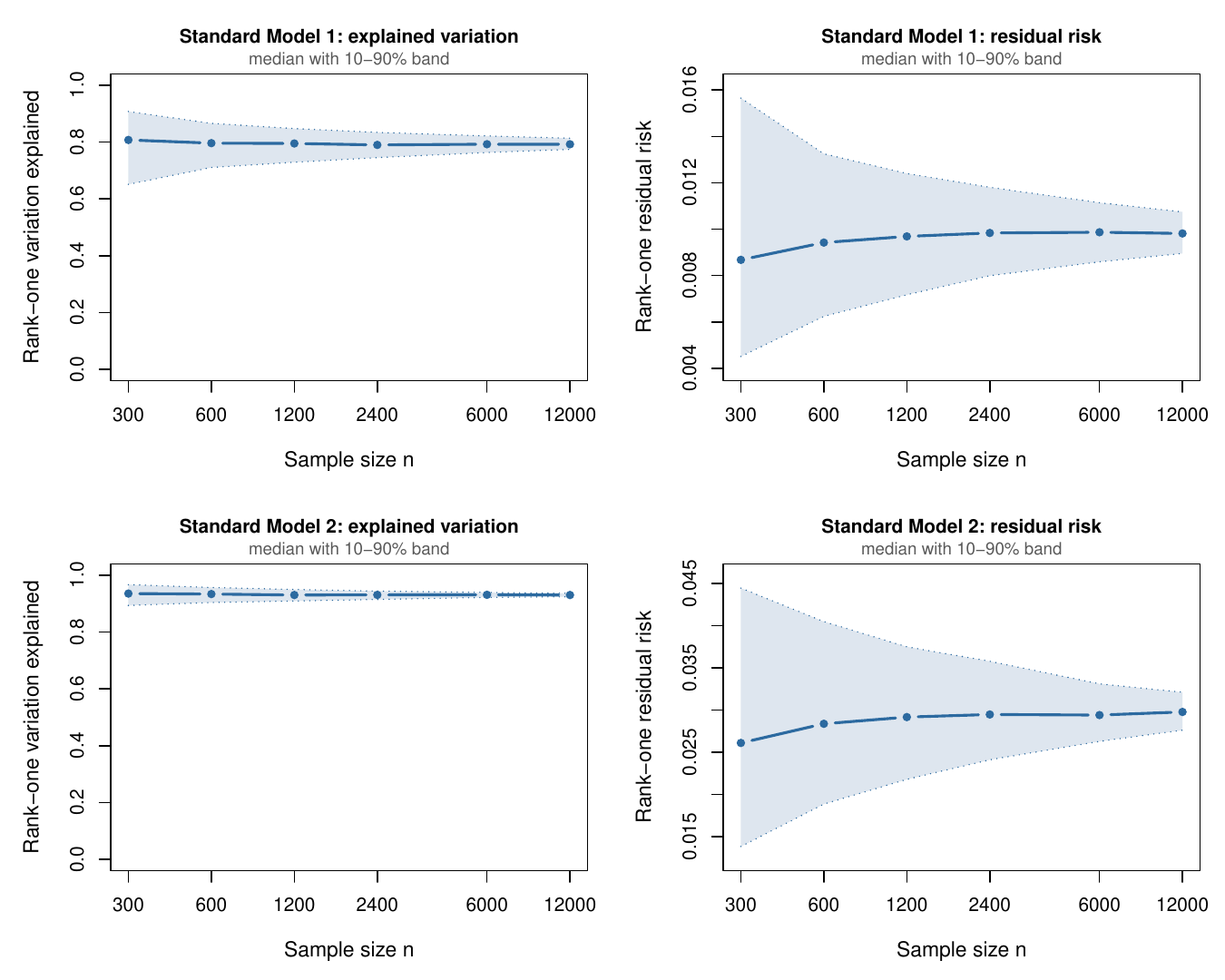}
\caption{Finite-sample anchored variation and residual-risk diagnostics for the two
three-dimensional simulations. Curves report medians with 10--90\% bands across repeated samples
at a fixed tail fraction.}
\label{fig:app-sim-3d-finite-n-variation-risk}
\end{figure}

\subsection{Ten-dimensional design}
\label{app:simulations-10d}

The ten-dimensional design checks that the preceding messages are not artifacts of a visible
three-dimensional sphere. Variables \(X_1,\ldots,X_8\) share a lower-rank angular mechanism,
while \(X_9\) and \(X_{10}\) have asymptotically independent near-axis regimes. After the
rank-Pareto transformation, the first two AGCs are expected to capture the two large
variable-specific regimes, and the next components should recover the lower-amplitude shared
variation among \(X_1,\ldots,X_8\). Figure~\ref{fig:app-sim-10d-variation} shows that the
canonical, Fr\'echet, and principal anchors give qualitatively the same anchored-variation
summary: the first few components dominate, and the ordering of the low-rank fit is stable across
anchors. The data-driven anchors therefore do not overturn the canonical summary, even though two
variables have asymptotically independent tails. The canonical anchor remains preferable for
interpretation because its loadings are departures from balanced complete dependence rather than
from an empirical center.

\begin{figure}[tbp]
\centering
\begin{minipage}[t]{0.48\linewidth}
\centering
\textbf{(A) Anchored variation}\par\medskip
\appinclude[width=\linewidth]{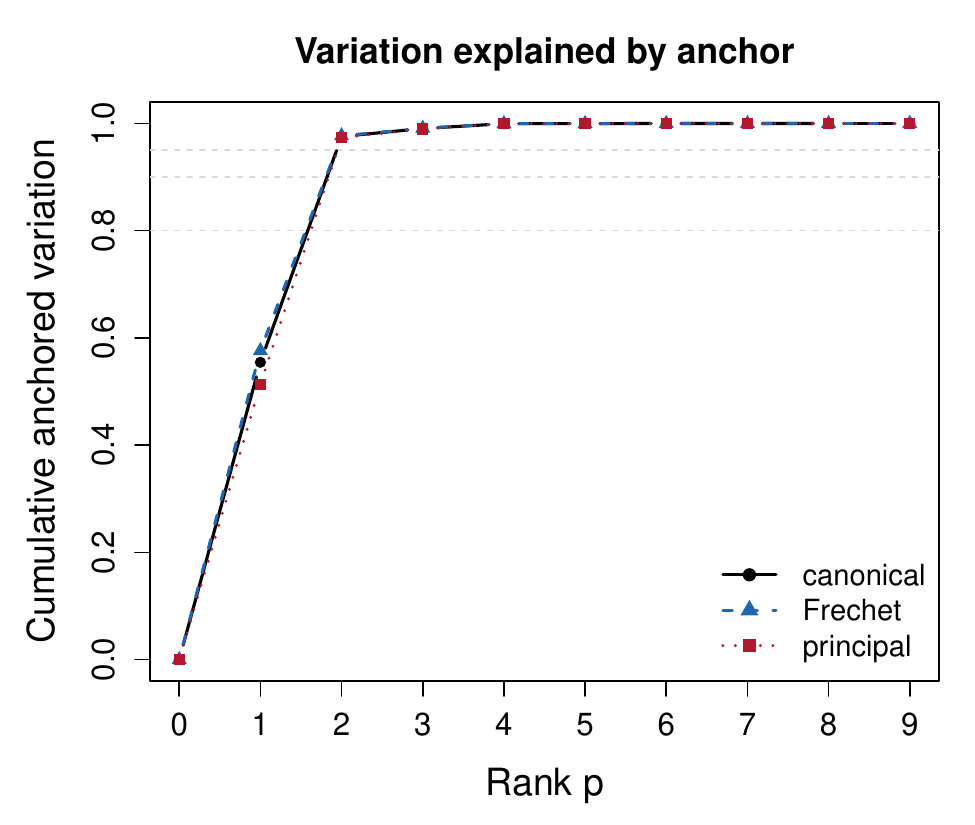}
\end{minipage}\hfill
\begin{minipage}[t]{0.48\linewidth}
\centering
\textbf{(B) Anchor distances}\par\medskip
\appinclude[width=0.7\linewidth]{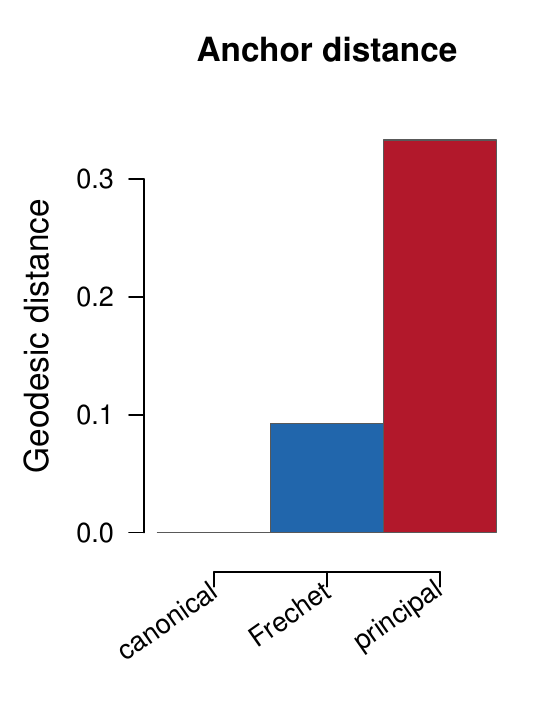}
\end{minipage}
\caption{Ten-dimensional anchor diagnostics. Panel (A) compares cumulative anchored variation
explained under the canonical, Fr\'echet, and principal anchors. Panel (B) reports the spherical
distance from the canonical anchor to the data-adaptive anchors.}
\label{fig:app-sim-10d-variation}
\end{figure}

Figures~\ref{fig:app-sim-10d-loadings} and~\ref{fig:app-sim-10d-loading-bars} give the loading
interpretation. The heatmap shows the leading canonical-anchor loadings: the first two components
concentrate on \(X_9\) and \(X_{10}\), which are the designed near-axis regimes, while the
following components show structured contrasts across \(X_1,\ldots,X_8\). The separate AGC~3 and
AGC~4 bar plots make this lower-amplitude shared structure explicit. Thus AGCA does not ignore the
shared block; it orders it after the larger near-axis departures because components are ranked by
anchored variation.

\begin{figure}[tbp]
\centering
\appinclude[width=0.78\linewidth]{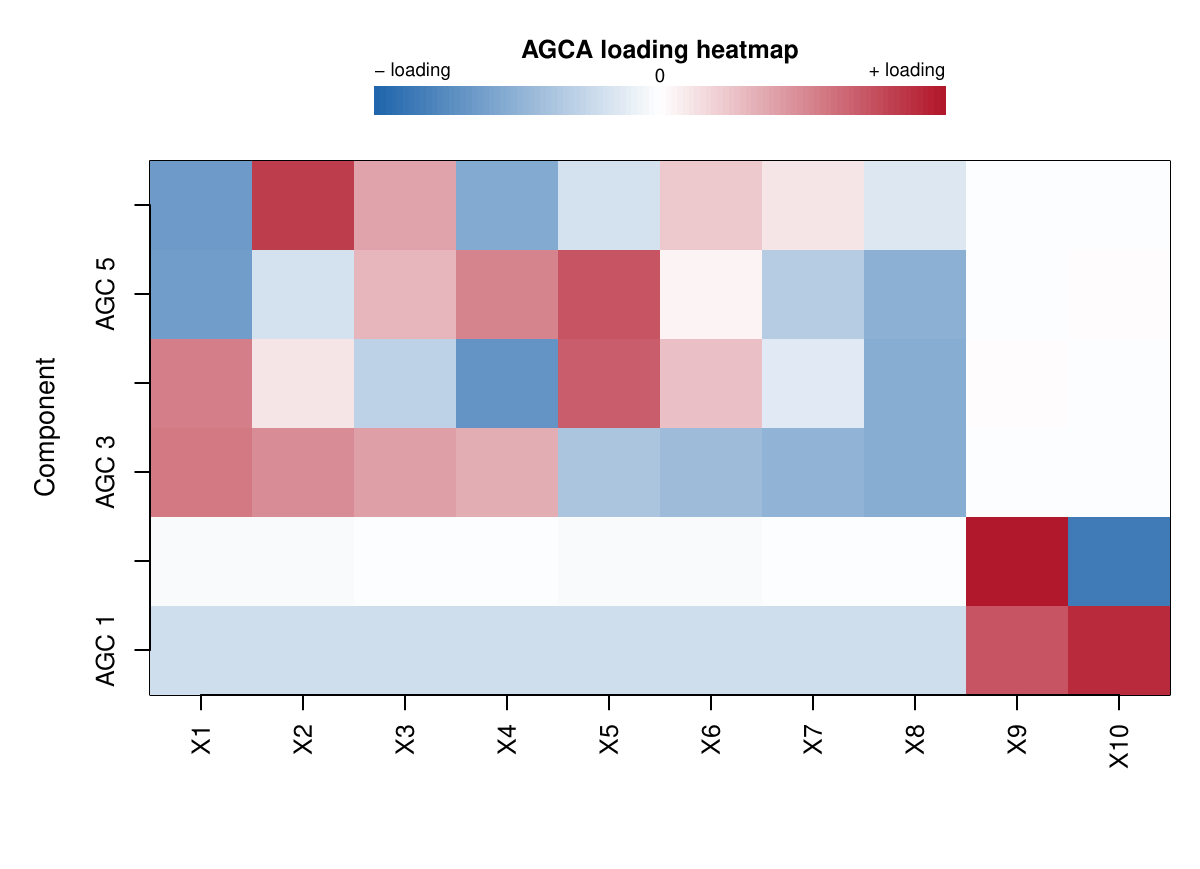}
\caption{Ten-dimensional canonical-anchor loading heatmap. The first two components concentrate
on the designed near-axis regimes, while later components recover structured contrasts in the
shared block \(X_1,\ldots,X_8\).}
\label{fig:app-sim-10d-loadings}
\end{figure}

\begin{figure}[tbp]
\centering
\begin{minipage}[t]{0.45\linewidth}
\centering
\textbf{(A) AGC 3}\par\medskip
\appinclude[width=\linewidth]{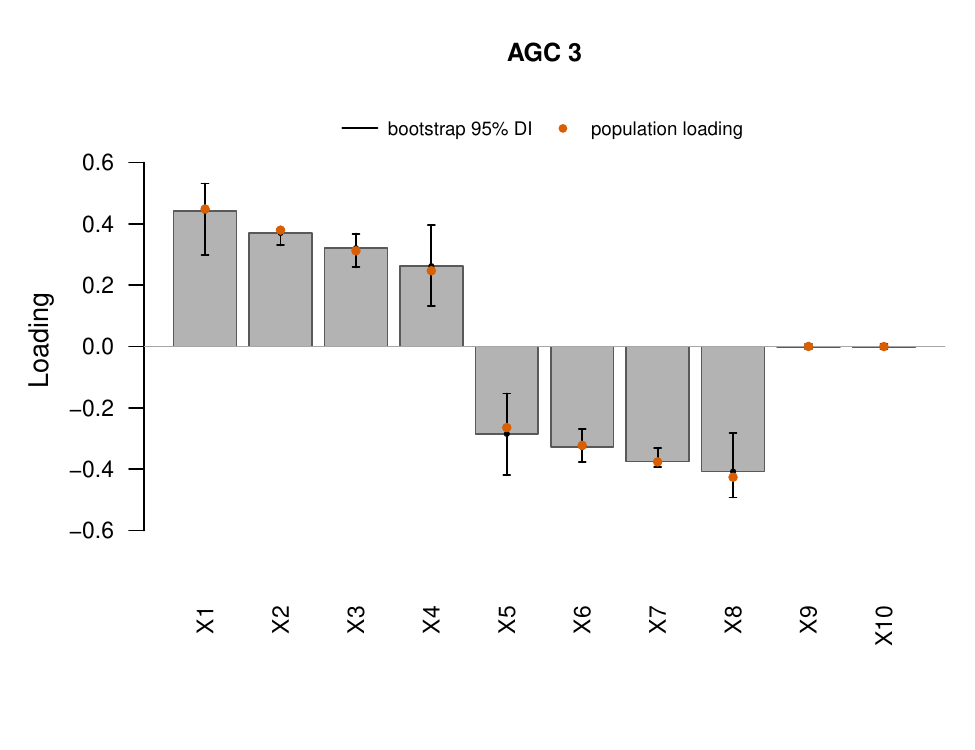}
\end{minipage}\hfill
\begin{minipage}[t]{0.45\linewidth}
\centering
\textbf{(B) AGC 4}\par\medskip
\appinclude[width=\linewidth]{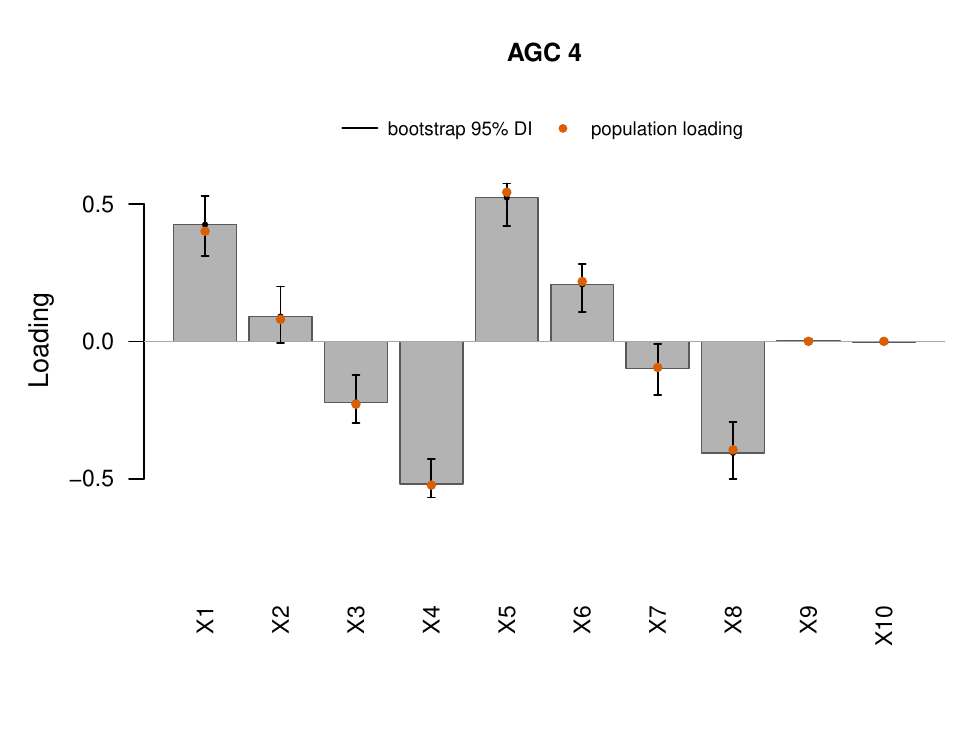}
\end{minipage}
\caption{Ten-dimensional canonical-anchor bar loadings for AGC~3 and AGC~4. Bars show sample
loadings, vertical intervals show bootstrap 95\% diagnostic intervals (DIs), and dots show
population loadings matched to the sample AGCs. These components recover the lower-amplitude
shared angular variation among \(X_1,\ldots,X_8\).}
\label{fig:app-sim-10d-loading-bars}
\end{figure}

Figure~\ref{fig:app-sim-10d-oracle} reports the oracle alignment matrix. The target columns are
known only because this is a simulation. The first two sample AGCs align with the two-dimensional
axis-regime subspace generated by \(X_9\) and \(X_{10}\), up to rotation within that subspace. The
third and fourth sample AGCs align with the two designed shared AGCs for \(X_1,\ldots,X_8\). The
threshold and bootstrap diagnostics in Figure~\ref{fig:app-sim-10d-stability} then show that this
interpretation is not driven by a narrow threshold choice: the explained-variation path and the
rank-two residual risk are stable around the main threshold, and bootstrap DIs for cumulative
variation remain narrow.

\begin{figure}[tbp]
\centering
\appinclude[width=0.62\linewidth]{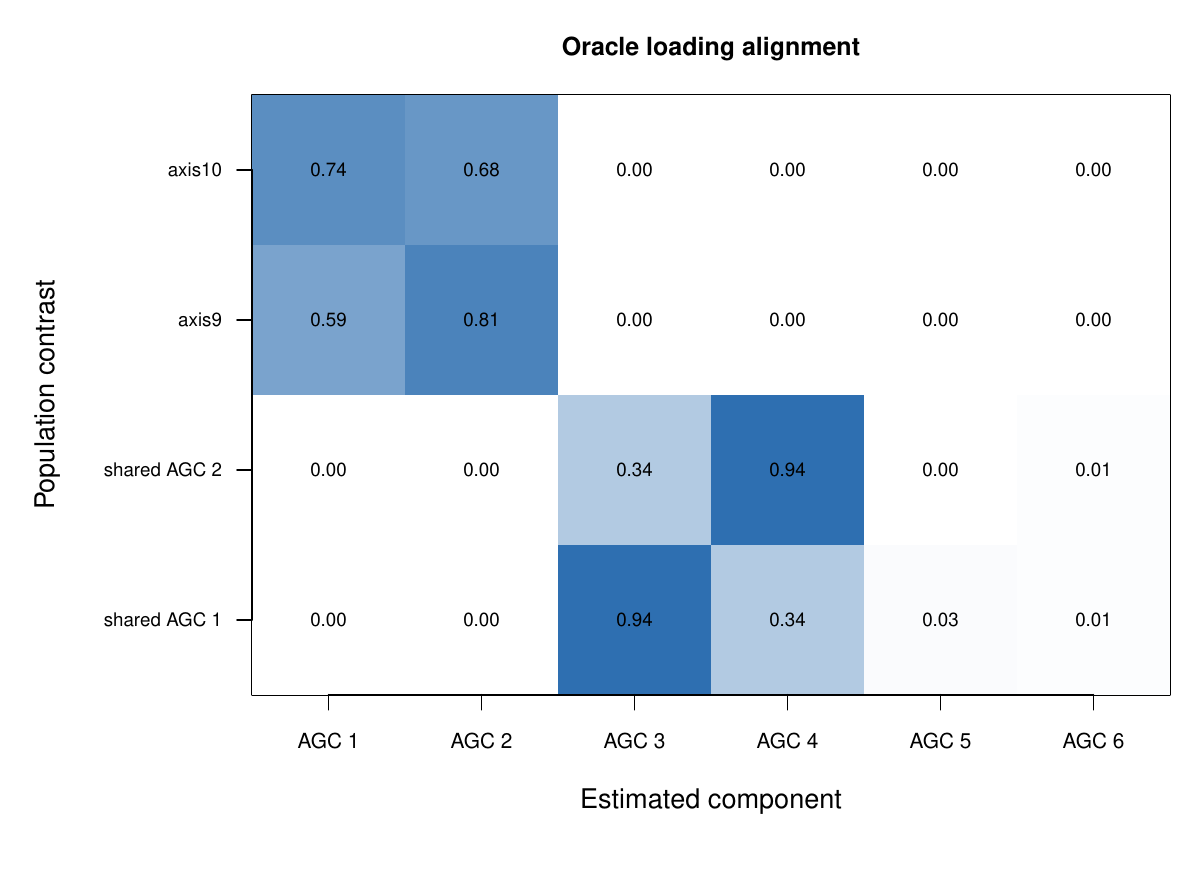}
\caption{Oracle alignment for the ten-dimensional simulation. Rows correspond to sample AGCs and
columns to known design contrasts.}
\label{fig:app-sim-10d-oracle}
\end{figure}

\begin{figure}[tbp]
\centering
\begin{minipage}[t]{0.48\linewidth}
\centering
\textbf{(A) Threshold sensitivity}\par\medskip
\appinclude[width=\linewidth]{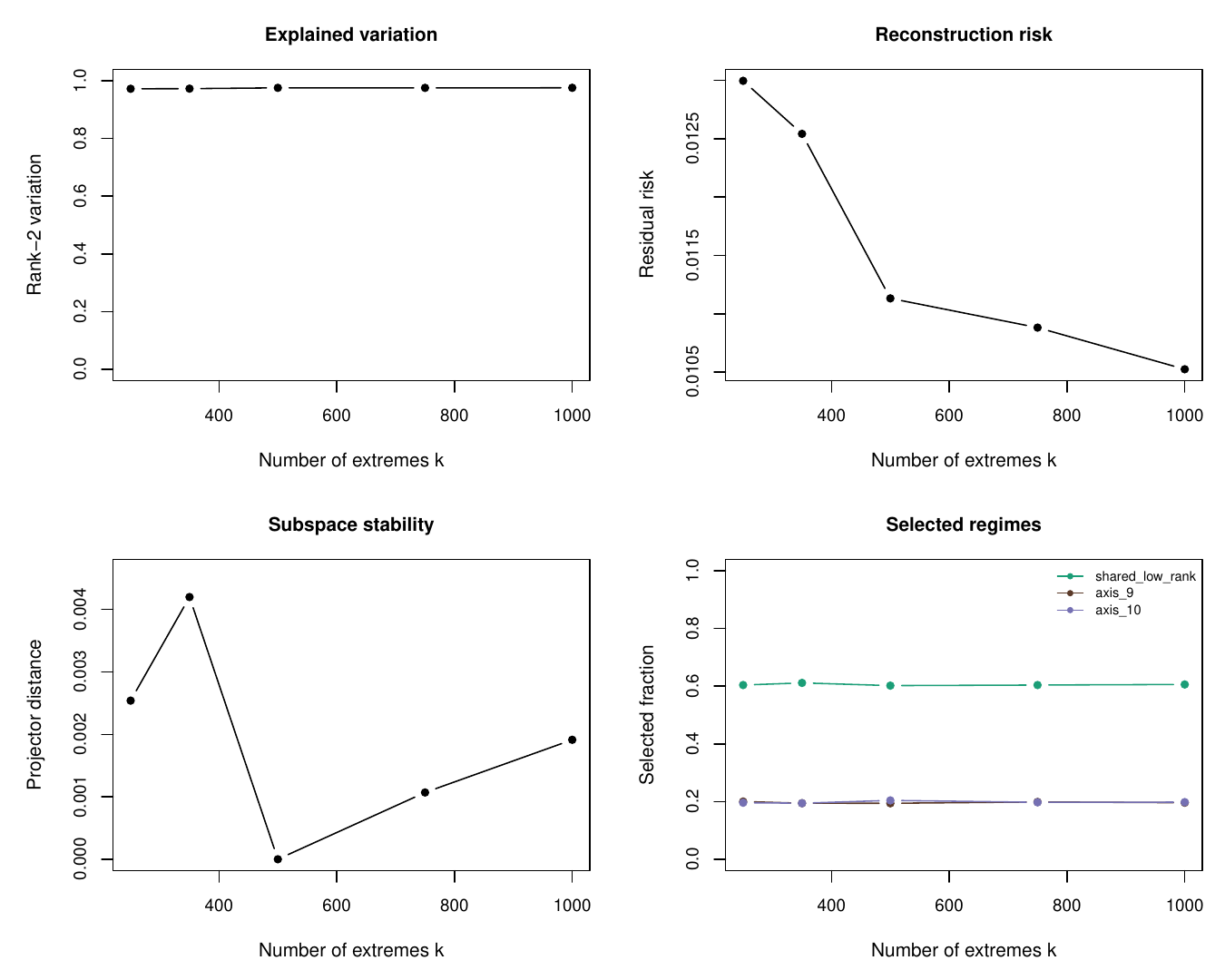}
\end{minipage}\hfill
\begin{minipage}[t]{0.48\linewidth}
\centering
\textbf{(B) Bootstrap DIs}\par\medskip
\appinclude[width=\linewidth]{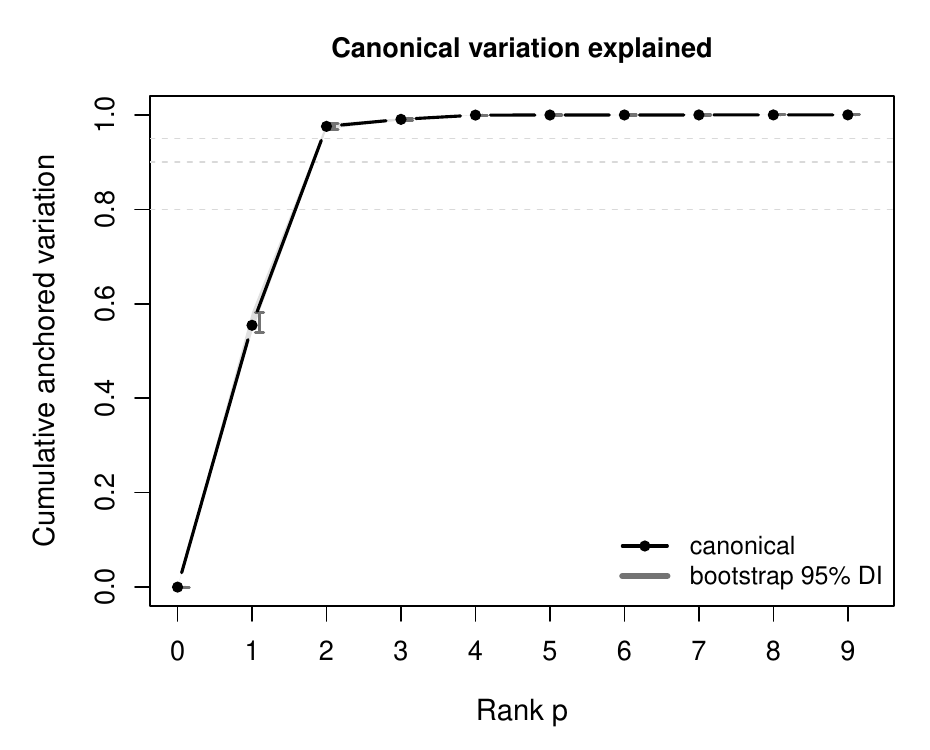}
\end{minipage}
\caption{Ten-dimensional stability diagnostics. Panel (A) varies the number of selected radial
extremes around the main threshold. Panel (B) shows bootstrap 95\% diagnostic intervals (DIs) for
canonical-anchor variation explained.}
\label{fig:app-sim-10d-stability}
\end{figure}

\subsection{Oracle interval coverage in the ten-dimensional design}
\label{app:simulations-10d-oracle-coverage}

We also use the ten-dimensional design to check the finite-sample behavior of the oracle
confidence intervals from Corollary~\ref{cor:oracle-agca-spectral-clt} and
Corollary~\ref{cor:oracle-agca-plugin-ci} in the main text. The goal is not to replace the bootstrap diagnostics
above, but to verify that the bounded-departure oracle variance formulas give reasonable
repeated-sampling calibration in the same setting used for the geometric diagnostics. We run
\(5000\) independent samples with \(n=10000\) and \(k=500\), using the canonical anchor and
rank \(p=2\). The population targets are estimated once from an independent sample of size
\(500000\), with the same tail fraction \(k/n=0.05\). Since the exact finite-sample marginal
distributions of the noisy ten-dimensional generator are not available in closed form, the
oracle-margin analysis uses an independent calibration sample of size \(500000\) to approximate
the true Pareto marginal transforms, with Pareto tail extrapolation beyond the calibration tail.
For comparison, we apply the same oracle interval formulas after the rank-Pareto marginal
standardization used in the main simulation analysis. This second comparison is diagnostic: as
discussed in Section~\ref{sec:rank-pareto-clt-scale}, the oracle variance formula is not claimed
to be a full rank-margin CLT.

\begin{table}[tbp]
\centering
\small
\begin{tabular}{@{}llrrrr@{}}
\toprule
Margins & Target & Population value & Coverage & Bias & Mean length \\
\midrule
Oracle & \(\tau_\mu\) & \(0.4637\) & \(0.941\) & \(0.0028\) & \(0.0577\) \\
Rank & \(\tau_\mu\) & \(0.4637\) & \(1.000\) & \(0.0005\) & \(0.0576\) \\
Oracle & \(\lambda_1\) & \(0.2576\) & \(0.969\) & \(0.0029\) & \(0.0266\) \\
Rank & \(\lambda_1\) & \(0.2576\) & \(1.000\) & \(-0.0003\) & \(0.0211\) \\
Oracle & \(\lambda_2\) & \(0.1942\) & \(0.952\) & \(0.0001\) & \(0.0423\) \\
Rank & \(\lambda_2\) & \(0.1942\) & \(1.000\) & \(-0.0000\) & \(0.0407\) \\
Oracle & \(\mathrm{AVE}_{\mu,2}\) & \(0.9742\) & \(0.931\) & \(0.0005\) & \(0.0099\) \\
Rank & \(\mathrm{AVE}_{\mu,2}\) & \(0.9742\) & \(0.878\) & \(-0.0017\) & \(0.0107\) \\
\bottomrule
\end{tabular}
\caption{Monte Carlo coverage of nominal \(95\%\) oracle plug-in intervals in the ten-dimensional
design. Bias is the Monte Carlo mean estimate minus the population target. In the rank-margin
rows, each sample is first rank-Pareto standardized and then the oracle plug-in standard error is
applied deliberately as a shortcut.}
\label{tab:app-sim-10d-oracle-coverage}
\end{table}

Table~\ref{tab:app-sim-10d-oracle-coverage} shows that the oracle intervals are reasonably
calibrated for the formal oracle problem. Coverage is close to the nominal level for
\(\lambda_2\), slightly conservative for \(\lambda_1\), and mildly below nominal for
\(\tau_\mu\) and \(\mathrm{AVE}_{\mu,2}\). The biases are small relative to the interval lengths,
so the remaining discrepancies are consistent with finite-threshold and plug-in variance effects
rather than a breakdown of the bounded-departure CLT. The rank-margin comparison is informative
but should be read differently. For \(\tau_\mu\), \(\lambda_1\), and \(\lambda_2\), the oracle
formula applied to rank margins overcovers: the rank transformation fixes much of the marginal
tail-scale variation across samples, while the oracle plug-in standard errors remain comparatively
large. For \(\mathrm{AVE}_{\mu,2}\), however, the same shortcut undercovers, with coverage
\(0.878\). This ratio functional is sensitive to the joint rank-margin correction in the numerator
and denominator, and the result illustrates why the paper uses bootstrap and stability diagnostics 
for rank-standardized empirical
analyses. At the same time, the oracle-margin results support the main inferential message:
for this high-dimensional design with near-axis regimes, AGCA's bounded spherical departures lead
to usable first-order uncertainty quantification for the oracle spectral summaries.

Together, the simulations isolate the core geometric messages. In Model~1, AGCA summarizes a
smooth, slightly bent angular law by a stable anchored direction plus a smaller residual
direction. In Model~2, AGCA summarizes an asymptotically independent variable by a finite
near-axis displacement from the balanced anchor. The same diagnostics used in the empirical
application (explained variation, scores, loadings, anchor checks, bootstrap DIs, and
threshold or finite-sample paths) are already visible in this low-dimensional setting and remain
informative in the ten-dimensional design with two near-axis regimes.

\FloatBarrier

%% file: appendix_empirics_portfolios.tex
\section{Portfolio empirical analysis: additional details}
\label{app:empirics-portfolios}

This supplementary section gives the data-construction details and secondary diagnostics for
Section~\ref{sec:empirics-portfolios}. All portfolio-return panels are converted to losses by
multiplying returns by \(-1\). Each margin is then transformed to the pseudo-Pareto scale using
\eqref{eq:pseudo-pareto-transform}, and AGCA is fit to Euclidean-spherical directions for the
largest radii. Unless otherwise stated, the anchor is the canonical complete-dependence
direction \(\mu_0=d^{-1/2}\one\).

\subsection{Data construction and anomaly-panel robustness design}
\label{app:emp-port-data}
\label{app:emp-port-osap-robustness}

The main Fama--French panel uses daily value-weighted \(2\times 3\) portfolios from the Kenneth
R. French Data Library \citep{kenFrenchDataLibrary}. The four included sorts are
Size--Book-to-Market, Size--Operating Profitability, Size--Investment, and Size--Momentum. Each
sort contributes six portfolios, corresponding to two size groups crossed with three
characteristic groups. The retained complete panel has \(d=24\) assets and \(n=13340\) daily
observations from July 2, 1973 through May 29, 2026. The main threshold keeps the largest
\(k=667\) radii.

The robustness panel uses the Open Source Asset Pricing daily value-weighted anomaly-sorted
quintile portfolio file \citep{openSourceAssetPricing}. We use six continuous signals related to
liquidity and trading frictions: bid--ask spread, illiquidity, dollar volume, one-month zero
trading, six-month zero trading, and price. Each signal contributes five quintile portfolios,
giving \(d=30\) assets. The retained complete panel has \(n=12988\) daily observations from July
2, 1973 through December 31, 2024. The main threshold keeps the largest \(k=649\) radii.

\begin{table}[H]
\centering
\caption{Portfolio blocks in the two empirical panels.}
\label{tab:app-emp-port-blocks}
\small
\begin{tabular}{@{}lll@{}}
\toprule
Panel & Block & Portfolios retained \\
\midrule
FF & Size--Book-to-Market & six \(2\times 3\) value-weighted portfolios \\
FF & Size--Operating Profitability & six \(2\times 3\) value-weighted portfolios \\
FF & Size--Investment & six \(2\times 3\) value-weighted portfolios \\
FF & Size--Momentum & six \(2\times 3\) value-weighted portfolios \\
\addlinespace
OSAP & Bid--ask spread & five value-weighted quintile portfolios \\
OSAP & Illiquidity & five value-weighted quintile portfolios \\
OSAP & Dollar volume & five value-weighted quintile portfolios \\
OSAP & Zero trading, one month & five value-weighted quintile portfolios \\
OSAP & Zero trading, six months & five value-weighted quintile portfolios \\
OSAP & Price & five value-weighted quintile portfolios \\
\bottomrule
\end{tabular}
\end{table}

The Fama--French panel is the main application because it gives a compact and familiar
portfolio-formation universe with clear economic blocks. The OSAP panel is deliberately
different: it is a modern anomaly-sorted trading-friction panel, distinct from a standard
size-characteristic panel. Its role is to test whether the portfolio-functional conclusions
survive in another long daily portfolio universe. The robustness evidence is reported below in
the cross-panel spectrum, OSAP loading heatmap, and OSAP portfolio-functional diagnostics.

\subsection{AGCA spectrum diagnostics}
\label{app:emp-port-diagnostics}

Figure~\ref{fig:app-emp-port-spectrum-panels} reports the main canonical-anchor spectra for both
portfolio panels. It is included here as the cross-panel version of the main-text Fama--French
spectrum: the Fama--French panel is more compressed, while the OSAP anomaly panel gives the same
qualitative low-rank message in a distinct portfolio universe.

\begin{figure}[H]
\centering
\appinclude[width=0.95\linewidth]{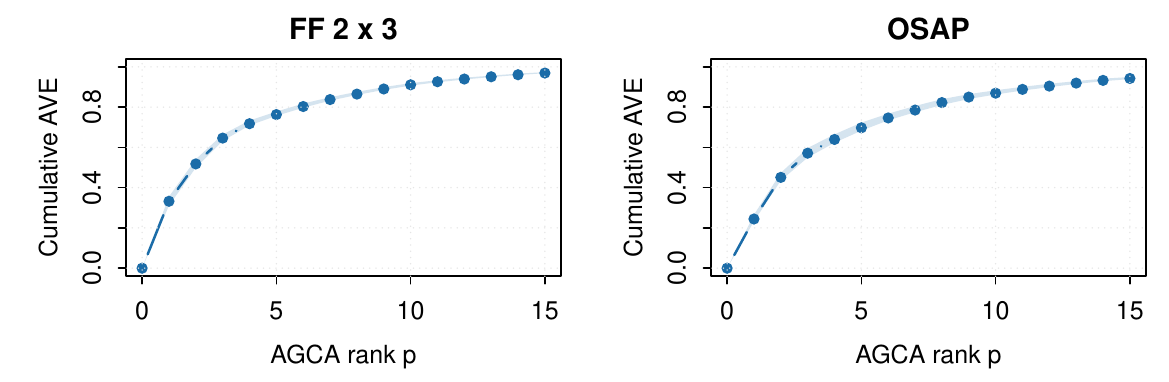}
\caption{Cumulative anchored variation explained for the Fama--French and OSAP portfolio panels.
The line is the canonical-anchor fit at the main \(5\%\) radial threshold. Shaded bands are
bootstrap 95\% diagnostic intervals, computed as pointwise \(2.5\%\) and \(97.5\%\)
conditional bootstrap percentiles after resampling the selected extreme directions and refitting
AGCA.}
\label{fig:app-emp-port-spectrum-panels}
\end{figure}

Figure~\ref{fig:app-emp-port-anchor} reports anchor sensitivity for the effective-dimension
summary. Besides the canonical anchor, we fit AGCA using the leading principal direction of the
selected angular sample and the normalized selected-sample mean direction. These anchors are
diagnostic references; the main text uses the canonical anchor. The canonical anchor keeps the
interpretation of loadings as departures from balanced complete dependence.

\begin{figure}[H]
\centering
\appinclude[width=0.95\linewidth]{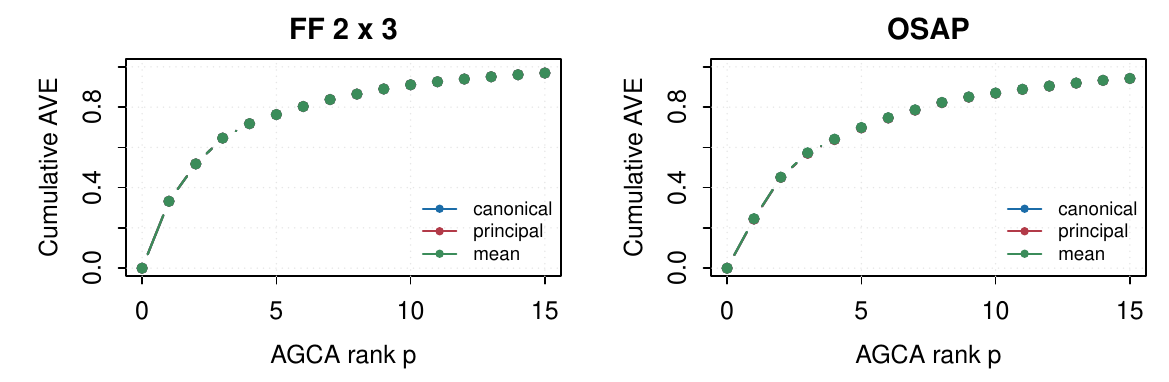}
\caption{Anchor sensitivity for cumulative anchored variation explained. Each panel reports the
canonical, principal, and mean-anchor spectra at the main \(5\%\) radial threshold.}
\label{fig:app-emp-port-anchor}
\end{figure}

Figure~\ref{fig:app-emp-port-anchor-distance} puts this diagnostic in geometric perspective by
reporting the spherical distance from each data-driven anchor to the canonical anchor over the
same threshold grid used below. The data-driven anchors move modestly but measurably away from
balanced complete dependence, especially at smaller thresholds in the OSAP panel, so the near
overlap in Figure~\ref{fig:app-emp-port-anchor} is not merely a numerical identity between
anchor choices.

\begin{figure}[H]
\centering
\appinclude[width=0.95\linewidth]{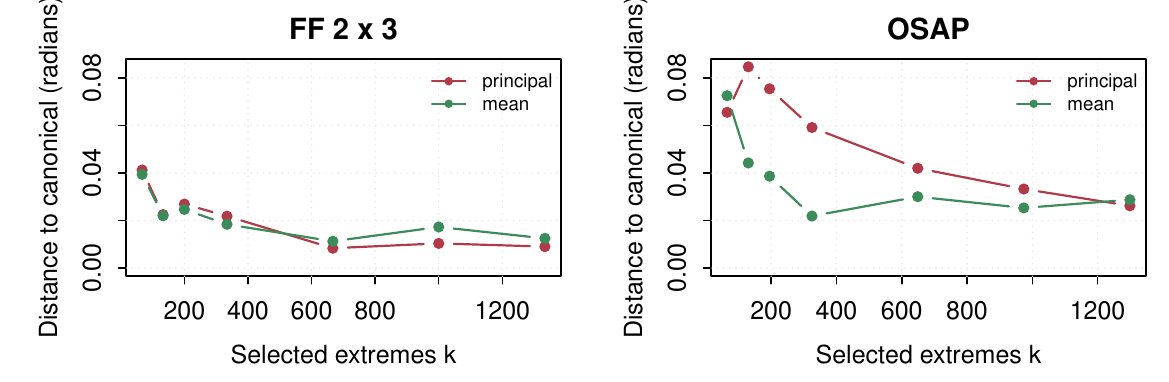}
\caption{Spherical distance from the canonical anchor to the principal and selected-sample mean
anchors over the threshold grid. Each panel reports distances for one portfolio panel.}
\label{fig:app-emp-port-anchor-distance}
\end{figure}

The scalar spectra are insensitive to these anchor changes. In the Fama--French panel, rank ten
explains about \(91\%\) of anchored variation under all three anchors. In the OSAP panel, the
rank-twelve explained variation remains close to \(90\%\). Thus the low-rank conclusion remains
stable across anchor choices, even though the diagnostic anchors move away from the canonical
anchor by nonzero spherical distances.

Figure~\ref{fig:app-emp-port-threshold} reports threshold sensitivity. For each panel we refit
the canonical-anchor AGCA model over a grid of radial tail fractions from \(0.5\%\) to
\(10\%\), keeping only thresholds with more than \(d+2\) selected directions. The expected
bias-variance tradeoff is visible: the smallest \(k\) curves are less smooth, while larger
thresholds admit less extreme days. The main empirical message nevertheless persists. The
Fama--French panel remains more compressed than the OSAP panel, and both retain a clearly
ordered low-rank spectrum across the displayed threshold choices.

\begin{figure}[H]
\centering
\appinclude[width=0.95\linewidth]{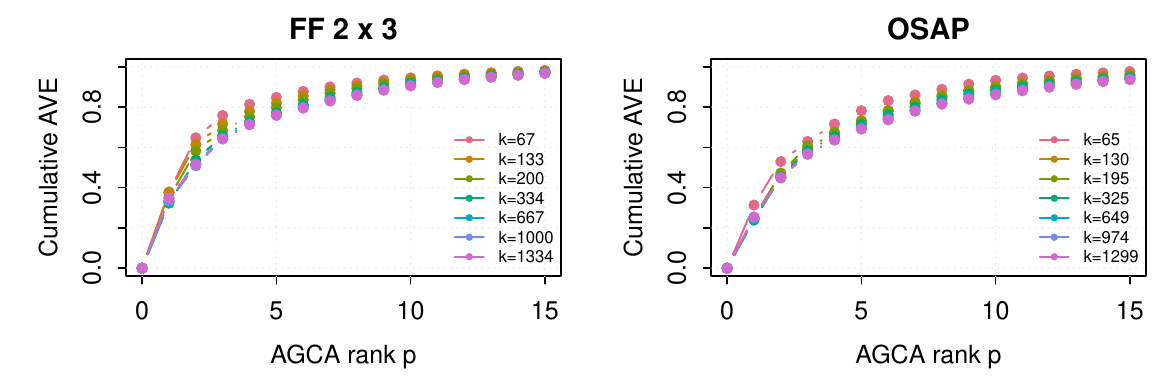}
\caption{Threshold sensitivity for cumulative anchored variation explained. Each line
corresponds to a different number \(k\) of selected radial extremes.}
\label{fig:app-emp-port-threshold}
\end{figure}

\subsection{Loading heatmaps and AGC1 stability}
\label{app:emp-port-loadings}

The main text reports the Fama--French loading heatmap. Figure~\ref{fig:app-emp-port-osap-loadings}
gives the corresponding OSAP display. In both heatmaps, assets are on the horizontal axis and
AGCA loading vectors are the rows. Signs are oriented by the deterministic convention used in
the implementation: within each loading vector, the entry with largest absolute magnitude is
positive. This is the same orientation convention used in the main text.

\begin{figure}[H]
\centering
\appinclude[width=0.95\linewidth]{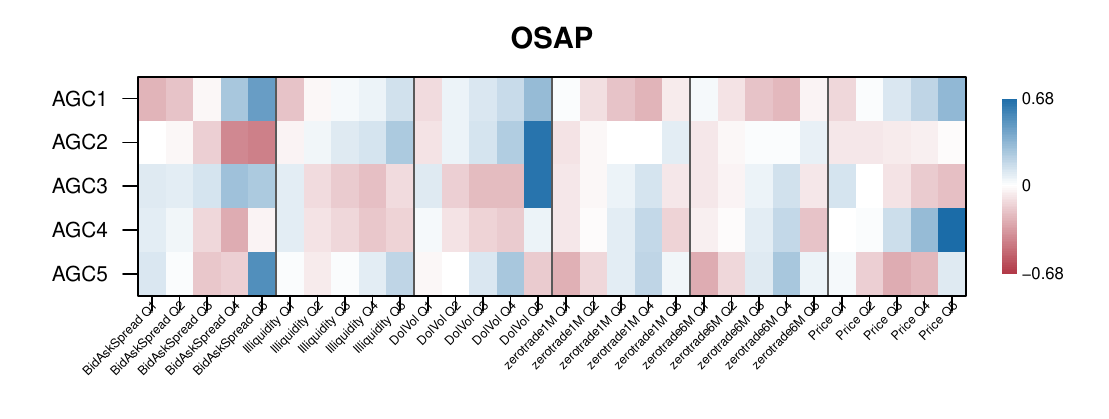}
\caption{Canonical-anchor AGCA loadings for the OSAP robustness panel. Columns are portfolios,
grouped by signal, and rows are the first five AGCA loading vectors. Loading signs are oriented
so that the largest absolute entry in each loading vector is positive.}
\label{fig:app-emp-port-osap-loadings}
\end{figure}

The OSAP heatmap is useful mainly as a robustness diagnostic. Because the assets are quintile
portfolios formed on liquidity and trading-friction signals, the loadings should be read as
contrasts across trading-friction portfolios, with the AGCA interpretation tied to angular tail
directions. The important comparison with the Fama--French panel is the preservation of both
angular variation and the portfolio functionals reported below with a similar number of AGCA
coordinates. AGC1 has a simple gradient interpretation for several OSAP signals: bid--ask
spread, dollar volume, and price move from negative loadings in the lower quintiles to positive
loadings in the upper quintiles. The first coordinate therefore captures a signal-sorted
rotation of the extreme-loss direction, analogous to the size contrast in the Fama--French
panel.

\begin{figure}[H]
\centering
\appinclude[width=0.95\linewidth]{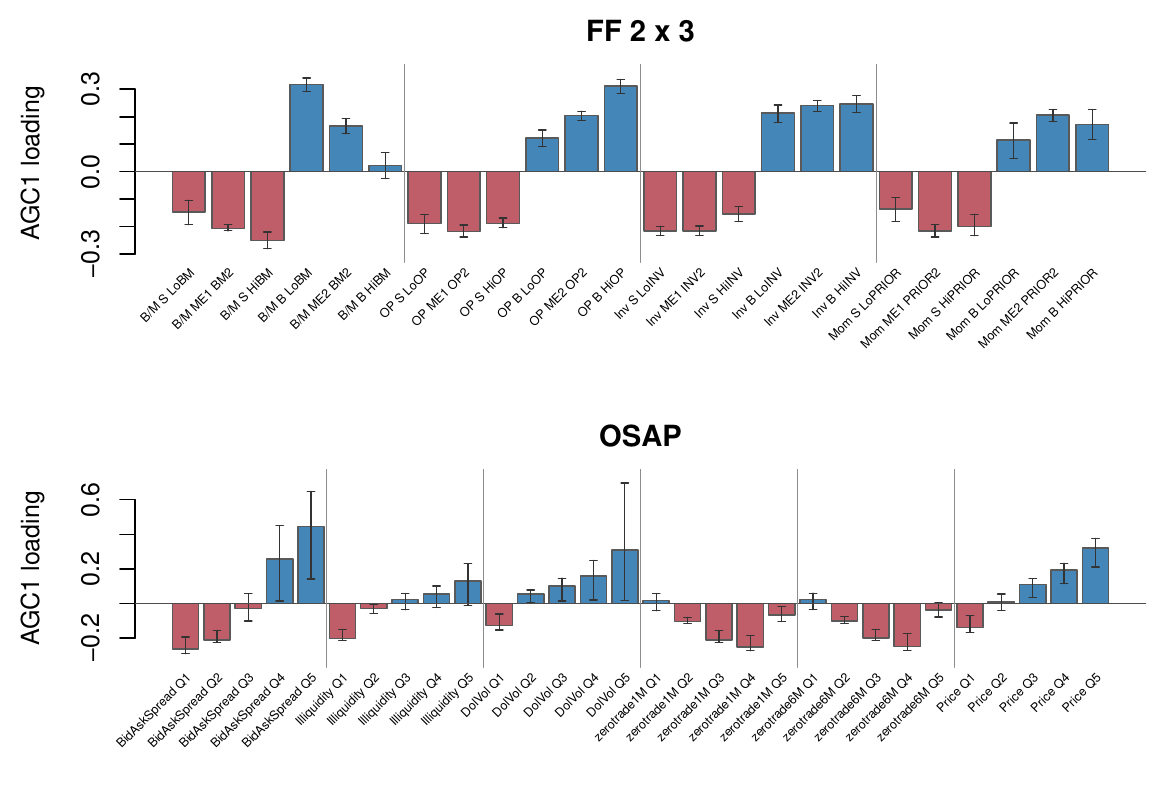}
\caption{AGC1 loadings with pointwise conditional bootstrap 95\% diagnostic intervals. Bars
are the canonical-anchor AGC1 loadings at the main \(5\%\) radial threshold. DIs are \(2.5\%\) and
\(97.5\%\) bootstrap percentiles after resampling the selected extreme directions and aligning
bootstrap loading signs to the full-sample AGC1.}
\label{fig:app-emp-port-agc1-bootstrap-loadings}
\end{figure}

Figure~\ref{fig:app-emp-port-agc1-bootstrap-loadings} checks whether the AGC1 signs are stable
under resampling of the selected extreme directions. The DIs are diagnostic rather than
formal confidence intervals, because they condition on the selected tail sample and on the
rank-Pareto marginal transformation. In the Fama--French panel, 23 of 24 DIs exclude
zero. The stable signs confirm that AGC1 is a repeated size contrast across the four bivariate
sorts; the only near-zero case is the big high book-to-market portfolio. In the OSAP panel, 22
of 30 DIs exclude zero. The most stable gradients occur for bid--ask spread, dollar
volume, and price, while several middle illiquidity and zero-trading quintiles have DIs
that cross zero. These crossings identify entries whose signs should be read as weak local
features of AGC1 rather than robust loadings.

\subsection{Portfolio tail-functional diagnostics}
\label{app:emp-port-functionals}

The portfolio-functional experiment uses four classes of weights. The first is the equal-weight
portfolio across all assets in the panel. The second contains block-equal portfolios, one per
Fama--French sort or OSAP signal. The third contains 250 randomly generated long-only portfolios
with weights summing to one. The fourth contains 250 randomly generated limited-leverage
portfolios with \(\one^\T w=1\) and \(\norm{w}_1\le1.5\). These random portfolios are fixed
across ranks within each data set.

The capped-excess functional is bounded and Lipschitz by construction. The VaR diagnostic uses
the Pareto-limit tail constant, with Pareto index \(\alpha=1\), matching the rank-Pareto
marginal scale. The main text reports rank-level average relative errors for the Fama--French
panel. Figure~\ref{fig:app-emp-port-ff-class} gives the Fama--French class-level breakdown, and
Figure~\ref{fig:app-emp-port-relative-rank-osap} adds the OSAP robustness panel on the same
relative-error scale.

\begin{figure}[H]
\centering
\appinclude[width=0.95\linewidth]{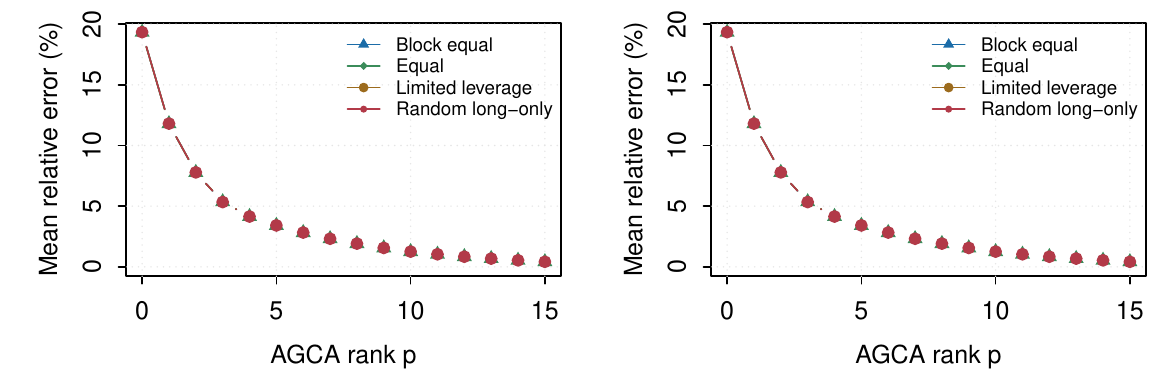}
\caption{Fama--French relative portfolio-functional errors by portfolio class. The left display
reports relative capped-excess error and the right display reports relative normalized VaR
error.}
\label{fig:app-emp-port-ff-class}
\end{figure}

\begin{figure}[H]
\centering
\appinclude[width=0.95\linewidth]{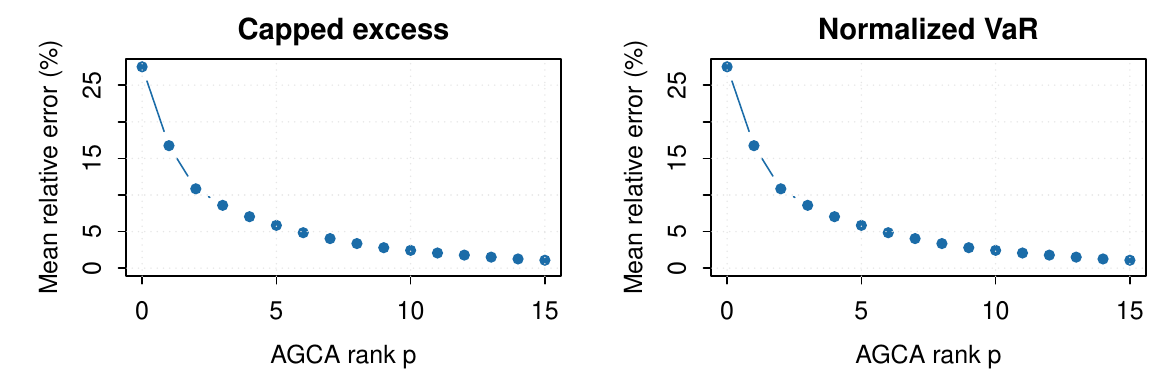}
\caption{OSAP relative portfolio tail-functional errors by AGCA rank. The left display reports
relative capped-excess errors; the right display reports relative normalized VaR errors.}
\label{fig:app-emp-port-relative-rank-osap}
\end{figure}

Table~\ref{tab:app-emp-port-functionals} reports selected-rank values on both the original error
scales and the relative financial scale used in the main text.

\begin{table}[H]
\centering
\caption{Mean portfolio-functional errors at selected AGCA ranks.}
\label{tab:app-emp-port-functionals}
\small
\begin{tabular}{@{}lrrrrrr@{}}
\toprule
Panel & Rank & AVE & Capped abs. & Capped rel. (\%) & Log VaR & VaR rel. (\%) \\
\midrule
FF & 3  & 0.646 & 0.0129 & 5.36 & 0.0522 & 5.36 \\
FF & 5  & 0.763 & 0.0082 & 3.41 & 0.0335 & 3.41 \\
FF & 8  & 0.864 & 0.0046 & 1.91 & 0.0189 & 1.91 \\
FF & 10 & 0.911 & 0.0030 & 1.25 & 0.0124 & 1.25 \\
FF & 12 & 0.939 & 0.0020 & 0.84 & 0.0083 & 0.84 \\
\addlinespace
OSAP & 3  & 0.571 & 0.0192 & 8.59 & 0.0824 & 8.59 \\
OSAP & 5  & 0.697 & 0.0131 & 5.85 & 0.0569 & 5.85 \\
OSAP & 8  & 0.822 & 0.0075 & 3.36 & 0.0331 & 3.36 \\
OSAP & 10 & 0.869 & 0.0054 & 2.44 & 0.0241 & 2.44 \\
OSAP & 12 & 0.904 & 0.0040 & 1.79 & 0.0177 & 1.79 \\
\bottomrule
\end{tabular}
\end{table}

Figure~\ref{fig:app-emp-port-absolute-rank-all} reports the original absolute capped-excess and
log-VaR scales, including the positive post-projection. The post-projected curves are close to
the unconstrained AGCA curves, so enforcing positive directions for simulation does not change
the rank-level conclusion in these data.

\begin{figure}[H]
\centering
\begin{minipage}[t]{0.49\linewidth}
\centering
\appinclude[width=\linewidth]{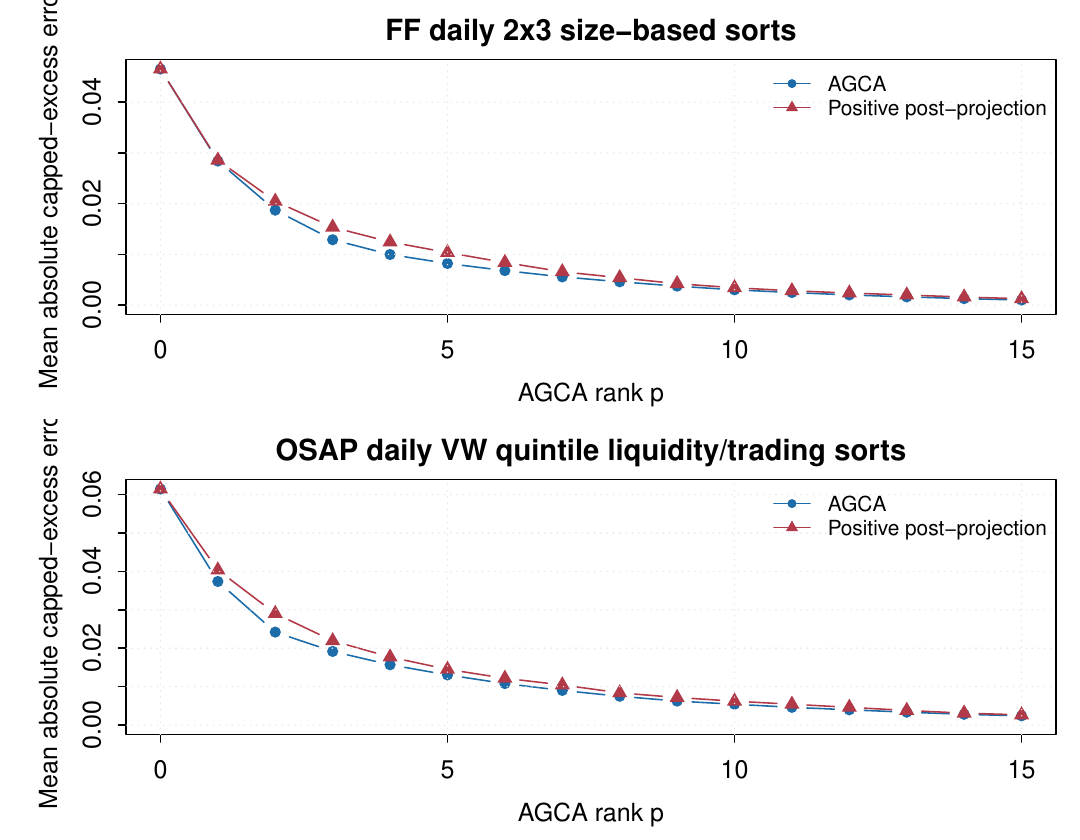}
\end{minipage}\hfill
\begin{minipage}[t]{0.49\linewidth}
\centering
\appinclude[width=\linewidth]{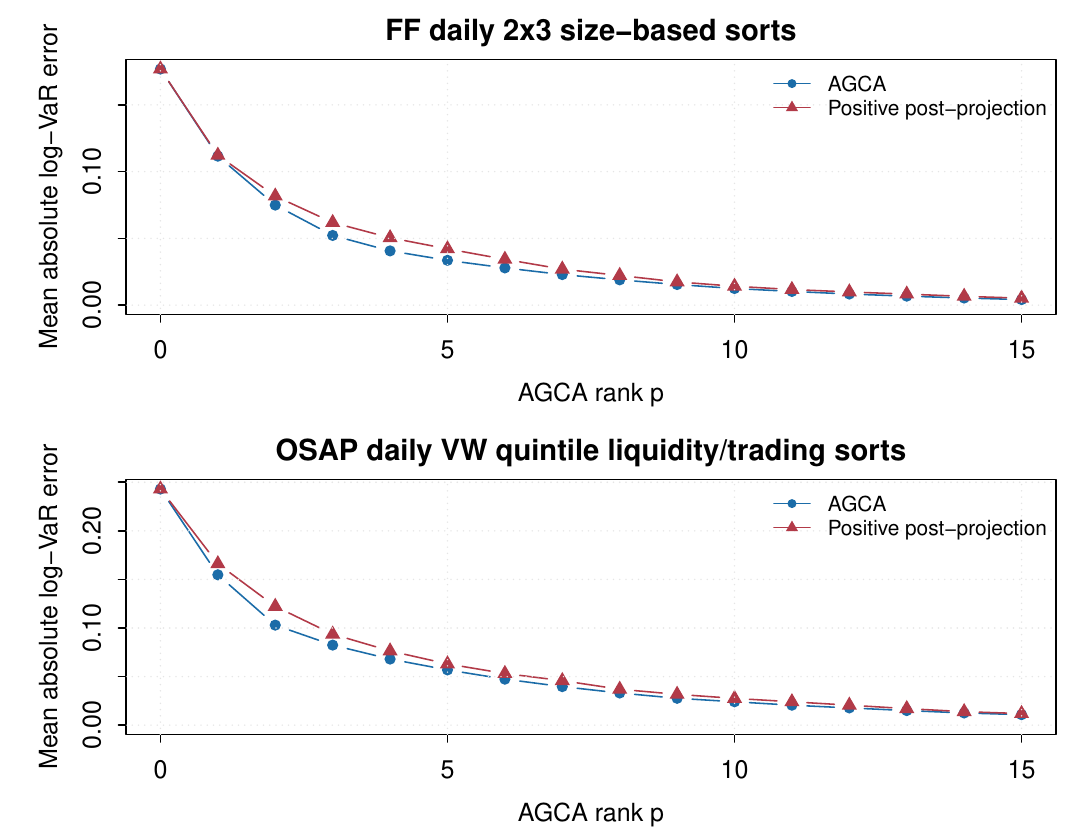}
\end{minipage}
\caption{Original-scale portfolio tail-functional errors by AGCA rank. The left display reports
absolute capped-excess error; the right display reports absolute log error in the normalized
Pareto-limit VaR tail constant. Each display includes the raw AGCA reconstruction and the
positive post-projection.}
\label{fig:app-emp-port-absolute-rank-all}
\end{figure}

The rank-level averages hide some variation across portfolio classes.
Figure~\ref{fig:app-emp-port-class-osap} reports the OSAP class-level relative errors, complementing
the Fama--French class-level display above. The class-level patterns are consistent with the main
conclusion: the errors fecline monotonically with rank for all classes, and the Fama--French panel
achieves small functional errors at lower ranks than the OSAP panel.

\begin{figure}[H]
\centering
\appinclude[width=0.95\linewidth]{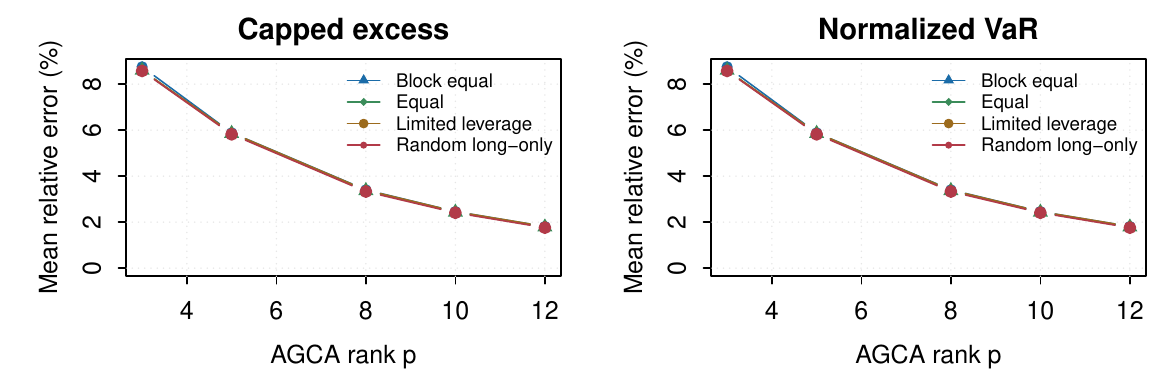}
\caption{OSAP relative portfolio-functional errors by portfolio class. The left display reports
relative capped-excess error and the right display reports relative normalized VaR error.}
\label{fig:app-emp-port-class-osap}
\end{figure}

Figure~\ref{fig:app-emp-port-functional-anchor} checks whether the same data-driven anchors used
in the spectral diagnostic materially reduce the portfolio-functional errors. The comparison
uses the same selected extreme directions and the same portfolio weights, changing only the
anchor used to fit the AGCA reconstruction. The rank effect is much larger than the anchor
effect: principal and mean anchors can move the curves slightly, but they do not materially
improve the capped-excess or normalized VaR errors relative to the canonical anchor. The
canonical anchor is therefore competitive for these scalar tail summaries while retaining the
clean loading interpretation used in the main text.

\begin{figure}[H]
\centering
\appinclude[width=0.95\linewidth]{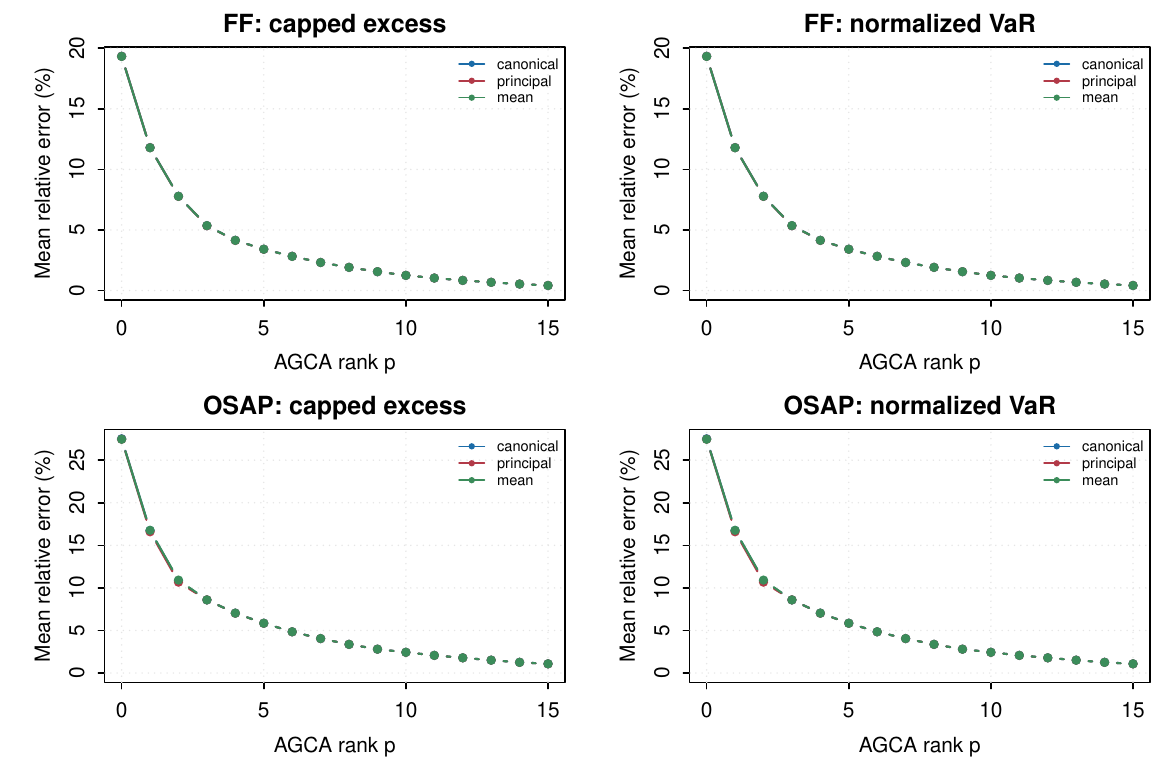}
\caption{Anchor sensitivity for relative portfolio tail-functional errors. Rows correspond to
the Fama--French and OSAP panels; columns report capped-excess and normalized VaR errors.}
\label{fig:app-emp-port-functional-anchor}
\end{figure}

\FloatBarrier

%% file: main.bbl
\begin{thebibliography}{}

\bibitem[\protect\citeauthoryear{Avella~Medina, Davis, and Samorodnitsky}{Avella~Medina et~al.}{2025}]{medina2025insights}
Avella~Medina, M., R.~A. Davis, and G.~Samorodnitsky (2025).
\newblock Insights into kernel pca with application to multivariate extremes.
\newblock {\em SIAM Journal on Mathematics of Data Science\/}~{\em 7\/}(2), 777--801.

\bibitem[\protect\citeauthoryear{Butsch and Fasen-Hartmann}{Butsch and Fasen-Hartmann}{2025}]{butsch2025estimation}
Butsch, L. and V.~Fasen-Hartmann (2025).
\newblock Estimation of the number of principal components in high-dimensional multivariate extremes.
\newblock {\em Scandinavian Journal of Statistics\/}~{\em 52\/}(4), 2270--2313.

\bibitem[\protect\citeauthoryear{Chautru}{Chautru}{2015}]{chautru2015dimension}
Chautru, E. (2015).
\newblock Dimension reduction in multivariate extreme value analysis.
\newblock {\em Electronic Journal of Statistics\/}~{\em 9\/}(1), 383--418.

\bibitem[\protect\citeauthoryear{Chen and Zimmermann}{Chen and Zimmermann}{2026}]{openSourceAssetPricing}
Chen, A.~Y. and T.~Zimmermann (2026).
\newblock {Open Source Asset Pricing}.
\newblock \url{https://www.openassetpricing.com/}.
\newblock Accessed July 7, 2026.

\bibitem[\protect\citeauthoryear{Cooley, Sabourin, and Wixson}{Cooley et~al.}{2026}]{cooley2026principal}
Cooley, D., A.~Sabourin, and T.~Wixson (2026).
\newblock Principal component analysis for multivariate extremes.
\newblock {\em arXiv preprint arXiv:2606.07213\/}.

\bibitem[\protect\citeauthoryear{Cooley and Thibaud}{Cooley and Thibaud}{2019}]{cooley2019decompositions}
Cooley, D. and E.~Thibaud (2019).
\newblock Decompositions of dependence for high-dimensional extremes.
\newblock {\em Biometrika\/}~{\em 106\/}(3), 587--604.

\bibitem[\protect\citeauthoryear{Dai and M{\"u}ller}{Dai and M{\"u}ller}{2018}]{dai2018principal}
Dai, X. and H.-G. M{\"u}ller (2018).
\newblock Principal component analysis for functional data on riemannian manifolds and spheres.
\newblock {\em The Annals of Statistics\/}~{\em 46\/}(6B), 3334--3361.

\bibitem[\protect\citeauthoryear{De~Haan and Ferreira}{De~Haan and Ferreira}{2006}]{de2006extreme}
De~Haan, L. and A.~Ferreira (2006).
\newblock {\em Extreme value theory: an introduction}.
\newblock Springer.

\bibitem[\protect\citeauthoryear{Drees}{Drees}{2025}]{drees2025asymptotic}
Drees, H. (2025).
\newblock Asymptotic behavior of principal component projections for multivariate extremes.
\newblock {\em arXiv preprint arXiv:2503.22296\/}.

\bibitem[\protect\citeauthoryear{Drees and Sabourin}{Drees and Sabourin}{2021}]{Drees2021principal}
Drees, H. and A.~Sabourin (2021).
\newblock Principal component analysis for multivariate extremes.
\newblock {\em Electronic Journal of Statistics\/}~{\em 15\/}(1), 908--943.

\bibitem[\protect\citeauthoryear{Einmahl, De~Haan, and Piterbarg}{Einmahl et~al.}{2001}]{einmahl2001nonparametric}
Einmahl, J. H.~J., L.~De~Haan, and V.~I. Piterbarg (2001).
\newblock Nonparametric estimation of the spectral measure of an extreme value distribution.
\newblock {\em The Annals of Statistics\/}~{\em 29\/}(5), 1401--1423.

\bibitem[\protect\citeauthoryear{Einmahl, De~Haan, and Sinha}{Einmahl et~al.}{1997}]{einmahl1997estimating}
Einmahl, J. H.~J., L.~De~Haan, and A.~K. Sinha (1997).
\newblock Estimating the spectral measure of an extreme value distribution.
\newblock {\em Stochastic Processes and their Applications\/}~{\em 70\/}(2), 143--171.

\bibitem[\protect\citeauthoryear{Einmahl, Krajina, and Segers}{Einmahl et~al.}{2012}]{einmahl2012m}
Einmahl, J. H.~J., A.~Krajina, and J.~Segers (2012).
\newblock An {M}-estimator for tail dependence in arbitrary dimensions.
\newblock {\em The Annals of Statistics\/}~{\em 40\/}(3), 1764--1793.

\bibitem[\protect\citeauthoryear{Einmahl and Segers}{Einmahl and Segers}{2009}]{einmahl2009maximum}
Einmahl, J. H.~J. and J.~Segers (2009).
\newblock Maximum empirical likelihood estimation of the spectral measure of an extreme-value distribution.
\newblock {\em The Annals of Statistics\/}~{\em 37\/}(5B), 2953--2989.

\bibitem[\protect\citeauthoryear{Engelke and Hitz}{Engelke and Hitz}{2020}]{engelke2020graphical}
Engelke, S. and A.~S. Hitz (2020).
\newblock Graphical models for extremes.
\newblock {\em Journal of the Royal Statistical Society: Series B (Statistical Methodology)\/}~{\em 82\/}(4), 871--932.

\bibitem[\protect\citeauthoryear{Fletcher, Lu, Pizer, and Joshi}{Fletcher et~al.}{2004}]{fletcher2004principal}
Fletcher, P.~T., C.~Lu, S.~M. Pizer, and S.~Joshi (2004).
\newblock Principal geodesic analysis for the study of nonlinear statistics of shape.
\newblock {\em IEEE Transactions on Medical Imaging\/}~{\em 23\/}(8), 995--1005.

\bibitem[\protect\citeauthoryear{Fomichov and Ivanovs}{Fomichov and Ivanovs}{2023}]{fomichov2023spherical}
Fomichov, V. and J.~Ivanovs (2023).
\newblock Spherical clustering in detection of groups of concomitant extremes.
\newblock {\em Biometrika\/}~{\em 110\/}(1), 135--153.

\bibitem[\protect\citeauthoryear{French}{French}{2026}]{kenFrenchDataLibrary}
French, K.~R. (2026).
\newblock {Kenneth R. French Data Library}.
\newblock \url{https://mba.tuck.dartmouth.edu/pages/faculty/ken.french/data_library.html}.
\newblock Accessed July 7, 2026.

\bibitem[\protect\citeauthoryear{Goix, Sabourin, and Cl{\'e}men{\c{c}}on}{Goix et~al.}{2017}]{goix2016sparsity}
Goix, N., A.~Sabourin, and S.~Cl{\'e}men{\c{c}}on (2017).
\newblock Sparse representation of multivariate extremes with applications to anomaly detection.
\newblock {\em Journal of Multivariate Analysis\/}~{\em 161}, 12--31.

\bibitem[\protect\citeauthoryear{Huckemann and Ziezold}{Huckemann and Ziezold}{2006}]{huckemann2006principal}
Huckemann, S. and H.~Ziezold (2006).
\newblock Principal component analysis for riemannian manifolds, with an application to triangular shape spaces.
\newblock {\em Advances in Applied Probability\/}~{\em 38\/}(2), 299--319.

\bibitem[\protect\citeauthoryear{Jan{\ss}en and Wan}{Jan{\ss}en and Wan}{2020}]{janssen2020kmeans}
Jan{\ss}en, A. and P.~Wan (2020).
\newblock $k$-means clustering of extremes.
\newblock {\em Electronic Journal of Statistics\/}~{\em 14\/}(1), 1211--1233.

\bibitem[\protect\citeauthoryear{Jung, Dryden, and Marron}{Jung et~al.}{2012}]{jung2012analysis}
Jung, S., I.~L. Dryden, and J.~S. Marron (2012).
\newblock Analysis of principal nested spheres.
\newblock {\em Biometrika\/}~{\em 99\/}(3), 551--568.

\bibitem[\protect\citeauthoryear{Lhaut, Rootz{\'e}n, and Segers}{Lhaut et~al.}{2026}]{lhaut2026simulation}
Lhaut, S., H.~Rootz{\'e}n, and J.~Segers (2026).
\newblock Simulation of multivariate extremes: A wasserstein--aitchison gan approach.
\newblock {\em Extremes\/}~{\em 29\/}(2), 157--194.

\bibitem[\protect\citeauthoryear{Meyer and Wintenberger}{Meyer and Wintenberger}{2021}]{meyer2021sparse}
Meyer, N. and O.~Wintenberger (2021).
\newblock Sparse regular variation.
\newblock {\em Advances in Applied Probability\/}~{\em 53\/}(4), 1115--1148.

\bibitem[\protect\citeauthoryear{Mourahib, Kiriliouk, and Segers}{Mourahib et~al.}{2025}]{mourahib2025multivariate}
Mourahib, A., A.~Kiriliouk, and J.~Segers (2025).
\newblock Multivariate generalized pareto distributions along extreme directions.
\newblock {\em Extremes\/}~{\em 28\/}(2), 239--272.

\bibitem[\protect\citeauthoryear{Resnick}{Resnick}{2007}]{resnick2007heavy}
Resnick, S.~I. (2007).
\newblock {\em Heavy-tail phenomena: probabilistic and statistical modeling}.
\newblock Springer.

\bibitem[\protect\citeauthoryear{Rootz{\'e}n, Segers, and Wadsworth}{Rootz{\'e}n et~al.}{2018}]{rootzen2018multivariate}
Rootz{\'e}n, H., J.~Segers, and J.~L. Wadsworth (2018).
\newblock Multivariate generalized pareto distributions: Parametrizations, representations, and properties.
\newblock {\em Journal of Multivariate Analysis\/}~{\em 165}, 117--131.

\bibitem[\protect\citeauthoryear{Rootz{\'e}n and Tajvidi}{Rootz{\'e}n and Tajvidi}{2006}]{rootzen2006multivariate}
Rootz{\'e}n, H. and N.~Tajvidi (2006).
\newblock Multivariate generalized pareto distributions.
\newblock {\em Bernoulli\/}~{\em 12\/}(5), 917--930.

\bibitem[\protect\citeauthoryear{Sommer, Lauze, and Nielsen}{Sommer et~al.}{2014}]{sommer2014optimization}
Sommer, S., F.~Lauze, and M.~Nielsen (2014).
\newblock Optimization over geodesics for exact principal geodesic analysis.
\newblock {\em Advances in Computational Mathematics\/}~{\em 40\/}(2), 283--313.

\bibitem[\protect\citeauthoryear{Tabaghi, Khanzadeh, Wang, and Mirarab}{Tabaghi et~al.}{2024}]{tabaghi2024principal}
Tabaghi, P., M.~Khanzadeh, Y.~Wang, and S.~Mirarab (2024).
\newblock Principal component analysis in space forms.
\newblock {\em IEEE Transactions on Signal Processing\/}~{\em 72}, 4428--4443.

\bibitem[\protect\citeauthoryear{Wan}{Wan}{2026}]{wan2026characterizing}
Wan, P. (2026).
\newblock Characterizing extremal dependence on a hyperplane.
\newblock {\em Biometrika\/}~{\em 113\/}(2), asag015.

\end{thebibliography}


\begin{thebibliography}{}

\bibitem[\protect\citeauthoryear{Bingham, Goldie, and Teugels}{Bingham et~al.}{1987}]{bingham1987regular}
Bingham, N.~H., C.~M. Goldie, and J.~L. Teugels (1987).
\newblock {\em Regular Variation}.
\newblock Cambridge University Press.

\bibitem[\protect\citeauthoryear{Chen and Zimmermann}{Chen and Zimmermann}{2026}]{openSourceAssetPricing}
Chen, A.~Y. and T.~Zimmermann (2026).
\newblock {Open Source Asset Pricing}.
\newblock \url{https://www.openassetpricing.com/}.
\newblock Accessed July 7, 2026.

\bibitem[\protect\citeauthoryear{De~Haan and Ferreira}{De~Haan and Ferreira}{2006}]{de2006extreme}
De~Haan, L. and A.~Ferreira (2006).
\newblock {\em Extreme value theory: an introduction}.
\newblock Springer.

\bibitem[\protect\citeauthoryear{Einmahl, De~Haan, and Piterbarg}{Einmahl et~al.}{2001}]{einmahl2001nonparametric}
Einmahl, J. H.~J., L.~De~Haan, and V.~I. Piterbarg (2001).
\newblock Nonparametric estimation of the spectral measure of an extreme value distribution.
\newblock {\em The Annals of Statistics\/}~{\em 29\/}(5), 1401--1423.

\bibitem[\protect\citeauthoryear{Einmahl, De~Haan, and Sinha}{Einmahl et~al.}{1997}]{einmahl1997estimating}
Einmahl, J. H.~J., L.~De~Haan, and A.~K. Sinha (1997).
\newblock Estimating the spectral measure of an extreme value distribution.
\newblock {\em Stochastic Processes and their Applications\/}~{\em 70\/}(2), 143--171.

\bibitem[\protect\citeauthoryear{Einmahl, Krajina, and Segers}{Einmahl et~al.}{2012}]{einmahl2012m}
Einmahl, J. H.~J., A.~Krajina, and J.~Segers (2012).
\newblock An {M}-estimator for tail dependence in arbitrary dimensions.
\newblock {\em The Annals of Statistics\/}~{\em 40\/}(3), 1764--1793.

\bibitem[\protect\citeauthoryear{Einmahl and Segers}{Einmahl and Segers}{2009}]{einmahl2009maximum}
Einmahl, J. H.~J. and J.~Segers (2009).
\newblock Maximum empirical likelihood estimation of the spectral measure of an extreme-value distribution.
\newblock {\em The Annals of Statistics\/}~{\em 37\/}(5B), 2953--2989.

\bibitem[\protect\citeauthoryear{French}{French}{2026}]{kenFrenchDataLibrary}
French, K.~R. (2026).
\newblock {Kenneth R. French Data Library}.
\newblock \url{https://mba.tuck.dartmouth.edu/pages/faculty/ken.french/data_library.html}.
\newblock Accessed July 7, 2026.

\bibitem[\protect\citeauthoryear{Genest and Segers}{Genest and Segers}{2009}]{genest2009rank}
Genest, C. and J.~Segers (2009).
\newblock Rank-based inference for bivariate extreme-value copulas.
\newblock {\em The Annals of Statistics\/}~{\em 37\/}(5B), 2990--3022.

\bibitem[\protect\citeauthoryear{Horn and Johnson}{Horn and Johnson}{2013}]{hornJohnson2013matrix}
Horn, R.~A. and C.~R. Johnson (2013).
\newblock {\em Matrix Analysis\/} (2 ed.).
\newblock Cambridge University Press.

\bibitem[\protect\citeauthoryear{Shorack and Wellner}{Shorack and Wellner}{1986}]{shorack1986empirical}
Shorack, G.~R. and J.~A. Wellner (1986).
\newblock {\em Empirical Processes with Applications to Statistics}.
\newblock Wiley.

\bibitem[\protect\citeauthoryear{Wellner}{Wellner}{1978}]{wellner1978limit}
Wellner, J.~A. (1978).
\newblock Limit theorems for the ratio of the empirical distribution function to the true distribution function.
\newblock {\em Zeitschrift f{\"u}r Wahrscheinlichkeitstheorie und verwandte Gebiete\/}~{\em 45\/}(1), 73--88.

\end{thebibliography}
